\documentclass{sig-alternate-05-2015}

\usepackage{enumitem}\setitemize{leftmargin=12pt}\setenumerate{leftmargin=12pt}

\usepackage{times}
\usepackage{theorem} 
\usepackage{stmaryrd} 
\usepackage{epsfig} 
\usepackage{misc}
\usepackage{amsmath,amssymb}
\usepackage{bbm}
\usepackage{bm}

\usepackage{boxedminipage}
\usepackage{pifont}
\usepackage{graphicx}
\usepackage{xspace} 
\usepackage{color}
\usepackage{url}

\newcommand{\coNP}{\textnormal{\sc coNP}\xspace}

\newtheorem{aremark}{Remark}
\newtheorem{adefinition}{Definition}
\newtheorem{atheorem}{Theorem}
\newtheorem{aexample}{Example}
\newtheorem{aclaim}{Claim}
\newtheorem{alemma}{Lemma}
\newtheorem{aproposition}{Proposition}
\newtheorem{acorollary}{Corollary}

\newenvironment{definition}{\vspace{-3mm}\begin{adefinition}}{\end{adefinition}\vspace{-1mm}}
\newenvironment{theorem}{\vspace{-3mm}\begin{atheorem}}{\end{atheorem}\vspace{-1mm}}
\newenvironment{example}{\vspace{-3mm}\begin{aexample}}{\end{aexample}\vspace{-1mm}}

\newenvironment{lemma}{\vspace{-3mm}\begin{alemma}}{\end{alemma}\vspace{-1mm}}

\newenvironment{corollary}{\vspace{-3mm}\begin{acorollary}}{\end{acorollary}\vspace{-1mm}}

\newcommand{\PTIME}{\PTime}

\newcommand{\SAT}{\textsc{SAT}}
\newcommand{\midsqcup}{\mathop{\mathop{\mbox{\midmathxx\symbol{116}}}}\limits}
\newcommand{\NPTime}{\textnormal{\sc NP}\xspace}

\newcommand{\isdef}{\mathrel{\mathop:}=}
\newcommand{\nat}{\mathbb{N}}
\newcommand{\set}[1]{\{{#1}\}}
\newcommand{\card}[1]{\lvert{#1}\rvert}
\newcommand{\union}{\cup}
\newcommand{\isect}{\cap}

\newcommand{\len}[1]{\lvert{#1}\rvert}

\newcommand{\uGF}{\ensuremath{\text{uGF}(=)}}
\newcommand{\uGC}{\ensuremath{\text{uGC}_2(=)}}
\newcommand{\bag}{\text{bag}}
\newcommand{\dist}{\mathsf{dist}}
\newcommand{\tp}{\mathsf{tp}}

\newcommand{\openGF}{\text{openGF}}
\newcommand{\openGC}{\ensuremath{\text{openGC}_2}}

\newcommand{\RunFit}{\textsc{RF}}

\begin{document}

\CopyrightYear{2017}
\setcopyright{acmlicensed}
\conferenceinfo{PODS'17,}{May 14 -- 19, 2017, Chicago, IL, USA.}
\isbn{978-1-4503-4198-1/17/05}\acmPrice{\$15.00}
\doi{http://dx.doi.org/10.1145/3034786.3056108}

\title{Dichotomies in Ontology-Mediated Querying with the Guarded Fragment}

\numberofauthors{4}
\addtolength{\auwidth}{2cm}

\author{
 \alignauthor
 Andr{\'e} Hernich\\
        \affaddr{University of Liverpool}\\
        \affaddr{United Kingdom}\\
        \email{andre.hernich@liverpool.ac.uk}
 \alignauthor
 Carsten Lutz\\
        \affaddr{University of Bremen}\\
        \affaddr{Germany}\\
        \email{clu@informatik.uni-bremen.de}
\and
\alignauthor
 Fabio Papacchini\\
        \affaddr{University of Liverpool}\\
        \affaddr{United Kingdom}\\
        \email{fabio.papacchini@liverpool.ac.uk}
\alignauthor Frank Wolter\\
        \affaddr{University of Liverpool}\\
        \affaddr{United Kingdom}\\
        \email{wolter@liverpool.ac.uk}
}

\maketitle   
\begin{abstract}
  We study the complexity of ontology-mediated querying when
  ontologies are formulated in the guarded fragment of first-order
  logic (GF). Our general aim is to classify the data complexity on
  the level of ontologies where query evaluation w.r.t.\ an ontology
  \Omc is considered to be in \PTime if all (unions of conjunctive)
  queries can be evaluated in \PTime w.r.t.~\Omc and {\sc coNP}-hard
  if at least one query is {\sc coNP}-hard w.r.t.~\Omc. We identify
  several large and relevant fragments of GF that enjoy a dichotomy
  between \PTime and {\sc coNP}, some of them additionally admitting a
  form of counting.  In fact, almost all ontologies in the BioPortal
  repository fall into these fragments or can easily be rewritten to
  do so. We then establish a variation of Ladner's Theorem on the
  existence of {\sc NP}-intermediate problems and use this result to
  show that for other fragments, there is provably no such dichotomy.
  Again for other fragments (such as full GF), establishing a
  dichotomy implies the Feder-Vardi conjecture on the complexity
  of constraint satisfaction problems. We also link these results to
  Datalog-rewritability and study the decidability of whether a given
  ontology enjoys \PTime query evaluation, presenting both positive
  and negative results.
\end{abstract}

\keywords{Ontology-Based Data Access; Query Answering; Dichotomies}

\section{Introduction}

In \emph{Ontology-Mediated Querying}, incomplete data is enriched with
an ontology that provides domain knowledge, enabling more complete
answers to queries \cite{DBLP:journals/jods/PoggiLCGLR08,DBLP:conf/rweb/BienvenuO15,DBLP:conf/rweb/KontchakovZ14}.
This paradigm has recently received a lot of interest, a significant fraction of the
research being concerned with the (data) complexity of querying
\cite{DBLP:journals/jar/OrtizCE08,DBLP:journals/ai/CalvaneseGLLR13} and, closely related, with the rewritability of
ontology-mediated queries into more conventional database query
languages \cite{Romans,DBLP:journals/ai/GottlobKKPSZ14,DBLP:journals/tods/GottlobOP14,DBLP:journals/ai/KaminskiNG16,DBLP:conf/ijcai/GottlobMP15}.
A particular emphasis has been put on designing ontology languages that result in \PTime data complexity,
and in delineating these from the {\sc coNP}-hard cases. This question
and related ones have given rise to a considerable array of ontology
languages, including many description logics (DLs) \cite{dllite-jair09,DBLP:conf/rweb/Krotzsch12}
and a growing number of classes of tuple-generating dependencies
(TGDs), also known as Datalog$^\pm$ and as existential rules
\cite{DBLP:journals/ai/CaliGP12,DBLP:conf/rweb/MugnierT14}. A general and uniform framework is provided by the
\emph{guarded fragment (GF)} of first-order logic and extensions
thereof, which subsume many of the mentioned ontology languages
\cite{DBLP:journals/corr/BaranyGO13,DBLP:journals/jacm/BaranyCS15}.

In practical applications, ontologies often need to use language
features that are only available in computationally expensive ontology
languages, but do so in a way such that one may hope for hardness to
be avoided. This observation has led to a more fine-grained study of
data complexity than on the level of ontology languages, initiated in
\cite{KR12-csp}, where the aim is to classify the complexity of individual
ontologies while quantifying over the actual query: query evaluation
w.r.t.\ an ontology \Omc is in \PTime if every CQ can be evaluated in
\PTime w.r.t.\ \Omc and it is {\sc coNP}-hard if there is at least one
CQ that is {\sc coNP}-hard to evaluate w.r.t.~\Omc. In this way, one
can identify tractable ontologies within ontology languages that are,
in general, computationally hard. Note that an even more fine-grained
approach is taken in \cite{DBLP:journals/tods/BienvenuCLW14}, where one aims to classify the
complexity of each pair $(\Omc,q)$ with \Omc an ontology and $q$ an
actual query. Both approaches are reasonable, the first one being
preferable when the queries to be answered are not fixed at the design
time of the ontology; this is actually often the case because
ontologies are typically viewed as general purpose artifacts to be used
in more than a single application. In this paper, we follow the former
approach.

\begin{figure*}[th]
  \begin{boxedminipage}{\textwidth}
  \centering
 \includegraphics{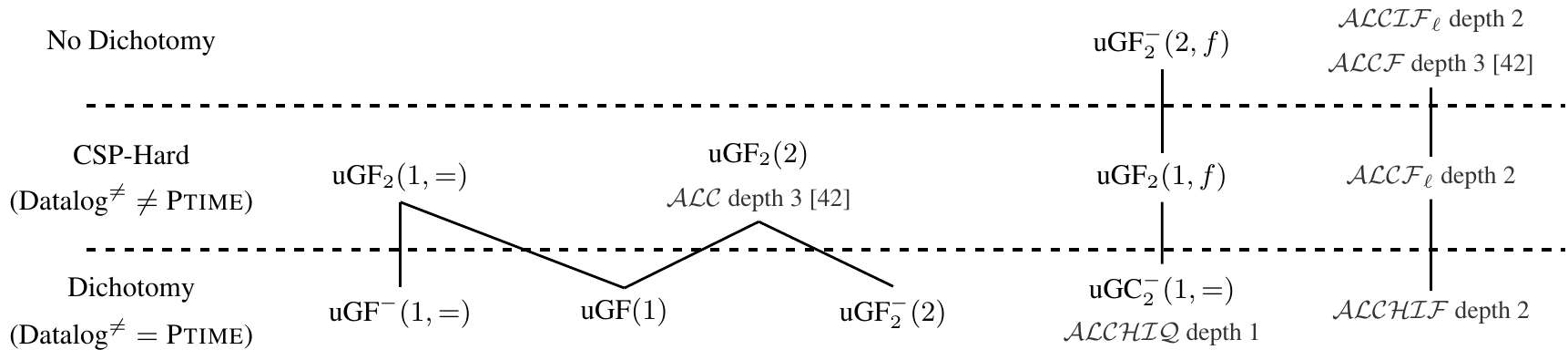}
  \end{boxedminipage}
  \caption{Summary of the results---Number in brackets indicates
    depth, $f$ presence of partial functions, $\cdot_2$ restriction to
    two variables, $\cdot^-$ restricts outermost guards to be
    equality, \Fmc globally function roles, $\Fmc_\ell$ concepts $(\leq 1 \, R)$.}
\label{fig:results}
\end{figure*}
The main aim of this paper is to \emph{identify fragments of GF (and
  of extensions of GF with different forms of counting) that result in
  a dichotomy between {\sc PTime} and {\sc coNP} when used as an
  ontology language and that cover as many real-world ontologies as
  possible, considering conjunctive queries (CQs) and unions thereof
  (UCQs) as the actual query language.} We also aim to provide insight
into which fragments of GF (with and without counting) do \emph{not}
admit such a dichotomy, to understand the relation between {\sc PTime}
data complexity and rewritability into Datalog (with inequality in
rule bodies, in case we start from GF with equality or counting), and
to clarify whether it is decidable whether a given ontology has {\sc
  PTime} data complexity.  Note that we concentrate on
\emph{fragments} of GF because for the full guarded fragment, proving
a dichotomy between {\sc PTime} and {\sc coNP} implies the
long-standing Feder-Vardi conjecture on constraint satisfaction
problems \cite{FederVardi} which indicates that it is very difficult
to obtain (if it holds at all).  In particular, we concentrate on the
fragment of GF that is invariant under disjoint unions, which we call
uGF, and on fragments thereof and their extension with forms of
counting.  Invariance under disjoint unions is a fairly mild
restriction that is shared by many relevant ontology languages, and it
admits a very natural syntactic characterization.


 Our results are summarized in Figure~\ref{fig:results}. We first
explain the fragments shown in the figure and then survey the obtained
results.  A uGF ontology is a set of sentences of the form $\forall
\vec{x} (R(\vec{x}) \rightarrow \varphi(\vec{x}))$ where $R(\vec{x})$
is a guard (possibly equality) and $\varphi(\vec{x})$ is a GF formula
that does not contain any sentences as subformulas and in which
equality is not used as a guard. The \emph{depth} of such a sentence
is the quantifier depth of $\varphi(\vec{x})$ (and thus the outermost
universal quantifier is not counted). A main parameter that we vary is
the depth, which is typically very small in real world ontologies. In
Figure~\ref{fig:results}, the depth is the first parameter displayed
in brackets. As usual, the subscript $\cdot_2$ indicates the
restriction to two variables while a superscript $\cdot^-$ means that
the guard $R(\vec{x})$ in the outermost universal quantifier can only
be equality, $=$ means that equality is allowed (in non-guard
positions), $f$ indicates the ability to declare binary relation
symbols to be interpreted as partial functions, and GC$_2$ denotes the
two variable guarded fragment extended with counting quantifiers, see
\cite{DBLP:conf/jelia/Kazakov04,DBLP:journals/logcom/Pratt-Hartmann07}.
While guarded fragments are displayed in black,
description logics (DLs) are shown in grey and smaller font
size. We use standard DL names except that `\Fmc' denotes globally
functional roles while `$\Fmc_\ell$' refers to counting concepts of
the form $(\leq 1 \, R)$. We do not explain DL names
here, but refer to the standard literature \cite{handbook}.

The bottommost part of Figure~\ref{fig:results} displays fragments for
which there is a dichotomy between {\sc PTime} and {\sc coNP}, the
middle part shows fragments for which such a dichotomy implies the
Feder-Vardi conjecture (from now on called CSP-hardness), and the
topmost part is for fragments that provably have no dichotomy (unless
$\text{\sc PTime} = \text{\sc NP}$).  The vertical lines indicate that
the linked results are closely related, often indicating a fundamental
difficulty in further generalizing an upper bound. For example,
uGF$^-(1,=)$ enjoys a dichotomy while uGF$_2(1,=)$ is CSP-hard, which
demonstrates that generalizing the former result by dropping the
restriction that the outermost quantifier has to be equality
(indicated by $\cdot^-$) is very challenging (if it is possible at
all).\footnote{A tentative proof of the Feder-Vardi conjecture has
  very recently been announced in
  \cite{DBLP:journals/corr/RafieyKF17}, along with an invitation to
  the research community to verify its validity.}  Our positive
results are thus optimal in many ways. All results hold both when CQs
and when UCQs are used as the actual query; in this context, it is
interesting to note that there is a GF ontology (which is not an uGF
ontology) for which CQ answering is in {\sc PTime} while UCQ-answering
is {\sc coNP}-hard. In the cases which enjoy a dichotomy, we also show
that {\sc PTime} query evaluation coincides with rewritability into
Datalog (with inequality in the rule bodies if we start from a
fragment with equality or counting). In contrast, for all fragments
that are CSP-hard or have no dichotomy, these two properties do
provably not coincide. This is of course independent of whether or not
the Feder-Vardi conjecture holds.

For $\mathcal{ALCHIQ}$ ontologies of depth~1, we also show that it is
decidable and \ExpTime-complete whether a given ontology admits {\sc
  PTime} query evaluation (equivalently: rewritability into
Datalog$^{\neq}$). For uGC$^-_2(1,=)$, we show a {\sc NExpTime} upper
bound. For $\mathcal{ALC}$ ontologies of depth~2, we establish {\sc
  NExpTime}-hardness. The proof indicates that more sophisticated
techniques are needed to establish decidability, if the problem is
decidable at all (which we leave open).

To understand the practical relevance of our results, we have analyzed
411 ontologies from the BioPortal repository \cite{whetzel2011bioportal}. After
removing all constructors that do not fall within $\mathcal{ALCHIF}$, an impressive
405 ontologies turned out to have depth~2 and thus belong to a fragment with dichotomy (sometimes modulo an
easy complexity-preserving rewriting).
For $\mathcal{ALCHIQ}$, still 385 ontologies had depth~1 and so belonged to a fragment with dichotomy.
%
As a concrete and simple
example, consider the two uGC$^-_2(1)$-ontologies
$$
\begin{array}{r@{\;}c@{\;}l}
  \Omc_1 &=& \{ \forall x \, (\mn{Hand}(x) \rightarrow
  \exists^{{=}5}  y\, \mn{hasFinger}(x,y)) \} \\[1mm]
  \Omc_2 &=& \{ \forall x \, (\mn{Hand}(x) \rightarrow
  \exists  y\, (\mn{hasFinger}(x,y) \wedge \mn{Thumb}(y))) \}
\end{array}
$$
which both enjoy {\sc PTime} query evaluation (and thus rewritability
into Datalog$^{\neq}$), but where query evaluation w.r.t.\ the union
$\Omc_1 \cup \Omc_2$ is {\sc coNP}-hard. Note that such subtle
differences cannot be captured when data complexity is studied on the
level of ontology languages, at least when basic compositionality
conditions are desired. 

We briefly highlight some of the techniques used to establish our
results.  An important role is played by the notions of
materializability and unraveling tolerance of an
ontology~\Omc. Materializability means that for every instance~\Dmf,
there is a universal model of \Dmf and \Omc, defined in terms of query
answers rather than in terms of homomorphisms (which, as we show, need
not coincide in our context). Unraveling tolerance means that the
ontology cannot distinguish between an instance and its unraveling
into a structure of bounded treewidth. While non-materializability of
\Omc implies that query evaluation w.r.t.\ \Omc is {\sc coNP}-hard,
unraveling tolerance of \Omc implies that query evaluation w.r.t.\
\Omc is in {\sc \PTime} (in fact, even rewritable into Datalog).
To establish dichotomies, we prove for the relevant fragments that
materializability implies unraveling tolerance which, depending on the fragment, can be
technically rather subtle.
To prove CSP-hardness or non-dichotomies, very informally speaking, we
need to express properties in the ontology that a (positive
existential) query cannot `see'. This is often very subtle and can
often be achieved only partially. While the latter is not a major
problem for CSP-hardness (where we need to deal with CSPs that `admit
precoloring' and are known to behave essentially in the same way as
traditional CSPs), it poses serious challenges when proving
non-dichotomy. To tackle this problem, we establish a variation of
Ladner's theorem on {\sc NP}-intermediate problems such that instead
of the word problem for {\sc NP} Turing machines, it speaks about the
run fitting problem, which is to decide whether a given partially
described run of a Turing machine (which corresponds to a precoloring
in the CSP case) can be extended to a full run that is accepting. Also
our proofs of decidability of whether an ontology admits {\sc PTime}
query evaluation are rather subtle and technical, involving e.g.\
mosaic techniques.

Due to space constraints, throughout the paper we defer proof details
to the appendix.

\medskip
\noindent {\bf Related Work}. Ontology-mediated querying has first
been considered in \cite{DBLP:journals/ai/LevyR98a,DBLP:conf/pods/CalvaneseGL98}; other important papers
include \cite{Romans,DBLP:journals/jair/CaliGK13,DBLP:journals/corr/BaranyGO13}. It is a form of reasoning
under integrity constraints, a traditional topic in database theory,
see e.g.\ \cite{DBLP:journals/jacm/BeeriV84,DBLP:journals/tods/BeeriB79} and references therein, and it is also
related to deductive databases, see e.g.\ the monograph
\cite{Minkerbook}. Moreover, ontology-mediated querying has drawn
inspiration from query answering under views \cite{LICS00,PODS03}.
In recent years, there has been significant interest in complete
classification of the complexity of hard querying problems. In the
context of ontology-mediated querying, relevant references include
\cite{KR12-csp,DBLP:journals/tods/BienvenuCLW14,DBLP:conf/ijcai/LutzSW13}. In fact, this paper closes a
number of open problems from \cite{KR12-csp} such as that \ALCI
ontologies of depth two enjoy a dichotomy and that materializability
(and thus \PTime complexity and Datalog-rewritability) is decidable in
many relevant cases.  Other areas of database theory where complete
complexity classifications are sought include consistent query
answering
\cite{DBLP:conf/pods/KoutrisW15,DBLP:conf/icdt/KoutrisS14,DBLP:journals/tocl/Fontaine15,DBLP:conf/pods/FaginKK15},
probabilistic databases \cite{DBLP:series/synthesis/2011Suciu}, and
deletion propagation
\cite{DBLP:conf/pods/Kimelfeld12,DBLP:journals/pvldb/FreireGIM15}.

\section{Preliminaries}
\label{sect:prelims}
We assume an infinite set $\Delta_{D}$ of data constants,
an infinite set $\Delta_{N}$ of labeled nulls disjoint from $\Delta_{D}$,
and a set $\Sigma$ of relation symbols
containing infinitely many relation symbols of any arity $\geq 1$.
A \emph{(database) instance} \Dmf is a non-empty set of \emph{facts}
$R(a_1,\dotsc,a_k)$, where $R \in \Sigma$, $k$ is the arity of $R$,
and $a_1,\dotsc,a_k\in \Delta_{D}$. We generally assume that instances
are finite, unless otherwise specified. An \emph{interpretation}
$\Amf$ is a non-empty set of \emph{atoms} $R(a_1,\dotsc,a_k)$, where
$R \in \Sigma$, $k$ is the arity of $R$, and $a_1,\dotsc,a_k\in
\Delta_{D}\cup \Delta_{N}$.  We use $\text{sig}(\Amf)$ and
$\text{dom}(\Amf)$ to denote the set of relation symbols and,
respectively, constants and labelled nulls in $\Amf$.  We always
assume that $\text{sig}(\Amf)$ is finite while $\text{dom}(\Amf)$ can
be infinite.  Whenever convenient, 
interpretations $\Amf$ are presented in the form
$(A,(R^{\mathfrak{A}})_{R\in \text{sig}(\Amf)})$ where $A =
\text{dom}(\Amf)$ 
and $R^{\mathfrak{A}}$ is a $k$-ary relation on $A$ for each $R\in
\text{sig}(\Amf)$ of arity $k$. An interpretation \Amf is a \emph{model} of
an instance \Dmf, written $\Amf \models \Dmf$, if $\Dmf \subseteq
\Amf$. We thus make a strong open world assumption (interpretations
can make true additional facts and contain additional constants and
nulls) and also assume \emph{standard names} (every constant in \Dmf
is interpreted as itself in \Amf). Note that every instance is also an
interpretation.

%

Assume $\Amf$ and $\Bmf$ are interpretations. A \emph{homomorphism}
$h$ from $\Amf$ to $\Bmf$ is a mapping from $\text{dom}(\Amf)$ to
$\text{dom}(\Bmf)$ such that $R(a_{1},\ldots,a_{k})\in \Amf$ implies
$R(h(a_{1}),\ldots,h(a_{k}))\in \Bmf$ for all $a_{1},\ldots,a_{k}\in
\text{dom}(\Amf)$ and $R\in \Sigma$ of arity $k$. We say that $h$
\emph{preserves a set $D$} of constants and labelled nulls if $h(a)=a$
for all $a\in D$ and that $h$ is an \emph{isomorphic embedding} if it
is injective and $R(h(a_{1}),\ldots,h(a_{k}))\in \Bmf$ entails
$R(a_{1},\ldots,a_{k})\in \Amf$. An interpretation $\Amf \subseteq
\Bmf$ is a \emph{subinterpretation} of $\Bmf$ if
$R(a_{1},\ldots,a_{k})\in \Bmf$ and $a_{1},\ldots,a_{k}\in
\text{dom}(\Amf)$ implies $R(a_{1},\ldots,a_{k})\in \Amf$; if
$\mn{dom}(\Amf)=A$, we denote $\Amf$ by $\Bmf_{|A}$ and call it
\emph{the subinterpretation of $\Bmf$ induced by $A$}.

\emph{Conjunctive queries (CQs) $q$ of arity $k$} take the form
$q(\vec{x}) \leftarrow \phi$, where $\vec{x}=x_{1},\ldots,x_{k}$ is
the tuple of \emph{answer variables} of $q$, and $\phi$ is a
conjunction of \emph{atomic formulas} $R(y_{1},\ldots,y_{n})$ with
$R\in \Sigma$ of arity $n$ and $y_{1},\ldots,y_{n}$ variables.  As
usual, all variables in $\vec{x}$ must occur in some atom of $\phi$.
Any CQ $q(\vec{x}) \leftarrow \phi$ can be regarded as an instance
$\Dmf_q$, often called the \emph{canonical database of $q$}, in which
each variable $y$ of $\phi$ is represented by a unique data
constant~$a_y$, and that for each atom $R(y_1,\dotsc,y_k)$ in $\phi$
contains the atom $R(a_{y_1},\dotsc,a_{y_k})$.  A tuple $\vec{a}=
(a_{1},\ldots,a_{k})$ of constants is an \emph{answer to
  $q(x_{1},\ldots,x_{k})$ in $\Amf$}, in symbols $\Amf\models
q(\vec{a})$, if there is a homomorphism $h$ from $\Dmf_q$ to $\Amf$
with $h(a_{x_{1}},\ldots,a_{x_{k}})=\vec{a}$.
%
%
A \emph{union of conjunctive queries (UCQ)} $q$ takes the form
$q_1(\vec{x}),\dotsc,q_n(\vec{x})$,
where each $q_i(\vec{x})$ is a CQ.
The $q_i$ are called \emph{disjuncts} of $q$.
A tuple $\vec{a}$ of constants is an \emph{answer} to $q$ in $\Amf$,
denoted by $\Amf \models q(\vec{a})$,
if $\vec{a}$ is an answer to some disjunct of $q$ in $\Amf$.

We now introduce the fundamentals of ontology-mediated querying.  An
\emph{ontology language}~\Lmc is a set of first-order sentences over
signature $\Sigma$ (that is, function symbols are not allowed) and an
\emph{\Lmc-ontology} \Omc is a finite set of sentences from $\Lmc$. We
introduce various concrete ontology languages throughout the paper,
including fragments of the guarded fragment and descriptions logics.
An interpretation $\Amf$ is a \emph{model of an ontology} $\Omc$, in
symbols $\Amf\models \Omc$, if it satisfies all its sentences.
An instance $\Dmf$ is \emph{consistent~w.r.t.~$\Omc$}
if there is a model of $\Dmf$ and $\Omc$.

An \emph{ontology-mediated query (OMQ)} is a pair $(\Omc,q)$, where
\Omc is an ontology and $q$ a UCQ.  The semantics of an
ontology-mediated query is given in terms of \emph{certain answers},
defined next. Assume that $q$ has arity $k$ and $\Dmf$ is an
instance. Then a tuple $\vec{a}$ of length $k$ in $\text{dom}(\Dmf)$
is a \emph{certain answer to $q$ on an instance $\Dmf$ given \Omc}, in
symbols $\Omc,\Dmf\models q(\vec{a})$, if $\Amf\models q(\vec{a})$ for
all models $\Amf$ of $\Dmf$ and $\Omc$.  The \emph{query evaluation
  problem for an OMQ $(\Omc,q(\vec{x}))$} is to decide, given an
instance $\Dmf$ and a tuple $\vec{a}$ in $\Dmf$,
whether $\Omc,\Dmf\models q(\vec{a})$.

We use standard notation for Datalog programs (a brief introduction is given in the appendix).
An OMQ $(\Omc,q(\vec{x}))$ is called \emph{Datalog-rewritable}
if there is a Datalog program $\Pi$ such that for all instances \Dmf
and $\vec{a} \in \mn{dom}(\Dmf)$, 
$\Omc,\Dmf\models q(\vec{a})$ iff $\Dmf\models
\Pi(\vec{a})$. \emph{Datalog$^{\neq}$-rewritability} is defined
accordingly, but allows the use of inequality in the body of Datalog
rules. We are mainly interested in the following properties of ontologies.
\begin{definition}\label{def:main}
  Let $\mathcal{O}$ be an ontology and $\Qmc$ a class of queries.
  Then
\begin{itemize}
\item \emph{\Qmc-evaluation w.r.t.\ $\Omc$ is in {\sc PTime}} if for
  every $q \in \Qmc$, the query evaluation problem for ($\Omc,q)$ is
  in {\sc PTime}.

\item \emph{\Qmc-evaluation w.r.t.\ $\Omc$ is Datalog-rewritable}
  (resp.\ \emph{Datalog$^{\neq}$-rewritable}) if for every $q \in
  \Qmc$, the query evaluation problem for ($\Omc,q)$ is
  Datalog-rewritable (resp.\ Datalog$^{\neq}$-rewritable).

\item \emph{\Qmc-evaluation w.r.t.\ $\Omc$ is {\sc coNP}-hard} if there
  is a $q \in \Qmc$ such that the query evaluation problem for $(\Omc,q)$
  is {\sc coNP}-hard.
\end{itemize}
\end{definition}

\subsection{Ontology Languages}
As ontology languages, we consider fragments of the guarded fragment
(GF) of FO, the two-variable guarded fragment of FO with counting, and
DLs.  Recall that GF formulas~\cite{ANvB98} are obtained by starting from atomic
formulas $R(\vec{x})$ over $\Sigma$ and equalities $x=y$ and then
using the boolean connectives and \emph{guarded quantifiers} of the
form
$$
\forall \vec{y}(\alpha(\vec{x},\vec{y})\rightarrow \varphi(\vec{x},\vec{y}))
,\quad
\exists \vec{y}(\alpha(\vec{x},\vec{y})\wedge \varphi(\vec{x},\vec{y}))
$$
where $\varphi(\vec{x},\vec{y})$ is a guarded formula with free
variables among $\vec{x},\vec{y}$ and $\alpha(\vec{x},\vec{y})$ is an
atomic formula or an equality $x=y$ that contains all variables in
$\vec{x},\vec{y}$. The formula $\alpha$ is called the \emph{guard of
  the quantifier}.

In ontologies, we only allow GF sentences $\varphi$ that are
\emph{invariant under disjoint unions}, that is, for all families
$\Bmf_{i}$, $i\in I$, of interpretations with mutually disjoint
domains, the following holds: $\Bmf_{i}\models \vp$ for all $i\in I$
if, and only if, $\bigcup_{i\in I}\Bmf_{i}\models \vp$.
We give a syntactic characterization of GF sentences that are invariant
under disjoint unions.  Denote by \emph{openGF} the fragment of GF
that consists of all (open) formulas whose subformulas are all open
and in which equality is not used as a guard. The fragment \emph{uGF}
of GF is the set of sentences obtained from openGF by a \emph{single guarded universal quantifier}:
if $\varphi(\vec{y})$ is in openGF, then $\forall \vec{y}(\alpha(\vec{y})\rightarrow \varphi(\vec{y}))$ is
in uGF, where $\alpha(\vec{y})$ is an atomic formula or an equality
$y=y$ that contains all variables in $\vec{y}$. We often omit equality guards in uGF
sentences of the form $\forall y(y=y \rightarrow \varphi(y))$ and
simply write $\forall y \varphi$. A \emph{uGF ontology} is a finite
set of sentences in uGF.
%
%
%
%
\begin{theorem}\label{thm:inv}
A GF sentence is invariant under disjoint unions iff it is equivalent to a uGF sentence.
\end{theorem}
\begin{proof}
The direction from right to left is straightforward. For the converse direction,
observe that every GF sentence is equivalent to a Boolean combination of uGF sentences.
Now assume that $\varphi$ is a GF sentence and
invariant under disjoint unions.
Let $\text{cons}(\varphi)$ be the set of all sentences $\chi$ in uGF
with $\varphi \models \chi$.  By compactness of FO it is sufficient to
show that $\text{cons}(\varphi)\models \varphi$.  If this is not the
case, take a model $\Amf_{0}$ of $\text{cons}(\varphi)$ refuting
$\varphi$ and take for any sentence $\psi$ in uGF that is not in
$\text{cons}(\varphi)$ an interpretation $\Amf_{\neg\psi}$ satisfying
$\varphi$ and refuting $\psi$. Let $\Amf_{1}$ be the disjoint union of
all $\Amf_{\neg \psi}$. By preservation of $\varphi$ under disjoint
unions, $\Amf_{1}$ satisfies $\varphi$. By reflection of $\varphi$ for
disjoint unions, the disjoint union $\Amf$ of $\Amf_{0}$ and
$\Amf_{1}$ does not satisfy $\varphi$. Thus $\Amf_{1}$ satisfies
$\varphi$ and $\Amf$ does not satisfy $\varphi$ but by construction
$\Amf$ and $\Amf_{1}$ satisfy the same sentences in uGF.  This is
impossible since $\varphi$ is equivalent to a Boolean combination of
uGF sentences.
\end{proof}
The following example shows that some very simple Boolean combinations of uGF
sentences are not invariant under disjoint unions.
\begin{example}\label{ex:examples33}
Let
\begin{eqnarray*}
\Omc_{\text{UCQ/CQ}} & = & \{(\forall x (A(x) \vee B(x)) \vee \exists x E(x)\}\\
\Omc_{\text{Mat/PTime}} & = & \{\forall x A(x) \vee \forall x B(x)\}
\end{eqnarray*}
Then $\Omc_{\text{Mat/PTime}}$ is not preserved under disjoint unions since $\Dmf_{1}=\{A(a)\}$
and $\Dmf_{2}=\{B(b)\}$ are models of $\Omc_{\text{Mat/PTime}}$ but $\Dmf_{1}\cup \Dmf_{2}$ refutes $\Omc_{\text{Mat/PTime}}$;
$\Omc_{\text{UCQ/CQ}}$ does not reflect disjoint unions since the disjoint union of $\Dmf_{1}'=\{E(a)\}$
and $\Dmf_{2}'=\{F(b)\}$ is a model of $\Omc_{\text{UCQ/CQ}}$ but $\Dmf_{2}'$ refutes $\Omc_{\text{UCQ/CQ}}$.
We will use these ontologies later to explain why we
restrict this study to fragments of GF that are invariant under disjoint unions.
\end{example}
%
%

%
When studying uGF ontologies, we are going to vary several parameters.
The \emph{depth} of a formula $\varphi$ in openGF is the nesting depth
of guarded quantifiers in $\varphi$. The \emph{depth} of a sentence
$\forall \vec{y}(\alpha(\vec{y})\rightarrow \varphi(\vec{y}))$ in uGF
is the depth of $\varphi(\vec{y})$, thus the outermost guarded
quantifier is not counted. The \emph{depth} of a uGF ontology is the
maximum depth of its sentences. We indicate restricted depth in
brackets, writing e.g.\ uGF$(2)$ to denote the set of all uGF
sentences of depth at most 2.
\begin{example}\label{ex:3}
The sentence
$$
\forall xy (R(x,y) \rightarrow (A(x) \vee \exists z S(y,z)))
$$
is in uGF$(1)$ since the openGF formula $A(x) \vee \exists z S(y,z)$
has depth~1.
\end{example}
For every GF sentence $\varphi$, one can construct in polynomial
time a conservative extension $\varphi'$ in uGF$(1)$ by
converting into Scott normal form \cite{DBLP:journals/jsyml/Gradel99}. Thus, the
satisfiability and CQ-evaluation problems for full GF can be polynomially
reduced to the corresponding problem for uGF$(1)$.
%
%

We use uGF$^{-}$ to denote the fragment of uGF where only equality
guards are admitted in the outermost universal quantifier applied to
an openGF formula. Thus, the sentence in Example~\ref{ex:3} (1) is a
uGF sentence of depth 1, but not a uGF$^{-}$ sentence of
depth~1. It is, however, equivalent to the following uGF$^{-}$ sentence of depth 1:
$$
\forall x(\exists y ((R(y,x) \wedge \neg A(y)) \rightarrow \exists z S(x,z)))
$$
An example of a uGF sentence of depth 1 that is not equivalent to a uGF$^{-}$
sentence of depth 1 is given in Example~\ref{ex:ex} below.
Intuitively, uGF sentences of depth 1 can be thought of as
uGF$^-$ sentences of `depth $1.5$' because giving up $\cdot^-$ allows
an additional level of `real' quantification (meaning: over guards
that are not forced to be equality), but only in a syntactically
restricted way.

The two-variable fragment of uGF is denoted with uGF$_2$. More
precisely, in uGF$_2$ we admit only the two fixed variables $x$ and
$y$ and disallow the use of relation symbols of arity exceeding two.
We also consider two extensions of uGF$_2$ with forms of counting.
First, uGF$_2(f)$ denotes the extension of uGF$_2$ with function
symbols, that is, an uGF$_2(f)$ ontology is a finite set of uGF$_2$
sentences and of \emph{functionality axioms} $\forall x\forall
y_{1}\forall y_{2} ((R(x,y_{1})\wedge R(x,y_{2}))\rightarrow
(y_{1}=y_{2}))$ \cite{DBLP:journals/jsyml/Gradel99}. Second, we consider the extension
uGC$_{2}$ of uGF$_{2}$ with counting quantifiers. More precisely, the
language openGC$_2$ is defined in the same way as the two-variable
fragment of openGF, but in addition admits \emph{guarded counting
  quantifiers}
\cite{DBLP:journals/logcom/Pratt-Hartmann07,DBLP:conf/jelia/Kazakov04}:
if $n\in \mathbb{N}$, $\{z_{1},z_{2}\}=\{x,y\}$, and
$\alpha(z_{1},z_{2})\in \{R(z_{1},z_{2}),R(z_{2},z_{1})\}$ for some
$R\in \Sigma$ and $\varphi(z_{1},z_{2})$ is in openGC$_2$, then
$\exists^{\geq n}z_{1}(\alpha(z_{1},z_{2}) \wedge
\varphi(z_{1},z_{2}))$ is in openGC$_2$.  The ontology language
uGC$_{2}$ is then defined in the same way as uGF$_{2}$, using
openGC$_{2}$ instead of openGF$_{2}$. The \emph{depth} of formulas in
uGC$_{2}$ is defined in the expected way, that is, guarded counting
quantifiers and guarded quantifiers both contribute to it.

The above restrictions can be freely combined and we use the obvious
names to denote such combinations. For example, uGF$^-_2(1,f)$ denotes
the two-variable fragment of uGF with function symbols and where all
sentences must have depth 1 and the guard of the outermost quantifier
must be equality. Note that uGF admits equality, although in a
restricted way (only in non-guard positions, with the possible
exception of the guard of the outermost quantifier). We shall also
consider fragments of uGF that admit no equality at all except as a
guard of the outermost quantifier. To emphasize that the restricted
use of equality is allowed, we from now on use the equality symbol in
brackets whenever equality is present, as in uGF$(=)$, uGF$^-(1,=)$,
and uGC$^-_2(1,=)$. Conversely, uGF, uGF$^-(1)$,
and uGC$^-_2(1)$ from now on denote the corresponding fragments
where equality is only allowed as a guard of the outermost quantifier.


\medskip

Description logics are a popular family of ontology languages that are
related to the guarded fragments of FO introduced above. We briefly
review 
the basic description logic \ALC, further details on this and other
DLs mentioned in this paper can be found in the appendix and in
\cite{handbook}.  DLs generally use relations of arity one and two,
only.  An \emph{\ALC concept} is formed according to the syntax rule
$$
C,D ::= A \mid \top \mid \bot \mid
\neg C \mid C \sqcap D \mid C \sqcup D \mid \exists R . C \mid
\forall R .C
$$
where $A$ ranges over unary relations and $R$ over binary
relations. An \emph{\ALC ontology} \Omc is a finite set of
\emph{concept inclusions} $C \sqsubseteq D$, with $C$ and $D$ \ALC
concepts. The semantics of \ALC concepts $C$ can be given by
translation to openGF formulas $C^*(x)$ with one free variable $x$ and
two variables overall. A concept inclusion $C\sqsubseteq D$ then
translates to the uGF$_{2}^{-}$ sentence $\forall x (C^*(x)
\rightarrow D^*(x))$.  The \emph{depth} of an $\mathcal{ALC}$ concept
is the maximal nesting depth of $\exists R$ and $\forall R$.  The
\emph{depth} on an $\mathcal{ALC}$ ontology is the maximum depth of
concepts that occur in it.  Thus, every $\mathcal{ALC}$ ontology of
depth $n$ is a uGF$^{-}_{2}$ ontology of depth $n$. When translating
into uGF$_2$ instead of into uGF$^{-}_{2}$, the depth might decrease
by one because one can exploit the outermost quantifier (which does
not contribute to the depth). A more detailed description of the relationship
between DLs and fragments of uGF is given in the appendix.
\begin{example}\label{ex:ex}
The \ALC concept inclusion $\exists S.A \sqsubseteq \forall R . \exists S . B$
  has depth~2, but is equivalent to the uGF$_{2}(1)$ sentence
$$
\forall xy( R(x,y) \rightarrow ((\exists S.A)^{\ast}(x) \rightarrow
(\exists S.B)^{\ast}(y))
$$
\end{example}
Note that for any ontology $\Omc$ in any DL considered in this paper one can construct
in a straightforward way in polynomial time a conservative extension $\Omc^{\ast}$ of
$\Omc$ of depth one. In fact, many DL algorithms for satisfiability or query evaluation
assume that the ontology is of depth one and normalized.

We also consider the extensions of $\mathcal{ALC}$ with inverse roles
$R^{-}$ (denoted in the name of the DL by the letter $\mathcal{I}$),
role inclusions $R\sqsubseteq S$ (denoted by $\mathcal{H}$), qualified
number restrictions $(\geq n \; R \;C)$ (denoted by $\mathcal{Q}$),
partial functions as defined above (denoted by $\mathcal{F}$), and local functionality expressed by $(\leq 1 R)$
(denoted by $\mathcal{F}_{\ell}$). The depth of
ontologies formulated in these DLs is defined in the obvious
way. Thus, $\mathcal{ALCHIQ}$ ontologies (which admit all the
constructors introduced above) translate into uGC$^{-}_{2}$
ontologies, preserving the depth.

For any syntactic object $O$ (such as an ontology or a query), we use
$|O|$ to denote the number of symbols needed to write $O$, counting
relation names, variable names, and so on as a single symbol and
assuming that numbers in counting quantifiers and DL number
restrictions are coded in unary.

\subsection{Guarded Tree Decompositions}
We introduce guarded tree decompositions and rooted acyclic queries~\cite{DBLP:books/daglib/p/Gradel014}.
A set $G \subseteq \text{dom}(\Amf)$ is \emph{guarded} in the
interpretation $\mathfrak{A}$
if $G$ is a singleton or there are $R \in \Sigma$ and
$R(a_1,\dotsc,a_k) \in \mathfrak{A}$ such that $G =
\{a_1,\dotsc,a_k\}$.  By $S(\mathfrak{A})$, we denote the set of all
guarded sets in $\mathfrak{A}$.  A tuple $(a_1,\dotsc,a_k) \in A^k$ is
\emph{guarded} in $\mathfrak{A}$ if $\{a_1,\dotsc,a_k\}$ is a subset
of some guarded set in $\mathfrak{A}$. A \emph{guarded tree
  decomposition} of $\mathfrak{A}$ is a triple $(T,E,\text{bag})$ with
$(T,E)$ an acyclic undirected graph and $\text{bag}$ a function that
assigns to every $t\in T$ a set $\text{bag}(t)$ of atoms such that
$\Amf_{|\text{dom}(\text{bag}(t))}=\text{bag}(t)$ and
\begin{enumerate}
\item $\mathfrak{A} = \bigcup_{t\in T}\text{bag}(t)$;
\item $\text{dom}(\text{bag}(t))$ is guarded for every $t\in T$;
\item $\{t \in T\mid a\in \text{dom}(\text{bag}(t))\}$ is connected in
  $(T,E)$, for every $a\in \text{dom}(\mathfrak{A})$.
\end{enumerate}
We say that $\mathfrak{A}$ is \emph{guarded tree decomposable} if
there exists a guarded tree decomposition of $\mathfrak{A}$.  We call
$(T,E,\text{bag})$ a \emph{connected guarded tree decomposition
  (cg-tree decomposition)} if, in addition, $(T,E)$ is connected
(i.e., a tree) and $\text{dom}(\text{bag}(t)) \cap
\text{dom}(\text{bag}(t'))\not=\emptyset$ for all $(t,t')\in E$. In
this case, we often assume that $(T,E)$ has a designated root $r$,
which allows us to view $(T,E)$ as a directed tree whenever convenient.

A CQ $q\leftarrow \phi$ is a \emph{rooted acyclic query (rAQ)} if
there exists a cg-tree decomposition $(T,E,\text{bag})$ of the
instance $\Dmf_{q}$ with root $r$ such that
$\text{dom}(\text{bag}(r))$ is the set of answer variables of $q$.
Note that, by definition, rAQs are non-Boolean queries.
\begin{example}
The CQ
$$
q(x) \leftarrow \phi, \quad \phi = R(x,y)\wedge R(y,z)\wedge R(z,x)
$$
is not an rAQ since $\Dmf_{q}$ is not guarded tree decomposable. By adding
the conjunct $Q(x,y,z)$ to $\phi$ one obtains an rAQ.
\end{example}
We will frequently use the following construction: let $\Dmf$ be an
instance and $\mathcal{G}$ a set of guarded sets in $\Dmf$. Assume
that $\Bmf_{G}$, $G\in \mathcal{G}$, are interpretations such that
$\text{dom}(\Bmf_{G})\cap \text{dom}(\Dmf)=G$ and
$\text{dom}(\Bmf_{G_{1}})\cap \text{dom}(\Bmf_{G_{2}})=G_{1}\cap
G_{2}$ for any two distinct guarded sets $G_{1}$ and $G_{2}$ in
$\mathcal{G}$. Then the interpretation
$$
\Bmf = \Dmf \cup \bigcup_{G\in \mathcal{G}}\Bmf_{G}
$$
is called the \emph{interpretation obtained from $\Dmf$ by hooking
  $\Bmf_{G}$ to $\Dmf$ for all $G\in \mathcal{G}$}.
If the $\Bmf_{G}$ are cg-tree decomposable interpretations with $\text{dom}(\text{bag}(r))=G$ for
the root $r$ of a (fixed) cg-tree decomposition of $\Bmf_{G}$, then $\Bmf$ is called a
\emph{forest model of $\Dmf$ defined using $\mathcal{G}$}.
If $\mathcal{G}$ is the set of all maximal guarded sets in $\Dmf$, then we call $\Bmf$
simply a \emph{forest model of~$\Dmf$}. The following result can be proved using standard
guarded tree unfolding~\cite{DBLP:journals/jsyml/Gradel99,DBLP:books/daglib/p/Gradel014}.
\begin{lemma}\label{lem:forestmodel}
  Let $\Omc$ be a uGF$(=)$ or uGC$_{2}(=)$ ontology, $\Dmf$ a
  possibly infinite instance, and $\Amf$ a model of $\Dmf$ and $\Omc$.
  Then there exists a forest model $\Bmf$ of $\Dmf$ and $\Omc$ and a homomorphism $h$ from $\Bmf$ to $\Amf$ that
  preserves $\text{dom}(\Dmf)$.
%
\end{lemma}
\section{Materializability}

We introduce and study materializability of ontologies as a necessary
condition for query evaluation to be in {\sc PTime}. In brief, an
ontology \Omc is materializable if for every instance \Dmf, there is a
model \Amf of \Omc and \Dmf such that for all queries, the answers on
\Amf agree with the certain answers on \Dmf given~\Omc. We show that
this sometimes, but not always, coincides with existence of universal
models defined in terms of homomorphisms. We then prove that in
uGF$(=)$ and uGC$_2(=)$, non-materializability implies {\sc coNP}-hard
query answering while this is not the case for GF. Using these
results, we further establish that in uGF$(=)$ and uGC$_{2}(=)$, query
evalution w.r.t.\ ontologies to be in {\sc PTime},
Datalog$^{\not=}$-rewritable, and {\sc coNP}-hard does not depend on
the query language, that is, all these properties agree for rAQs, CQs,
and UCQs. Again, this is not the case for GF.

%
%
\begin{definition}[Materializability]\label{cond0}
  Let $\Omc$ be an FO$(=)$-ontology, $\Qmc$ a class of queries, and $\mathcal{M}$ a class of instances.
  Then
  \begin{itemize}

  \item an interpretation $\Bmf$ is a \emph{$\Qmc$-materialization} of $\Omc$ and an instance $\Dmf$
   if it is a model of $\Omc$ and $\Dmf$ and for all
$q(\vec{x}) \in \Qmc$ and $\vec{a}$ in $\text{dom}(\Dmf)$,
    $\mathfrak{B} \models q(\vec{a})$ iff $\Omc,\Dmf\models q(\vec{a})$.

  \item $\Omc$ is \emph{$\Qmc$-materializable for~$\mathcal{M}$}
  if for every instance $\Dmf\in \mathcal{M}$ that is consistent w.r.t.~$\Omc$, there is a
  $\Qmc$-materialization of $\Omc$ and~$\Dmf$.
  \end{itemize}
  If $\mathcal{M}$ is the class of all instances, we simply speak of
 $\Qmc$-materializability of \Omc.
\end{definition}
We first observe that the materializability of ontologies does not
depend on the query language (although concrete materializations do).
\begin{theorem}\label{thm:materializabilityeq}
Let \Omc be a uGF$(=)$ or uGC$_{2}(=)$ ontology and \Mmc a class of
instances. Then the following conditions are equivalent:
\begin{enumerate}
\item $\Omc$ is rAQ-materializable for \Mmc;
\item $\Omc$ is CQ-materializable for \Mmc;
\item $\Omc$ is UCQ-materializable for \Mmc.
\end{enumerate}
\end{theorem}
\begin{proof} The only non-trivial implication is (1) $\Rightarrow$
  (2). It can be proved by using Lemma~\ref{lem:forestmodel} and
  showing that if \Amf is a rAQ-materialization of an ontology \Omc
  and an instance \Dmf, then any forest model \Bmf of \Omc and \Dmf
  which admits a homomorphism to \Amf that preserves $\mn{dom}(\Dmf)$
  is a CQ-materialization of \Omc and \Dmf.
\end{proof}
Because of Theorem~\ref{thm:materializabilityeq}, we from now on speak
of \emph{materializability} without reference to a query language and
of \emph{materializations} instead of UCQ-materializations (which are
then also CQ-materializations and rAQ-materializations).

\smallskip

A notion closely related to materializations are (homomorphically)
universal models as used e.g.\ in data exchange~\cite{FaginKMP05,DeutschNR08}. A model of an ontology \Omc and an
instance \Dmf is \emph{hom-universal} if there is a homomorphism
preserving $\text{dom}(\Dmf)$ into any model of \Omc and~\Dmf.  We say
that an ontology \Omc \emph{admits hom-universal models} if there is a
hom-universal model for \Omc and any instance \Dmf. It is well-known
that hom-universal models are closely related to what we call
UCQ-materializations. In fact, in many DLs and in uGC$_{2}(=)$,
materializability of an ontology \Omc coincides with \Omc admitting
hom-universal models (although for concrete models, being
hom-universal is not the same as being a materialization). We show in
the long version that this is not the case for ontologies in uGF$(2)$
(with three variables).
The proof also shows that admitting
hom-universal models is not a necessary condition for query evaluation
to be in {\sc PTime} (in contrast to materializability).
%
%
\begin{lemma}
\label{lem:canhom}
A uGC$_{2}(=)$ ontology is materializable iff it admits hom-universal
models.  This does not hold for uGF$(2)$ ontologies.
\end{lemma}
The following theorem links materializability to computational
complexity, thus providing the main reason for our interest into this
notion. The proof is by reduction of 2+2-SAT~\cite{Schaerf-93}, a
variation of a related proof from \cite{KR12-csp}.
\begin{theorem}
\label{thm:nomatlower}
Let \Omc be an FO$(=)$-ontology that is invariant under disjoint unions.
If \Omc is not materializable, then rAQ-evaluation w.r.t.\ \Omc is
{\sc coNP}-hard.
\end{theorem}
We remark that, in the proof of Theorem~\ref{thm:nomatlower}, we use
instances and rAQs that use additional fresh (binary) relation symbols, that
is, relation symbols that do not occur in \Omc.

The ontology $\Omc_{\text{Mat/PTime}}$ from
Example~\ref{ex:examples33} shows that Theorem~\ref{thm:nomatlower}
does not hold for GF ontologies, even if they are of depth 1 and use
only a single variable.  In fact, $\Omc_{\text{Mat/PTime}}$ is not
CQ-materializable, but CQ-evaluation is in
{\sc PTime}  (which is both easy to see).
\begin{theorem}\label{thm:PTime}
  For all $\uGF$ and $\uGC$ ontologies $\Omc$, the following are
  equivalent:
  \begin{enumerate}
  \item rAQ-evaluation w.r.t.~$\Omc$ is in \PTime;

  \item CQ-evaluation w.r.t.~$\Omc$ is in \PTime;

  \item UCQ-evaluation w.r.t.~$\Omc$ is in \PTime.

  \end{enumerate}
  This remains true when `in \PTime' is replaced with
  `Datalog$^{\neq}$-rewritable' and with `\coNP-hard'
  (and with `Datalog-rewritable' if \Omc is a uGF ontology).
\end{theorem}

\begin{proof}
  By Theorem~\ref{thm:nomatlower}, we can concentrate on ontologies
  that are materializable.  For the non-trivial implication of Point~3
  by Point~1, we exploit materializability to rewrite UCQs into a
  finite disjunction of queries $q \wedge \bigwedge_i q_i$ where $q$
  is a ``core CQ'' that only needs to be evaluated over the input
  instance $\Dmf$ (ignoring labeled nulls) and each $q_i$ is a rAQ.
  This is similar to squid decompositions
  in~\cite{DBLP:journals/jair/CaliGK13}, but more subtle due to the
  presence of subqueries that are not connected to any answer variable
  of $q$.  
  Similar constructions are used also to deal with
  Datalog$^{\neq}$-rewritability and with {\sc coNP}-hardness.
\end{proof}

The ontology $\Omc_{\text{UCQ/CQ}}$ from Example~\ref{ex:examples33}
shows that Theorem~\ref{thm:PTime} does not hold for GF ontologies,
even if they use only a single variable and are of depth 1 up to an
outermost universal quantifier with an equality guard.
\begin{lemma}
\label{lem:ucqcq}
  CQ-evaluation w.r.t.\ $\Omc_{\text{UCQ/CQ}}$ is in \PTime and
  UCQ-evaluation w.r.t.\ $\Omc_{\text{UCQ/CQ}}$ is {\sc coNP}-hard.
\end{lemma}
The lower bound essentially follows the construction in the proof of
Theorem~\ref{thm:nomatlower} and the upper bound is based on a
case analysis, depending on which relations occur in the CQ and in
the input instance.
%
%
%
\section{Unravelling Tolerance}
\label{sec:unravelling-tolerance}
While materializability of an ontology is a necessary condition for
\PTime query evaluation in uGF$(=)$ and uGC$_{2}(=)$, we now identify
a sufficient condition called unravelling tolerance that is based on
unravelling instances into cg-tree decomposable instances (which might
be infinite). In fact, unravelling tolerance is even a sufficient
condition of Datalog$^{\not=}$-rewritability and we will later
establish our dichotomy results by showing that, for the ontology
languages in question, materializability implies unravelling
tolerance.

%
We start with introducing suitable forms of unravelling (also called guarded tree unfolding,
see~\cite{DBLP:books/daglib/p/Gradel014} and references therein).
The \emph{uGF-unravelling} $\Dmf^{u}$ of an instance $\Dmf$ is constructed
as follows. Let $T(\Dmf)$ be the set of all sequences
$t=G_{0}G_{1}\cdots G_{n}$ where $G_{i}$, $0\leq i \leq n$, are
maximal guarded sets of $\Dmf$ and
\begin{description}
\item[(a)] $G_{i}\not= G_{i+1}$,
\item[(b)] $G_{i}\cap G_{i+1}\not=\emptyset$, and
\item[(c)] $G_{i-1}\not= G_{i+1}$.
\end{description}
In the following, we associate each $t \in T(\Dmf)$ with a set of atoms $\mn{bag}(t)$. Then we define $\Dmf^u$ as
$\bigcup_{t\in T(\Dmf)}\text{bag}(t)$ and note that
$(T(\Dmf),E,\text{bag})$ is a cg-tree decomposition of $\Dmf^{u}$
where $(t,t')\in E$ if $t'=tG$ for some $G$.

Set $\text{tail}(G_{0}\cdots G_{n})=G_{n}$.  Take an infinite supply
of \emph{copies} of any
$d\in \text{dom}(\Dmf)$. We set
$e^{\uparrow}=d$ if $e$ is a copy of~$d$. We define $\text{bag}(t)$
(up to isomorphism) proceeding by induction on the length of the
sequence $t$. For any $t=G$, $\text{bag}(t)$ is an instance whose
domain is a set of copies of $d\in G$ such that the mapping
$e\mapsto e^{\uparrow}$ is an isomorphism from $\text{bag}(G)$ onto
the subinstance $\Dmf_{|G}$ of $\Dmf$ induced by $G$.  To define
$\text{bag}(t')$ for $t'=tG'$ when
$\mn{tail}(t)=G$, 
take for any $d\in G'\setminus G$ a fresh
copy $d'$ of $d$ and define $\text{bag}(t')$ with domain $\{ d' \mid
d\in G'\setminus G\}\cup \{ e \in \text{bag}(t) \mid e^{\uparrow} \in
G' \cap G)\}$ such that the mapping $e\mapsto e^{\uparrow}$ is an
isomorphism from $\text{bag}(t')$ onto $\Dmf_{|G'}$.
The following example illustrates the construction of $\Dmf^{u}$.
\begin{example}\label{ex:ddd}
(1) Consider the instance $\Dmf$ depicted below with the maximal guarded sets
$G_{1},G_{2},G_{3}$. Then the unravelling $\Dmf^{u}$ of $\Dmf$ consists
of three isomorphic chains (we depict only one such chain):

\includegraphics{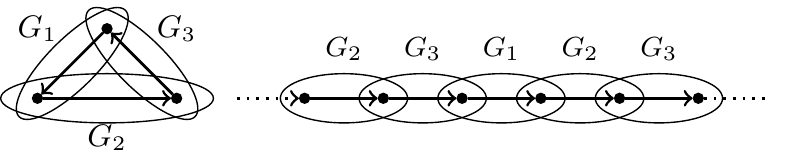}

(2) Next consider the instance $\Dmf$ depicted below which has the shape of a
tree of depth one with root $a$ and has three maximal guarded sets $G_{1},G_{2},G_{3}$.
Then the unravelling $\Dmf^{u}$ of $\Dmf$ consists of three isomorphic trees of depth one
of infinite outdegree (again we depict only one):

\includegraphics{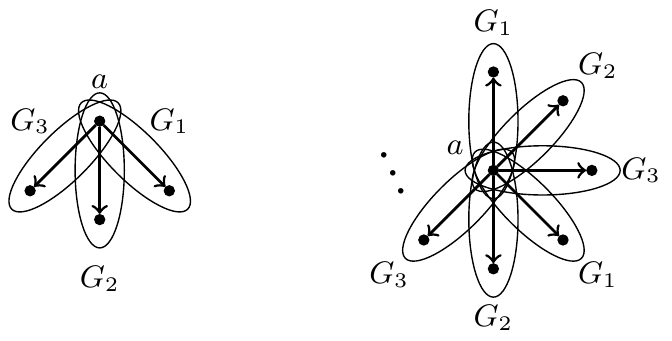}
\end{example}
By construction, the mapping $h: e \mapsto e^{\uparrow}$ is a homomorphism from $\Dmf^{u}$ onto $\Dmf$ and
the restriction of $h$ to any guarded set $G$ is an isomorphism. It follows that for any uGF$(=)$ ontology
$\Omc$, UCQ $q(\vec{x})$, and $\vec{a}$ in $\Dmf^{u}$, if $\Omc,\Dmf^{u}\models q(\vec{a})$, then
$\Omc,\Dmf\models q(h(\vec{a}))$. This implication does not hold for ontologies in the guarded fragment
with functions or counting. To see this, let
$$
\Omc=\{\forall x(\exists^{\geq 4}y R(x,y) \rightarrow A(x))\}
$$
Then $\Omc,\Dmf^{u}\models A(a)$ for the instance $\Dmf$ from Example~\ref{ex:ddd} (2) but $\Omc,\Dmf\not\models A(a)$.
For this reason the uGF-unravelling is not appropriate for the guarded fragment with functions or
counting. By replacing Condition~(c) by the stronger condition
\begin{description}
\item[(c$'$)] $G_{i}\cap G_{i-1} \not= G_{i} \cap G_{i+1}$,
\end{description}
we obtain an unravelling that we call \emph{uGC$_{2}$-unravelling} and that we apply whenever all relations
have arity at most two. One can show that the uGC$_{2}$-unravelling of an instance preserves the number of
$R$-successors of constants in $\Dmf$ and that, in fact, the implication `$\Omc,\Dmf^{u}\models q(\vec{a}) \Rightarrow
\Omc,\Dmf\models q(h(\vec{a}))$' holds for every uGC$_{2}(=)$ ontology $\Omc$, UCQ $q(\vec{x})$, and tuple $\vec{a}$ in
the uGC$_{2}$-unravelling $\Dmf^{u}$ of $\Dmf$.

\smallskip

We are now ready to define unravelling tolerance.
For a maximal guarded set $G$ in $\Dmf$, the \emph{copy in
  $\text{bag}(G)$ of a tuple $\vec{a}=(a_{1},\ldots,a_{k})$ in $G$} is
the unique $\vec{b}= (b_{1},\ldots,b_{k})$ in
$\text{dom}(\text{bag}(G))$ such that $b_{i}^{\uparrow} = a_{i}$ for
$1\leq i \leq k$.

\begin{definition}\label{def:unrav}
  A uGF$(=)$ (resp.\ uGC$_{2}(=)$) ontology $\Omc$ is
  \emph{unravelling tolerant} if for every instance $\Dmf$, every rAQ
  $q(\vec{x})$, and every tuple $\vec{a}$ in $\Dmf$ such that the set
  $G$ of elements of $\vec{a}$ is maximally guarded in $\Dmf$ the
  following are equivalent:
\begin{enumerate}
\item $\Omc,\Dmf\models q(\vec{a})$;
\item $\Omc,\Dmf^{u}\models q(\vec{b})$ where $\vec{b}$ is the copy of
  $\vec{a}$ in $\text{bag}(G)$
\end{enumerate}
where $\Dmf^u$ is the uGF-unravelling (resp.\ the
uGC$_{2}$-unravelling) of~$\Dmf$.
\end{definition}
We have seen above that the implication (2) $\Rightarrow$ (1) in
Definition~\ref{def:unrav} holds for every uGF$(=)$ and uGC$_{2}(=)$
ontology and every UCQ. Note that it is pointless to define unravelling tolerance using the implication (1) $\Rightarrow$ (2)
for UCQs or CQs that are not acyclic. The following example shows that (1) $\Rightarrow$ (2) does not always hold for rAQs.
\begin{example}
Consider the uGF ontology $\Omc$ that contains the sentences
$$
\forall x \big(X(x) \rightarrow (\exists y (R(x,y) \wedge X(y)) \rightarrow E(x))\big)
$$
with $X\in \{A,\neg A\}$ and
$$
\forall x \big(E(x) \rightarrow ((R(x,y) \vee R(y,x)) \rightarrow E(y))\big)
$$
For instances $\Dmf$ not using $A$, $\Omc$ states that $E(a)$ is entailed for all $a\in \text{dom}(\Dmf)$
that are $R$-connected to some $R$-cycle in $\Dmf$ with an odd number of constants. Thus, for the instance
$\Dmf$ from Example~\ref{ex:ddd} (1) we have $\Omc,\Dmf \models E(a)$ for every $a\in \text{dom}(\Dmf)$
but $\Omc,\Dmf^{u}\not\models E(a)$ for any $a\in \text{dom}(\Dmf^{u})$.
\end{example}
We now show that, as announced, unraveling tolerance implies that
query evaluation is Datalog$^{\not=}$-rewritable.
\begin{theorem}\label{thm:datalog}
  For all $\uGF$ and $\uGC$ ontologies \Omc, unravelling tolerance of
  \Omc implies that rAQ-evaluation w.r.t.\ \Omc is
  Datalog$^{\neq}$-rewritable (and Datalog-rewritable if \Omc is
  formulated in uGF).
\end{theorem}
\begin{proof}
  We sketch the proof for the case that $\Omc$ is a $\uGF$ ontology;
  similar constructions work for the other cases.
  Suppose that $\Omc$ is unravelling tolerant, and that $q(\vec{x})$ is a rAQ.
  We construct a Datalog$^{\neq}$ program $\Pi$ that, given an instance $\Dmf$,
  computes the certain answers $\vec{a}$ of $q$ on $\Dmf$ given $\Omc$,
  where w.l.o.g.\ we can restrict our attention to answers $\vec{a}$
  such that the set $G$ of elements of $\vec{a}$ is maximally guarded in $\Dmf$.
  By unravelling tolerance,
  it is enough to check if $\Omc,\Dmf^u \models q(\vec{b})$,
  where $\vec{b}$ is the copy of $\vec{a}$ in $\bag(G)$
  and $\Dmf^u$ is the uGF-unravelling of $\Dmf$.

  The Datalog$^{\neq}$ program $\Pi$
  assigns to each maximally guarded tuple $\vec{a} = (a_1,\dotsc,a_k)$ in $\Dmf$
  a set of \emph{types}.
  Here, a type is a maximally consistent set of uGF formulas
  with free variables in $x_1,\dotsc,x_k$,
  where the variable $x_i$ represents the element $a_i$.
  It can be shown that we only need to consider types
  with formulas of the form $\phi$ or $\lnot \phi$,
  where $\phi$ is obtained from a subformula of $\Omc$ or $q$
  by substituting a variable in $x_1,\dotsc,x_k$
  for each of its free variables,
  or $\phi$ is an atomic formula in the signature of $\Omc,q$
  with free variables in $x_1,\dotsc,x_k$.
  In particular, the set of all types is finite.
  We further restrict our attention to types $\theta$
  that are realizable in some model of $\Omc$,
  i.e., there is a model $\Bmf(\theta)$ of $\Omc$
  containing all elements of $\vec{a}$
  that is a model of each formula in $\theta$
  under the interpretation $x_i \mapsto a_i$.
  The Datalog$^{\neq}$ program $\Pi$ ensures the following:
  \begin{enumerate}
  \item
    for any two maximally guarded tuples $\vec{a} = (a_1,\dotsc,a_k)$,
    $\vec{b} = (b_1,\dotsc,b_l)$ in $\Dmf$
    that share an element,
    and any type $\theta$ assigned to $\vec{a}$
    there is a type $\theta'$ assigned to $\vec{b}$
    that is \emph{compatible} to $\theta$
    (intuitively, the two types agree on all formulas
     that only talk about elements shared by $\vec{a}$ and $\vec{b}$);
  \item
    a tuple $\vec{a} = (a_1,\dotsc,a_k)$ is an answer to $\Pi$
    if all types assigned to $\vec{a}$ contain $q(x_1,\dotsc,x_k)$,
    or some maximally guarded tuple $\vec{b}$ in $\Dmf$
    has no type assigned to it.
  \end{enumerate}
  It can be shown that $\vec{a}$ is an answer to $\Pi$
  iff $\Omc,\Dmf^u \models q(\vec{a})$.

  The interesting part is the ``if'' part.
  Suppose $\vec{a} = (a_1,\dotsc,a_k)$ is \emph{not} an answer to $\Pi$.
  Then, each maximally guarded tuple $\vec{b}$ in $\Dmf$
  is assigned to at least one type,
  and for some type $\theta^*$ assigned to $\vec{a}$
  we have $q(x_1,\dotsc,x_k) \notin \theta^*$.
  We use this type assignment
  to label each maximally guarded tuple $\vec{b}$ of $\Dmf^u$
  with a type $\theta_{\vec{b}}$
  so that
  (1) for each maximally guarded tuple $\vec{c}$ of $\Dmf^u$
  that shares an element with $\vec{b}$
  the two types $\theta_{\vec{b}}$ and $\theta_{\vec{c}}$ are compatible;
  and
  (2) $\theta^* = \theta_{\vec{a}^*}$,
  where $\vec{a}^*$ is the copy of $\vec{a}$ in $G = \set{a_1,\dotsc,a_k}$.
  We can now show that the interpretation $\Amf$ obtained from $\Dmf^u$
  by hooking $\Bmf(\theta_{\vec{b}})$ to $\Dmf^u$,
  for all maximally guarded tuples $\vec{b}$ of $\Dmf^u$,
  is a model of $\Omc$ and $\Dmf^u$ with $\Amf \not\models q(\vec{a}^*)$.
\end{proof}

\section{Dichotomies}
\label{sect:dichotomies}

We prove dichotomies between {\sc PTime} and {\sc coNP} for query
evaluation in the five ontology languages displayed in the bottommost
part of Figure~\ref{fig:results}. In fact, the dichotomy is even
between Datalog$^{\not=}$-rewritability and {\sc coNP}.
%
The proof establishes that for ontologies $\Omc$ formulated in any of
these languages, CQ-evaluation w.r.t.\ \Omc is
Datalog$^{\not=}$-rewritable iff it is in {\sc PTime} iff \Omc is
unravelling tolerant iff $\Omc$ is materializable for the class of
(possibly infinite) cg-tree decomposable instances iff \Omc is
materializable and that, if none of this is the case, CQ-evaluation
w.r.t.\ \Omc is {\sc coNP}-hard. The main step towards the dichotomy
result is provided by the following theorem.
%
%
%
\begin{theorem}\label{thm:infiniteunravel}
  Let $\Omc$ be an ontology formulated in one of uGF$(1)$,
  uGF$^{-}(1,=)$, uGF$^{-}_{2}(2)$, uGC$^{-}_{2}(1,=)$, or an
  $\mathcal{ALCHIF}$ ontology of depth~2.  If $\Omc$ is materializable
  for the class of (possibly infinite) cg-tree decomposable instances $\Dmf$ with
  $\text{sig}(\Dmf) \subseteq \text{sig}(\Omc)$, then $\Omc$ is
  unravelling tolerant.
\end{theorem}
\begin{proof}
  We sketch the proof for uGF$(1)$ and uGF$^{-}_{2}(2)$ ontologies
  $\Omc$ and then discuss the remaining cases.  Assume that \Omc
  satisfies the precondition from
  Theorem~\ref{thm:infiniteunravel}. Let $\Dmf$ be an instance and
  $\Dmf^{u}$ its uGF unravelling. Let $\vec{a}$ be a tuple in a
  maximal guarded set $G$ in $D$ and $\vec{b}$ be the copy in
  $\text{dom}(\text{bag}(G))$ of $\vec{a}$. Further let $q$ be an rAQ
  such that $\Omc,\Dmf^{u}\not\models q(\vec{b})$.  We have to show
  that $\Omc,\Dmf\not\models q(\vec{a})$. Using the condition that $\Omc$ is materializable
  for the class of cg-tree decomposable instances $\Dmf$ with
  $\text{sig}(\Dmf) \subseteq \text{sig}(\Omc)$, it can be shown
  that there exists a materialization $\Bmf$ of $\Omc$ and $\Dmf^{u}$.

  By Lemma~\ref{lem:forestmodel} we may assume that $\Bmf$
  is a forest model which is obtained from $\Dmf^{u}$ by hooking
  cg-tree decomposable $\Bmf_{\text{bag}(t)}$ to maximal guarded
  $\text{bag}(t)$ in
  $\Dmf^{u}$. 
  Now we would like to obtain a model of $\Omc$ and the original
  $\Dmf$ by hooking for any maximal guarded $G$ in $\Dmf$ the
  interpretation $\Bmf_{\text{bag}(G)}$ to $\Dmf$ rather than to
  $\Dmf^{u}$. However, the resulting model is then not guaranteed to
  be a model of $\Omc$. The following example illustrates this. Let
  $\Omc$ contain
$$
\forall x \exists y (S(x,y) \wedge A(y)),
$$
and for $\varphi(x) =  \exists z (S(x,z) \wedge \neg A(z))$
$$
\forall xy( R(x,y) \rightarrow (\varphi(x)\rightarrow \varphi(y))
$$
Thus in every model of $\Omc$ each node has an $S$-successor in $A$ and having an
$S$-successor that is not in $A$ is propagated along $R$. $\Omc$ is unravelling tolerant.
Consider the instance $\Dmf$ from Example~\ref{ex:ddd} (1)
depicted here again with the maximal guarded sets $G_{1},G_{2},G_{3}$.
\includegraphics{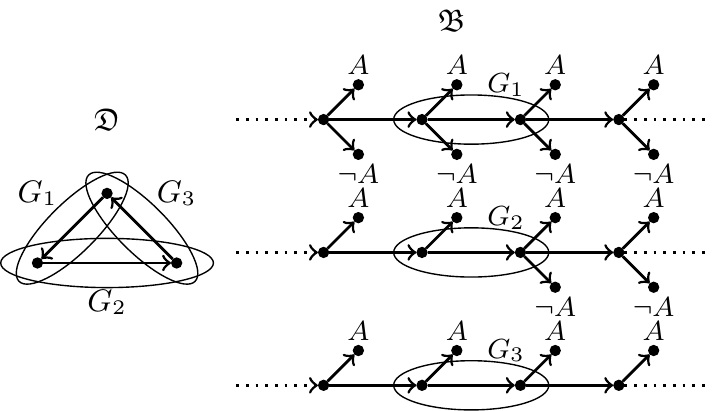}
\newline
We have seen that the unravelling $\Dmf^{u}$ of $\Dmf$ consists of three chains. An example
of a forest model $\Bmf$ of $\Omc$ and $\Dmf^{u}$ is given in the figure.
Even in this simple example a naive way of hooking the models $\Bmf_{G_{i}}$, $i=1,2,3$, to the original
instance $\Dmf$ will lead to an interpretation not satisfying $\Omc$ as the propagation
condition for $S$-successors not in $A$ will not be satisfied.
To ensure that we obtain a model of $\Omc$ we first define a new instance $\Dmf^{u+}\supseteq \Dmf^{u}$
by adding to each maximal guarded set in $\Dmf^{u}$ a copy of any entailed rAQ.
The following facts are needed for this to work:
\begin{enumerate}
\item Automorphisms: for any $t,t'\in T(\Dmf)$ with $\text{tail}(t)=\text{tail}(t')$
there is an automorphism $\hat{h}_{t,t'}$ of $\Dmf^{u}$ mapping
$\text{bag}(t)$ onto $\text{bag}(t')$ and such that
$\hat{h}_{t,t'}(a)^{\uparrow}=a^{\uparrow}$ for all $a\in \text{dom}(\Dmf^{u})$. (This is trivial in the
example above.) It is for this property that we need that $\Dmf^{u}$ is obtained from $\Dmf$ using maximal
guarded sets only and the assumption that $G_{i-1}\not=G_{i+1}$. It follows that if $\text{tail}(t)=\text{tail}(t')$
then the same rAQs are entailed at $\text{bag}(t)$ and $\text{bag}(t')$ in $\Dmf^{u}$.
\item Homomorphism preservation: if there is a homomorphism $h$ from instance $\Dmf$ to instance $\Dmf'$ then
$\Omc,\Dmf\models q(\vec{a})$ entails $\Omc,\Dmf'\models q(h(\vec{a}))$. Ontologies in uGF$(1)$ and uGF$^{-}_{2}(2)$
have this property as they do not use equality nor counting. Because of homomorphism preservation the answers in $\Dmf^{u}$
to rAQs are invariant under moving from $\Dmf^{u}$ to $\Dmf^{u+}$. Note that the remaining ontology languages
in Theorem~\ref{thm:infiniteunravel} do not have this property.
\end{enumerate}
Now using that $\Dmf^{u+}$ is materializable w.r.t.~$\Omc$ one can uniformize a materialization
$\Bmf^{u+}$ of $\Dmf^{u+}$ that is a forest model in such a
way that the automorphisms $\hat{h}_{t,t'}$ for $\text{tail}(t)=\text{tail}(t')$ extend to
automorphisms of the resulting model $\Bmf^{u\ast}$ which also still satisfies $\Omc$. In the example,
after uniformization all chains will behave in the same way in the sense that every node
receives an $S$-successor not in $A$.
We then obtain a forest model $\Bmf^{\ast}$ of $\Dmf$ by hooking the interpretations
$\Bmf^{u\ast}_{\text{bag}(G)}$ to the maximal guarded sets $G$ in $\Dmf$. $(\Bmf^{\ast},\vec{a})$ and
$(\Bmf^{u\ast},\vec{b})$ are guarded bisimilar.
Thus $\Bmf^{\ast}$ is a model of $\Omc$ and $\Bmf^{\ast}\not\models q(\vec{a})$, as required.

For uGF$^{-}(1,=)$ and uGC$^{-}_{2}(1,=)$ the intermediate step of constructing $\Dmf^{u+}$ is not required
as sentences have smaller depth and no uniformization is needed to satisfy the ontology in the new model.
For $\mathcal{ALCHIF}$ ontologies of depth~2 uniformization by constructing $\Dmf^{u+}$ is needed and has to be done
carefully to preserve functionality when adding copies of entailed rAQs to $\Dmf^{u}$.
\end{proof}
We can now prove our main dichotomy result.
\begin{theorem}\label{thm:dichotomy}
  Let $\Omc$ be an ontology formulated in one of uGF$(1)$,
  uGF$^{-}(1,=)$, uGF$^{-}_{2}(2)$, uGC$^{-}_{2}(1,=)$, or an
  $\mathcal{ALCHIF}$ ontology of depth~2. Then the following
  conditions are equivalent (unless $\PTime = \text{\sc NP}$):
\begin{enumerate}
\item $\Omc$ is materializable; 
\item $\Omc$ is materializable for the class of cg-tree decomposable instances $\Dmf$
with $\text{sig}(\Dmf)\subseteq \text{sig}(\Omc)$;
\item $\Omc$ is unravelling tolerant;
\item query evaluation w.r.t.~$\Omc$ is Datalog$^{\not=}$-rewritable \\
  (and Datalog-rewritable if \Omc is formulated in uGF);
\item query evaluation w.r.t.~$\Omc$ is in {\sc PTime}.
\end{enumerate}
Otherwise, query evaluation w.r.t.~$\Omc$ is {\sc coNP}-hard.
\end{theorem}
\begin{proof}
  (1) $\Rightarrow$
  (2) is not difficult to establish by a compactness argument.
%
  (2) $\Rightarrow$
  (3) is Theorem~\ref{thm:infiniteunravel}.  (3) $\Rightarrow$
  (4) is Theorem~\ref{thm:datalog}. (4) $\Rightarrow$
  (5) is folklore. (5) $\Rightarrow$
  (1) is Theorem~\ref{thm:nomatlower} (assuming $\PTime \neq \text{\sc NP}$).
\end{proof}
The qualification `with $\text{sig}(\Dmf)
\subseteq
\text{sig}(\Omc)$' in Point~2 of Theorem~\ref{thm:dichotomy} can be
dropped without compromising the correctness of the theorem, and the
same is true for Theorem~\ref{thm:infiniteunravel}. It will be useful,
though, in the decision procedures developed in Section~\ref{sect:decprob}.

\section{CSP-hardness}
We establish the four CSP-hardness results displayed in the middle
part of Figure~\ref{fig:results}, starting with a formal definition of
CSP-hardness. In addition, we derive from the existence of CSPs in \PTime that are not Datalog definable
the existence of ontologies in any of these languages with \PTime query evaluation that are not Datalog$^{\not=}$ rewritability.

Let $\Amf$ be an instance. The \emph{constraint satisfaction problem}
CSP$(\Amf)$ is to decide, given an instance $\Dmf$, whether there is a
homomorphism from $\Dmf$ to $\Amf$, which we denote with $\Dmf
\rightarrow \Amf$. In this context, $\Amf$ is called the
\emph{template} of CSP$(\Amf)$.  We will generally and w.l.o.g.\
assume that relations in $\text{sig}(\Amf)$ have arity at most two and
that the template $\Amf$ \emph{admits precoloring},
that is, for each $a \in \mn{dom}(\Amf)$, there is a unary relation
symbol $P_{a}$ such that $P_{a}(b)\in \Amf$ iff $b=a$
\cite{ComplexityConstraintLanguages}. It is known that for every
template $\Amf$, there is a template $\Amf'$ of this form such that
$\text{CSP}(\Amf)$ is polynomially equivalent to
$\text{CSP}(\Amf')$~\cite{DBLP:journals/tcs/LaroseT09}. We use
coCSP$(\Amf)$ to denote the complement of CSP$(\Amf)$.
\begin{definition}\label{def:CSPhard}
Let $\Lmc$ be an ontology language and $\Qmc$ a class of queries. Then $\Qmc$-evaluation
w.r.t.~$\Lmc$ is \emph{CSP-hard} if for every template $\Amf$, there exists an $\Lmc$ ontology $\Omc$ such that
\begin{enumerate}
\item there is a $q\in \Qmc$ such that coCSP$(\Amf)$ polynomially
  reduces to evaluating the OMQ $(\Omc,q)$ and
\item for every $q\in \Qmc$, evaluating the OMQ $(\Omc,q)$ is polynomially reducible
to coCSP$(\Amf)$.
\end{enumerate}
\end{definition}
It can be verified that a dichotomy between {\sc PTime} and {\sc coNP}
for \Qmc-evaluation w.r.t.~\Lmc ontologies implies a dichotomy between
{\sc PTime} and {\sc NP} for CSPs, a notorious open problem known as
the Feder-Vardi conjecture
\cite{FederVardi,DBLP:journals/siglog/Barto14}, when \Qmc-evaluation
w.r.t.\ \Lmc is CSP-hard. As noted in the introduction, a tentative
proof of the conjecture has recently been announced, but at the time
this article is published, its status still remains unclear.

\smallskip

The following theorem summarizes our results on CSP-hardness. We
formulate it for CQs, but remark that due to Theorem~\ref{thm:PTime}, a
dichotomy between {\sc PTime} and {\sc coNP} for any of the
mentioned ontology languages and any query language from the set
$\{ \text{rAQ}, \text{CQ}, \text{UCQ} \}$ implies the Feder-Vardi conjecture.
\begin{theorem}\label{thm:csphard}
  For any of the following ontology languages, CQ-evaluation w.r.t.\
  \Lmc is CSP-hard: uGF$_{2}(1,=)$, uGF$_{2}(2)$, uGF$_2(1,f)$, and
  the class of $\mathcal{ALCF}_{\ell}$ ontologies of depth~2.
\end{theorem}
\begin{proof}
  We sketch the proof for uGF$_2(1,=)$ and then indicate the
  modifications needed for uGF$_2(1,f)$ and $\mathcal{ALCF}_{\ell}$
  ontologies of depth~2. For uGF$_{2}(2)$, the result follows from a
  corresponding result in \cite{KR12-csp} for $\mathcal{ALC}$
  ontologies of depth~3.

  Let $\Amf$ be a template and assume w.l.o.g.\ that \Amf admits
  precoloring. Let $R_a$ be a binary relation for each $a\in
  \text{dom}(\Amf)$, and set
\begin{eqnarray*}
\varphi_{a}^{\not=}(x) & = & \exists y(R_{a}(x,y) \wedge \neg (x=y))\\
\varphi_{a}^{=}(x) & = & \exists y(R_{a}(x,y) \wedge (x=y))
\end{eqnarray*}
Then $\Omc$ contains
$$
\begin{array}{ll}
\multicolumn{2}{l}{\displaystyle\forall x (\bigwedge_{a\not=a'}\neg(\varphi_{a}^{\not=}(x)\wedge \varphi_{a'}^{\not=}(x))
\wedge \bigvee_{a}\varphi_{a}^{\not=}(x))}\\
\forall x (A(x) \rightarrow \neg \varphi_{a}^{\not=}(x)) & \text{when $A(a)\not\in \Amf$}\\
\forall xy (R(x,y) \rightarrow \neg (\varphi_{a}^{\not=}(x) \wedge
\varphi_{a'}^{\not=}(y))) & \text{when $R(a,a')\not\in \Amf$}\\[0.6mm]
\forall x \varphi_{a}^{=}(x) & \text{for all $a\in \text{dom}(\Amf)$}
\end{array}
$$
where $A$ and $R$ range over symbols in $\text{sig}(\Amf)$ of the
respective arity. A formula $\varphi_{a}^{\not=}(x)$ being true at a
constant $c$ in an instance $\Dmf$ means that $c$ is mapped to $a\in
\text{dom}(\Amf)$ by a homomorphism from \Dmf to \Amf. The first
sentence in \Omc thus ensures that every node in $\Dmf$ is mapped to
exactly one node in $\Amf$ and the second and third set of sentences
ensure that we indeed obtain a homomorphism. The last set of sentences
enforces that $\varphi_{a}^{=}(x)$ is true at every constant $c$. This
makes the disjunction in the first sentence `invisible' to the query
(in which inequality is not available), thus avoiding that \Omc is
{\sc coNP}-hard for trivial reasons.
%
%
In the long version, we show that \Omc satisfies Conditions~1 and~2 from
Definition~\ref{def:CSPhard} where there query $q$ used in Condition~1
is $q\leftarrow N(x)$ with $N$ a fresh unary relation.

For uGF$_2(1,f)$, state that a binary relation $F$ is a function and
that $\forall x F(x,x)$.  Now replace in $\Omc$ the formulas
$\varphi_{a}^{\not=}(x)$ by $\exists y (R_{a}(x,y) \wedge \neg
F(x,y))$ and $\varphi_{a}^{=}(x)$ by $\exists y (R_{a}(x,y) \wedge
F(x,y))$.

For $\mathcal{ALCF}_{\ell}$ of depth~2, replace in $\Omc$ the formulas
$\varphi_{a}^{\not=}(x)$ by $\exists^{\geq 2}y R_{a}(x,y)$ and $\varphi_{a}^{=}(x)$ by $\exists y R_{a}(x,y)$. The resulting
ontology is equivalent to a $\mathcal{ALCF}_{\ell}$ ontology of depth~2.
\end{proof}
It is known that for some templates $\Amf$, CSP$(\Amf)$ is in {\sc
  PTime} while coCSP$(\Amf)$ is not
Datalog$^{\not=}$-definable~\cite{FederVardi}.  Then CQ-evaluation
w.r.t.~the ontologies $\Omc$ constructed from \Amf in the proof of
Theorem~\ref{thm:PTime} is in {\sc PTime}, but not
Datalog$^{\not=}$-rewritable.
\begin{theorem}\label{thm:datalognot}
In any of the following ontology languages $\Lmc$ there exist ontologies
with \PTime CQ-evaluation which are not Datalog$^{\not=}$-rewritable:
uGF$_{2}(1,=)$, uGF$_{2}(2)$, uGF$_2(1,f)$, and the class of $\mathcal{ALCF}_{\ell}$ ontologies of depth~2.
\end{theorem}
The ontology languages in
Theorem~\ref{thm:datalog} thus behave provably different from the
languages for which we proved a dichotomy in
Section~\ref{sect:dichotomies}, since there {\sc PTime} query evaluation and
Datalog$^{\not=}$-rewritability coincide.

\section{Non-Dichotomy and Undecidability}
\label{sect:nodicho}
We show that ontology languages that admit sentences of depth 2 as
well as functions symbols tend to be computationally problematic as
they do neither enjoy a dichotomy between {\sc PTime} and {\sc coNP}
nor decidability of meta problems such as whether query evaluation
w.r.t.\ a given ontology \Omc is in \PTime,
Datalog$^{\not=}$-rewritable, or {\sc coNP}-hard, and whether \Omc is
materializable. We actually start with these undecidability results.
%
\begin{theorem}
\label{thm:undecidability}
For the ontology languages uGF$^{-}_{2}(2,f)$ and
$\mathcal{ALCIF}_{\ell}$ of depth 2, it is undecidable whether for a
given ontology~\Omc,
\begin{enumerate}

\item query evaluation w.r.t.\ \Omc is in \PTime,
  Datalog$^{\not=}$-rewritable, or {\sc coNP}-hard
  (unless $\text{\sc PTime}=\text{\sc NP}$);

\item \Omc is materializable.

\end{enumerate}
%
 \end{theorem}
\begin{proof}
  The proof is by reduction of the undecidable finite rectangle tiling
  problem. To establish both Points~1 and~2, it suffices to exhibit,
  for any such tiling problem \Pmf, an ontology $\Omc_\Pmf$ such that
  if \Pmf admits a tiling, then $\Omc_\Pmf$ is not materializable and
  thus query evaluation w.r.t.\ $\Omc_\Pmf$ is {\sc coNP}-hard and
  if \Pmf admits no tiling, then query evaluation w.r.t.\ $\Omc_\Pmf$
  is Datalog$^{\neq}$-rewritable and thus materializable (unless
  $\text{\sc PTime}=\text{\sc NP}$).

  The rectangle to be tiled is represented in input instances using
  the binary relations $X$ and $Y$, and $\Omc_\Pmf$ declares these
  relations and their inverses to be functional. The main idea in the
  construction of $\Omc_\Pmf$ is to verify the existence of a properly
  tiled grid in the input instance by propagating markers from the top
  right corner to the lower left corner. During the propagation, one
  makes sure that grid cells close (that is, the XY-successor
  coincides with the YX-successor) and that there is a tiling that
  satisfies the constraints in \Pmf. Once the existence of a properly
  tiled grid is completed, a disjunction is derived by $\Omc_\Pmf$ to
  achieve non-materializability and {\sc coNP}-hardness. The challenge
  is to implement this construction such that when \Pmf has no
  solution (and thus the verification of a properly tiled grid can
  never complete), $\Omc_\Pmf$ is Datalog$^{\neq}$-rewritable. In
  fact, achieving this involves a lot of technical subtleties.

  A central issue is how to implement the markers (as formulas with
  one free variable) that are propagated through the grid during the
  verification.  The markers must be designed in a way so that they
  cannot be `preset' in the input instance as this would make it
  possible to prevent the verification of a (possibly defective) part
  of the input. In $\mathcal{ALCIF}_{\ell}$, we use formulas of the form
  $\exists^{=1} y P(x,y)$ while additionally stating in $\Omc_\Pmf$
  that $\forall x \exists y P(x,y)$. Thus, the choice is only between
  whether a constant has exactly one $P$-successor (which means that
  the marker is set) or more than one $P$-successor (which means that
  the marker is not set). Clearly, this difference is invisible to
  queries and we cannot preset a marker in an input instance in the
  sense that we make it true at some constant. We can, however, easily
  make the marker false at a constant $c$ by adding two $P$-successors
  to $c$ in the input instance. It seems that this effect, which gives
  rise to many technical complications, can only be avoided by
  using marker formulas with higher quantifier depth which would
  result in $\Omc_\Pmf$ not falling within $\mathcal{ALCIF}_{\ell}$ depth~2.
	For uGF$_{2}^{-}(2,f)$ we work with $\neg \exists y (P(x,y) \wedge \neg
  F(x,y))$, where $F$ is a function for which we state $\forall x F(x,x)$ (as in the CSP encoding).

  Full proof details can be found in the long version. We only mention
  that closing of a grid cell is verified by using marker formulas as
  second-order variables.
\end{proof}
\begin{theorem}
\label{thm:nondichotomy}
For the ontology languages uGF$^{-}_{2}(2,f)$ and
$\mathcal{ALCIF}_{\ell}$ of depth 2, there is no dichotomy between
\PTime and {\sc coNP} (unless $\text{\sc PTime}=\text{\sc coNP}$).
\end{theorem}
By Ladner's theorem~\cite{DBLP:journals/jacm/Ladner75}, there is a
non-deterministic polynomial time Turing machine (TM) whose word
problem is neither in $\PTime$ nor $\NPTime$-hard (unless
$\text{\sc PTime}=\text{\sc coNP}$). Ideally, we would like to reduce
the word problem of such TMs to prove Theorem~\ref{thm:nondichotomy}.
However, this does not appear to be easily possible, for the following
reason. In the reduction, we use a grid construction and marker
formulas as in the proof of Theorem~\ref{thm:undecidability}, with the
grid providing the space in which the run of the TM is simulated and
markers representing TM states and tape symbols. We cannot avoid that
the markers can be preset either positively or negatively in the input
(depending on the marker formulas we choose), which means that some
parts of the run are not `free', but might be predetermined or at
least constrained in some way. We solve this problem by first
establishing an appropriate variation of Ladner's theorem.

We consider non-deterministic TMs $M$ with a single one-sided infinite
tape.  Configurations of $M$ are represented by strings $vqw$, where
$q$ is the state, and $v$ and $w$ are the contents of the tape to the
left and to the right of the tape head, respectively.  A \emph{partial
  configuration} of $M$ is obtained from a configuration $\gamma$ of
$M$ by replacing some or all symbols of $\gamma$ by a wildcard symbol
$\star$.  A partial configuration $\tilde\gamma$ \emph{matches} a
configuration $\gamma$ if it has the same length and agrees with
$\gamma$ on all non-wildcard symbols.  A \emph{partial run} of $M$ is
a finite sequence $\tilde\gamma_0,\dotsc,\tilde\gamma_m$ of partial
configurations of $M$ of the same length.  It is a \emph{run} if each
$\tilde\gamma_i$ is a configuration, and it \emph{matches} a run
$\gamma_0,\dotsc,\gamma_n$ if $m = n$ and each $\tilde\gamma_i$
matches $\gamma_i$.  A run is accepting if its last configuration has
an accepting state. Note that runs need not start in any specific
configuration (unless specified by a partial run that they
extend). The \emph{run fitting problem} for $M$ is to decide
whether a given partial run of $M$ matches some accepting run of
$M$. It is easy to see that for any TM $M$, the run fitting
problem for $M$ is in {\sc NP}. We prove the following result in the
long version by a careful adaptation of the proof of Ladner's
theorem given in~\cite{DBLP:books/daglib/0023084}.
%
\begin{theorem}\label{intermediate}
  There is a non-de\-ter\-mi\-nistic Turing machine whose run fitting
  problem is 
  neither in $\PTIME$ nor $\NPTime$-hard (unless $\PTIME = \NPTime$).
\end{theorem}
Now Theorem~\ref{thm:nondichotomy} is a consequence of the following lemma.
\begin{lemma}
\label{lem:nondicho}
For every Turing machine $M$, there is a
uGF$^{-}_{2}(2,f)$ ontology $\Omc$ and an $\mathcal{ALCIF}_{\ell}$
ontology \Omc of depth 2 such that the following hold, where $N$
is a distinguished unary relation:
\begin{enumerate}
\item there is a polynomial reduction of the run fitting
problem for $M$ to the complement of evaluating the OMQ $(\Omc,q\leftarrow N(x))$;
\item for every UCQ $q$, evaluating the OMQ $(\Omc,q)$
is polynomially reducible to the complement of the run fitting problem for $M$.
\end{enumerate}
\end{lemma}
To establish Lemma~\ref{lem:nondicho}, we re-use the ontology
$\Omc_{\mathfrak{P}}$ from the proof of
Theorem~\ref{thm:undecidability}, using a trivial rectangle tiling
problem. When the existence of the grid has been verified, instead of
triggering a disjunction as before, we now start a simulation of $M$
on the grid.
For both $\mathcal{ALCIF}_{\ell}$ and uGF$^{-}_{2}(2,f)$, we
represent states $q$ and tape symbols $G$ using the same formulas
as in the CSP encoding of homomorphisms. Thus, for $\mathcal{ALCIF}_{\ell}$
we use formulas $\exists^{\geq 2}y q(x,y)$ and $\exists^{\geq 2} y G(x,y)$, respectively,
using $q$ and $G$ as binary relations.
Note that here the encoding $\exists^{=1}y q(x,y)$ from the tiling problem does not work
because states and tape symbols can be positively preset in the input instance rather
than negatively, which is in correspondence with the run fitting  problem.


\section{Decision Problems}
\label{sect:decprob}

We study the decidability and complexity of the problem to decide
whether a given ontology admits {\sc PTime} query evaluation.
Realistically, we can only hope for positive results in cases where
there is a dichotomy between \PTime and {\sc coNP}: first, we have
shown in Section~\ref{sect:nodicho} that for cases with provably no
such dichotomy, meta problems are typically undecidable; and second,
it does not seem very likely that in the CSP-hard cases, one can
decide whether an ontology admits {\sc PTime} query evaluation without
resolving the dichotomy question and thus solving the Feder-Vardi
conjecture. Our main results are {\sc ExpTime}-completeness of
deciding {\sc PTime}-query evalutation of $\mathcal{ALCHIQ}$
ontologies of depth one (the same complexity as for satisfiability)
and a {\sc NExpTime} upper bound for uGC$_{2}^{-}(1,=)$
ontologies. Note that, in both of the considered languages, {\sc
  PTime}-query evalutation coincides with rewritability into
Datalog$^{\neq}$.  We remind the reader that according to our
experiments, a large majority of real world ontologies are
$\mathcal{ALCHIQ}$ ontologies of depth~1.  We also show that for \ALC
ontologies of depth~2, the mentioned problem is \NExpTime-hard.

Since the ontology languages relevant here admit at most binary
relations, an interpretation \Bmf is cg-tree decomposable if and only
if the undirected graph $G_{\Bmf}=\{ \{a,b\}\mid R(a,b) \in \Bmf, a\not=b\}$ is a tree. For
simplicity, we speak of \emph{tree interpretations} and of \emph{tree instances}, defined likewise.
The \emph{outdegree of $\Bmf$} is the outdegree of~$G_{\Bmf}$.
\begin{theorem}\label{thm:exp}
  For uGC$^{-}_{2}(1,=)$ ontologies, deciding whether query evaluation
  w.r.t.\ a given ontology is in {\sc PTime} (equivalently: rewritable
  into Datalog$^{\neq}$) is in {\sc NExpTime}. For $\mathcal{ALCHIQ}$
  ontologies of depth~1, this problem is in {\sc ExpTime}-complete.
\end{theorem}
The main insight underlying the proof of Theorem~\ref{thm:exp} is
that for ontologies formulated in the mentioned languages,
materializability (which by Theorem~\ref{thm:dichotomy} coincides with
\PTime query evaluation) already follows from the existence of
materializations for tree instances of depth~1. We make this precise
in the following lemma.
Given a tree interpretation $\Bmf$ and $a\in \text{dom}(\Bmf)$, define
the \emph{1-neighbourhood} $\Bmf_{a}^{\leq 1}$ of $a$ in $\Bmf$ as $\Bmf_{|X}$,
where $X$ is the union of all guarded sets in $\Bmf$ that contain $a$.
$\Bmf$ is a \emph{bouquet with root $a$} if $\Bmf_{a}^{\leq 1} = \Bmf$
and it is \emph{irreflexive} if there exists no atom of the form $R(b,b)$ in $\Bmf$.
%
\begin{lemma}
\label{lem:bouquetchar}
Let $\Omc$ be a uGC$_{2}^{-}(1,=)$ ontology (resp.\ an
$\mathcal{ALCHIQ}$ ontology of depth~1).  Then $\Omc$ is
materializable iff $\Omc$ is materializable for the class of all
(respectively, all irreflexive) bouquets $\Dmf$ of outdegree
$\leq |\Omc|$ with $\text{sig}(\Dmf)\subseteq \text{sig}(\Omc)$.
%
\end{lemma}
\begin{proof}
We require some notation. An instance $\Dmf$ is called \emph{$\Omc$-saturated} for an ontology $\Omc$
if for all facts $R(\vec{a})$ with $\vec{a}\subseteq \text{dom}(\Dmf)$
such that $\Omc,\Dmf\models R(\vec{a})$ it follows that $R(\vec{a})\in \Dmf$.
For every $\Omc$ and instance $\Dmf$ there exists a unique minimal (w.r.t.~set-inclusion)
$\Omc$-saturated instance $\Dmf_{\Omc}\supseteq \Dmf$.
We call $\Dmf_{\Omc}$ the \emph{$\Omc$-saturation of $\Dmf$}.
It is easy to see that there is a materialization of $\Omc$ and $\Dmf$ if and only if
there is a materialization of $\Omc$ and the $\Omc$-saturation of $\Dmf$.

We first prove Lemma~\ref{lem:bouquetchar} for uGC$_{2}^{-}(1,=)$ ontologies $\Omc$ and without the condition
on the outdegree. Let $\Sigma_{0}=\text{sig}(\Omc)$ and assume that $\Omc$ is
materializable for the class of all $\Sigma_{0}$-bouquets.
By Theorem~\ref{thm:dichotomy} it suffices to prove that $\Omc$
is materializable for the class of $\Sigma_{0}$-tree instances. Fix a $\Sigma_{0}$-tree instance $\Dmf$ that is
consistent w.r.t.~$\Omc$. We may assume that $\Dmf$ is $\Omc$-saturated.
Note that a forest model materialization $\Bmf$ of an ontology $\Omc$ and an $\Omc$-saturated instance $\Fmf$
consists of $\Fmf$ and tree interpretations $\Bmf_{a}$, $a\in \text{dom}(\Fmf)$, that are hooked to $\Fmf$ at $a$.
Take for any $a\in \text{dom}(\Dmf)$ the bouquet $\Dmf^{\leq 1}_{a}$ with root $a$ and hook to $\Dmf$ at $a$ the interpretation
$\Bmf_{a}$ that is hooked to $\Dmf^{\leq 1}_{a}$ at $a$ in a forest model materialization $\Bmf$ of $\Dmf^{\leq 1}_{a}$
and $\Omc$ (such a forest model materialization exists since $\Dmf^{\leq 1}_{a}$ is materializable).
Denote by $\Amf$ the resulting interpretation.
Using the condition that $\Omc$ is a uGC$^{-}_{2}(1,=)$ ontology it is not difficult to prove
that $\Amf$ is a materialization of $\Omc$ and $\Dmf$.

We now prove the restriction on the outdegree. Assume $\Omc$ is given. Let $\Dmf$ be a bouquet with
root $a$ of minimal outdegree such that there is no materialization of
$\Omc$ and $\Dmf$. We show that the outdegree of $\Dmf$ does not exceed $|\Omc|$. Assume
the outdegree of $\Dmf$ is at least three (otherwise we are done). We may assume
that $\Dmf$ is $\Omc$-saturated.
Take for any formula $\chi=\exists^{\geq n}z_{1} \alpha(z_{1},z_{2}) \wedge \varphi(z_{1},z_{2})$ that
occurs as a subformula in $\Omc$ the set $Z_{\chi}$ of all $b\not=a$ such that
$\Dmf\models \alpha(b,a)\wedge \varphi(b,a)$. Let $Z_{\chi}'=Z_{\chi}$ if $|Z_{\chi}|\leq n+1$;
otherwise let $Z'_{\chi}$ be a subset of $Z_{\chi}$ of cardinality $n+1$.
Let $\Dmf'$ be the restriction $\Dmf_{|Z}$ of $\Dmf$ to the union $Z$ of all $Z_{\chi}'$ and $\{a\}$.
We show that there exists no materialization of $\Dmf'$ and $\Omc$. Assume for a proof
by contradiction that there is a materialization $\Bmf$ of $\Dmf'$. Let $\Bmf'$ be
the union of $\Dmf\cup \Bmf$ and the interpretations $\Bmf_{b}$, $b\in \text{dom}(\Dmf)\setminus (Z\cup\{a\})$,
that are hooked to $\Dmf_{|\{a,b\}}$ at $b$ in a forest model materialization of $\Dmf_{|\{a,b\}}$.
We show that $\Bmf'$ is a materialization of $\Dmf$ and $\Omc$ (and thus derive a contradiction).
Using the condition that $\Dmf$ is $\Omc$-saturated
one can show that the restriction $\Bmf'_{|\text{dom}(\Dmf)}$ of $\Bmf'$ to $\text{dom}(\Dmf)$ coincides with $\Dmf$.
Using the condition that $\Omc$ has depth~1 it is now easy to show that $\Bmf'$ is a model
of $\Omc$. It is a materialization of $\Dmf$ and $\Omc$ since it is composed of materializations of subinstances
of $\Dmf$ and~$\Omc$.

The proof that irreflexive bouquets are sufficient for ontologies of depth~1 in $\mathcal{ALCHIQ}$ is similar
to the proof above and uses the fact that one can
always unravel models of $\mathcal{ALCHIQ}$ ontologies into irreflexive tree models.
\end{proof}
We now develop algorithms that decide {\sc PTime} query evaluation by
checking the conditions given in Lemma~\ref{lem:bouquetchar}, starting
with the (easier) case of $\mathcal{ALCHIQ}$.  Let $\Dmf$ be a
bouquet with root $a$. Call a bouquet $\Bmf\supseteq \Dmf$ a
\emph{1-materialization of \Omc and $\Dmf$} if
\begin{itemize}
\item there exists a model $\Amf$ of $\Omc$ and $\Dmf$ such that
  $\Bmf= \Amf_{a}^{\leq 1}$;
\item for any model $\Amf$ of $\Dmf$ and $\Omc$ there exists a
  homomorphism from $\Bmf$ to $\Amf$ that preserves $\text{dom}(\Dmf)$.
\end{itemize}
It turns out that, when checking materializability, not
only is it sufficient to consider bouquets instead of unrestricted
instances, but additionally one can concentrate on 1-materializations
of bouquets.
\begin{lemma}\label{lem:1materialization}
  Let $\Omc$ be an $\mathcal{ALCHIQ}$ ontology of depth~1. If for all
  irreflexive bouquets $\mathfrak{D}$ that are consistent
  w.r.t.~$\Omc$, of outdegree $\leq |\Omc|$, and satisfy
  $\text{sig}(\Dmf) \subseteq \text{sig}(\Omc)$ there is a
  1-materialization of $\Omc$ and \Dmf, then $\Omc$ is materializable for the
  class of all such bouquets.
\end{lemma}
\begin{proof}
  For brevity, we call an irreflexive bouquet \Fmf \emph{relevant} if
  it is consistent w.r.t.~$\Omc$, of outdegree $\leq |\Omc|$ and
  satisfies $\text{sig}(\Fmf) \subseteq \text{sig}(\Omc)$.
%
  An \emph{irreflexive 1-materializability witness} $(\Fmf,a,\Bmf)$ consists of a
  relevant irreflexive bouquet $\Fmf$ with root $a$ and a
  1-materialization $\Bmf$ of $\Fmf$ w.r.t.~$\Omc$. One can show that
  $\Bmf$ is an irreflexive tree interpretation.

  Now let \Dmf be a relevant irreflexive bouquet with root $a$ and
  assume that \Dmf is 1-materializable w.r.t.\ \Omc. We have to show
  that there exists a materialization of $\Omc$ and $\Dmf$. Note that
  for every relevant irreflexive bouquet \Fmf, there is a 1-materializability witness
  $(\Fmf,a,\Bmf)$. We construct the desired materialization
  step-by-step using these pairs also memorizing sets of frontier
  elements that have to be expanded in the next step. We start with the
  irreflexive 1-materializability witness $(\Dmf,a,\Bmf)$ and set $\Bmf^{0}=\Bmf$ and
  $F_{0}=\text{dom}(\Bmf)\setminus\{a\}$.  Then we construct a
  sequence of irreflexive tree interpretations $\Bmf^{0}\subseteq
  \Bmf^{1}\subseteq \ldots$ and frontier sets $F_{i+1}\subseteq
  \text{dom}(\Bmf^{i+1})\setminus \text{dom}(\Bmf^{i})$ inductively as
  follows: given $\Bmf^{i}$ and $F_{i}$, take for any $b\in F_{i}$ its
  predecessor $a$ in $\Bmf^{i}$ and an irreflexive
  1-materializability witness $(\Bmf^{i}_{|\{a,b\}},b,\Bmf_{b})$ and set
$$
\Bmf^{i+1}   :=  \Bmf^{i}\cup \bigcup_{b\in F_{i}}\Bmf_{b} \qquad
F_{i+1}  :=  \bigcup_{b\in F_{i}} \text{dom}(\Bmf_{b})\setminus\{b\}
$$
Let $\Bmf^{\ast}$ be the union of all $\Bmf_{i}$. We show that $\Bmf$
is a materialization of $\Omc$ and $\Dmf$.  $\Bmf$ is a model of
$\Omc$ by construction since $\Omc$ is an $\mathcal{ALCHIQ}$ ontology
of depth~1. Consider a model $\Amf$ of $\Omc$ and $\Dmf$. It suffices
to construct a homomorphism $h$ from $\Bmf^{\ast}$ to $\Amf$ that
preserves $\text{dom}(\Dmf)$. We may assume that $\Amf$ is an
irreflexive tree interpretation.  We construct $h$ as the limit of a
sequence $h_{0},\ldots$ of homomorphisms from $\Bmf^{i}$ to $\Amf$.
By definition, there exists a homomorphism $h_{0}$ from $\Bmf^{0}$ to
$\Amf_{a}^{\leq 1}$ preserving $\text{dom}(\Dmf)$.  Now, inductively,
assume that $h_{i}$ is a homomorphism from $\Bmf^{i}$ to
$\Amf$. Assume $c$ has been added to $\Bmf^{i}$ in the construction of
$\Bmf^{i+1}$. Then there exists $b\in F_{i}$ and its predecessor $a$
in $\Bmf^{i}$ such that $c\in \text{dom}(\Bmf_{b})\setminus \{b\}$,
where $\Bmf_{b}$ is the irreflexive tree interpretation that has been
added to $\Bmf^{i}$ as the last component of the irreflexive model
pair $(\Bmf^{i}_{|\{a,b\}},b,\Bmf_{b})$. But then, as $\Bmf_{b}$ is a
1-materialization of $\Bmf^{i}_{|\{a,b\}}$ and $h_{i}$ is injective on
$\Bmf^{i}_{|\{a,b\}}$ (since $\Amf$ is irreflexive), we can expand the
homomorphism $h_{i}$ to a homomorphism to $\Amf$ with domain
$\text{dom}(\Bmf^{i})\cup \{c\}$. Thus, we can expand $h_{i}$ to a
homomorphism from $\Bmf^{i+1}$ to $\Amf$.
\end{proof}
Lemma~\ref{lem:bouquetchar} and Lemma~\ref{lem:1materialization} imply
that an $\mathcal{ALCHIQ}$ ontology $\Omc$ of depth~1 enjoys {\sc PTime} query evaluation
if and only if all irreflexive bouquets
$\mathfrak{D}$ that are consistent w.r.t.~$\Omc$, of outdegree $\leq
|\Omc|$, and satisfy $\text{sig}(\Dmf) \subseteq \text{sig}(\Omc)$
have a 1-materialization w.r.t.~$\Omc$. The latter condition can be
checked in deterministic exponential time since the satisfiability
problem for $\mathcal{ALCHIQ}$ ontologies is in {\sc
  ExpTime}. Moreover, there are only exponentially many relevant
bouquets. We have thus proved the \ExpTime upper bound in
Theorem~\ref{thm:exp}. A matching lower bound can be proved by a
straightforward reduction from satisfiability.

\smallskip

The following example shows that, in contrast to $\mathcal{ALCHIQ}$
depth~1, for uGC$_{2}^{-}(1,=)$ the existence of 1-materializations
does not guarantee materializability of bouquets.
\begin{example}\label{loop}
  We use $\exists^{\not=}yW(x,y)$ to abbreviate
  $\exists y (W(x,y) \wedge (x \not= y))$ and likewise for
  $\exists^{\not=}yW(y,x)$. Let $S,S',R,R'$ be binary relation symbols
  and $\Omc$ the uGF$_{2}^{-}(1,=)$ ontology $\Omc$ that
  contains
$$
\begin{array}{c}
\forall x \big(S(x,x) \rightarrow (R(x,x) \rightarrow
(\exists^{\not=}yR(x,y) \vee \exists^{\not=}yS(x,y)))\big) \\[1mm]
\forall x (\exists^{\not=}yW(y,x) \rightarrow \exists y W'(x,y))
\end{array}
$$
where $(W,W')$ range over $\{(R,R'),(S,S')\}$. Observe that for the instance $\Dmf=\{(S(a,a),R(a,a)\}$
and the Boolean UCQ
$$
q\leftarrow R'(x,y) \vee S'(x,y),
$$
we have $\Omc,\Dmf \models q$. Also, for
$q_{R}\leftarrow R'(x,y)$ and $q_{S}\leftarrow S'(x,y)$  we have $\Omc,\Dmf\not\models q_{R}$ and
$\Omc,\Dmf\not\models q_{S}$. Thus, $\Omc$ is not materializable. It is, however, easy to
show that for every bouquet $\Dmf$ there exists a 1-materialization of $\Dmf$ w.r.t.~$\Omc$.
\end{example}

In uGC$_{2}^{-}(1,=)$, we thus have to check unrestricted
materializability of bouquets, instead of 1-materializability.  In
fact, it suffices to consider materializations that are tree
interpretations. To decide the existence of such materialization, we
use a mosaic approach. In each mosaic piece, we essentially record a
1-neighborhood of the materialization, a 1-neighborhood of a model of
the bouquet and ontology, and a homomorphism from the former to the
latter.  We then identify certain conditions that characterize when a
set of mosaics can be assembled into a materialization in a way that
is similar to the model construction in the proof of
Lemma~\ref{lem:1materialization}. There are actually two different
kinds of mosaic pieces that we use, with one kind of piece explicitly
addressing reflexive loops which, as illustrated by
Example~\ref{loop}, are the reason why we cannot work with
1-materializations. The decision procedure then consists of guessing a
set of mosaics and verifying that the required conditions are
satisfied.  Details are in the long version.

\medskip

Theorem~\ref{thm:exp} only covers ontology languages of depth 1.  It
would be desirable to establish decidability also for ontology
languages of depth 2 that enjoy a dichotomy between {\sc PTime} and
{\sc coNP}, such as uGF$^{-}_{2}(2)$. The following example shows that
this requires more sophisticated techniques than those used above.  In
particular, materializability of bouquets does not imply
materializability.
\begin{example}
\label{exlast}
We give a family of \ALC-ontologies $(\Omc_n)_{n \geq 0}$ of depth~2
such that each $\Omc_n$ is materializable for the class of tree
interpretations of depth at most $2^n-1$ while it is not
materializable. The idea is that any instance \Dmf that witnesses
non-materializability of $\Omc_n$ must contain an $R$-chain of length
$2^n$, $R$ a binary relation. The presence of this chain is verified
by propagating a marker upwards along the chain. To avoid that \Omc is
1-materializable, we represent this marker by a universally quantified
formula and also hide some other unary predicates in the same way. For
each unary predicate $P$, let $H_P(x)$ denote the formula $\forall y
(S(x,y) \rightarrow P(y))$ and include in $\Omc_n$ the sentence
$\forall x \exists y (S(x,y) \wedge P(y))$. The remaining sentences in
$\Omc_n$ are:
  $$
  \begin{array}{r@{\;}c@{\;}l}
\overline{X}_1(x) \wedge \cdots \wedge \overline{X}_{n}(x)
  &\rightarrow& H_V(x)  \\[1mm]
\displaystyle
      X_i(x) \wedge \exists R. (X_i(y) \wedge X_j(y)) &\rightarrow& H_{\mn{ok}_i}(x)  \\[1mm]
\displaystyle
  \overline{X}_i(x) \wedge \exists R . (\overline{X}_i(y) \wedge
    X_j(y)) &\rightarrow& H_{\mn{ok}_i}(x)  
  \end{array}
  $$
  $$
  \begin{array}{r@{\;}c@{\;}l}
\displaystyle
  X_i(x) \wedge \exists R . (\overline{X}_i(y) \wedge X_1(y) \wedge
    \cdots \wedge X_{i-1}(y)) &\rightarrow& H_{\mn{ok}_i}(x)  \\[1mm]
\displaystyle
  \overline{X}_i(x) \wedge \exists R . (X_i(y) \wedge X_1(y) \wedge
    \cdots \wedge \overline{X}_{i-1}(y)) &\rightarrow& H_{\mn{ok}_i}(x)  \\[1mm]
  H_{\mn{ok}_1}(x) \wedge  \cdots \wedge H_{\mn{ok}_{n}}(x)
\wedge
  \exists R . H_V(y) &\rightarrow& H_V(x) \\[1mm]
  \exists R . X_i(y) \wedge \exists R . \overline{X}_i(y) &\rightarrow& \bot \\[1mm]
\multicolumn{3}{c}{X_1(x) \wedge \cdots \wedge X_{n}(x)
    \wedge H_V(x)
   \rightarrow  B_1(x) \vee B_2(x)}
  \end{array}
  $$
  where $x$ is universally quantified, $\exists R . \vp(y)$ is an
  abbreviation for $\exists y (R(x,y) \wedge \vp(y))$, and $i$ ranges
  over 1..$n$. 
  Note that $X_1,\dots,X_n$ and $\overline{X}_1,\dots,\overline{X_n}$
  represent a binary counter and that lines two to five implement
  incrementation of this counter.
  The second last formula is necessary to
  avoid that multiple successors of a node interact in undesired ways.
  On instances that contain no $R$-chain of length $2^n$, a
  materialization
  can be constructed by a straightforward chase procedure.
%
\end{example}
We also observe that the ideas from Example~\ref{exlast} gives rise
to a {\sc NExpTime} lower bound.
\begin{theorem}
\label{thm:complexitydepthtwo}
For $\mathcal{ALC}$ ontologies of depth~2, deciding whether
query-evaluation is in {\sc PTime} is {\sc NExpTime}-hard (unless
$\text{\sc Ptime}=\text{\sc coNP}$).
\end{theorem}
We remark that the decidability of {\sc PTime} query evaluation of
\ALC ontologies of depth~2 remains open.

\section{Conclusion}
Perhaps the most surprizing result of our analysis is that it is
possible to escape Ladner's Theorem and prove a {\sc PTime}/{\sc coNP}
dichotomy for query evaluation for rather large subsets of the guarded
fragment that cover almost all practically relevant DL
ontologies. This result comes with a characterization of {\sc PTime}
query evaluation in terms of materializability and unravelling
tolerance, with the guarantee that {\sc PTime} query evaluation
coincides with Datalog$^{\neq}$-rewritability, and with decidability
of (and complexity results for) meta problems such as deciding whether
a given ontology enjoys {\sc PTime} query evaluation.
Our study also shows that when we increase the expressive power in
seemingly harmless ways, then often there is provably no {\sc
  PTime}/{\sc coNP} dichotomy or one obtains CSP-hardness.
%
The proof of the non-dichotomy results comes with a variation of
Ladner's Theorem that could prove useful in other contexts where some
form of precoloring of the input is unavoidable, such as in consistent
query answering \cite{DBLP:conf/icdt/LutzW15}.

There are a number of interesting future research questions. The main
open questions regarding dichotomies are whether the {\sc PTime}/{\sc
  coNP} dichotomy can be generalized from uGF$_{2}^{-}(2)$ to
uGF$^{-}(2)$ and whether the CSP-hardness results can be sharpened to
either CSP-equivalence results (this is known for $\mathcal{ALC}$
ontologies of depth 3 \cite{KR12-csp}) or to non-dichotomy results. 
Also of interest is the complexity of deciding {\sc PTime} query
evaluation for uGF$(1)$, where the characterization of {\sc PTime}
query evaluation via hom-universal models fails. Improving our current
complexity results to tight complexity bounds for {\sc Ptime} query
evaluation for $\mathcal{ALCHIF}$ ontologies of depth~2 and
uGF$_{2}^{-}(2)$ ontologies appears to be challenging as well. It
would also be interesting to study the case where invariance under
disjoint union is not guaranteed (as we have observed, the
complexities of CQ and UCQ evaluation might then diverge), and to add
the ability to declare in an ontology that a binary relation is transitive.

\section{Acknowledgments}
Andr{\'e} Hernich, Fabio Papacchini, and Frank Wolter were supported
by EPSRC UK grant EP/M012646/1. Carsten Lutz was supported by ERC CoG
647289 CODA.


\clearpage

\appendix

\section{Introduction to Description Logic}
We give a brief introduction to the syntax and semantics of DLs and establish their relationship
to the guarded fragment of FO. We consider the DL $\mathcal{ALC}$
and its extensions by inverse roles, role inclusions, qualified number
restrictions, functional roles, and local functionality. Recall that
\emph{$\mathcal{ALC}$-concepts} are constructed according to the rule
$$
C,D\; := \; \top \mid \bot \mid A \mid C\sqcap D \mid C \sqcup D \mid
\neg C \mid \exists R.C \mid \forall R . C
$$
where $A$ ranges over unary relations and $R$ ranges over binary
relations.
DLs extended by \emph{inverse roles} (denoted in the name of a DL
by the letter \Imc) admit, in addition, \emph{inverse relations}
denoted by $R^{-}$, with $R$ a relation. Thus, in $\mathcal{ALCI}$ inverse relations
can be used in place of relations in any $\mathcal{ALC}$ concept.
DLs extended by \emph{qualified number restrictions} (denoted by \Qmc) admit concepts of the form
$(\geq n\; R\; C)$, $(= n\; R\;C)$, and $(\leq n\;R\;C)$, where $n\geq
1$ is a natural number, $R$ is a relation or an inverse relation (if inverse relations are in the original DL), and $C$ is a concept.
When extending a DL with \emph{local functionality} (denoted by $\Fmc_\ell$)
one can use only number restrictions of the form $(\leq 1 \;R\;\top)$ in which
$R$ is a relation or an inverse relation (if inverse relations are in the original DL).
We abbreviate $(\leq 1 \;R\;\top)$ with $(\leq 1R)$
and use $(= 1R)$ as an abbreviation for
$(\exists R.\top) \sqcap (\leq 1 R)$ and $(\geq
2R)$ as an abbreviation for $(\exists R.\top) \sqcap \neg (\leq 1R)$.

In DLs, ontologies are formalized as finite sets of \emph{concept
  inclusions} $C\sqsubseteq D$, where $C,D$ are concepts in the
respective language.  We use $C \equiv D$ as an abbreviation for $C
\sqsubseteq D$ and $D\sqsubseteq C$.  In the DLs extended with \emph{functionality}
(denoted by \Fmc) one
can use \emph{functionality assertions} ${\sf func}(R)$, where
$R$ is a relation or an inverse relation (if present in the original DL).
Such an $R$ is interpreted as a partial function. Extending a DL with
\emph{role inclusions} (denoted by \Hmc) allows one to use expressions of the form $R \sqsubseteq
S$, where $R$ and $S$ are relations or inverse relations (if present in the
original DL), and which state that $R$ is a subset of $S$.

The semantics of DLs is given by interpretations $\Amf$. The
interpretation $C^{\Amf}$ of a concept $C$ in an interpretation $\Amf$
is defined inductively as follows:
\[
  \small
  \begin{array}{@{}r@{\;}c@{\;}l@{\ \ \ \ \ \ \ }r@{\;}c@{\;}l}
    \top^{\Amf} &=& \text{dom}(\Amf) &
    \bot^{\Amf} &=& \emptyset \\[1mm]
    A^{\Amf} &=& \{ a\in \text{dom}(\Amf) \mid A(a)\in \Amf\} &
    (\neg C)^{\Amf} &=& \text{dom}(\Amf)\setminus C^{\Amf}\\[1mm]
    (C \sqcap D)^{\Amf} &=&  C^{\Amf} \cap D^{\Amf} &
    (C \sqcup D)^{\Amf} &=&  C^{\Amf} \cup D^{\Amf} \\[1mm]
    \multicolumn{6}{c}{(\exists
      R.C)^{\Amf} = \{ a\in \text{dom}(\Amf) \mid \exists
      a':\; R(a,a')\in \Amf \mbox{ and } a'\in C^{\Amf}\}}\\[1mm]
    \multicolumn{6}{c}{\hspace{.5cm}(\forall R.C)^{\Amf} = \{ a\in
      \text{dom}(\Amf) \mid \forall a':\; R(a,a')\in \Amf \mbox{
        implies } a'\in A^{\Amf}\}}\\[1mm]
    \multicolumn{6}{c}{(\geq n\;R\;C)^{\Amf} = \{ a\in \text{dom}(\Amf) \mid \lvert\{ b
      \mid R(a,b)\in \Amf \mbox{ and } b \in C^{\Amf}\}\rvert\geq
      n\}}\\[1mm]
    \multicolumn{6}{c}{(\leq n\;R\;C)^{\Amf} = \{ a\in \text{dom}(\Amf) \mid \lvert\{ b
      \mid R(a,b)\in \Amf \mbox{ and } b \in C^{\Amf}\}\rvert\leq
      n\}}\\[1mm]
    \multicolumn{6}{c}{(= n\;R\;C)^{\Amf} = \{ a\in \text{dom}(\Amf) \mid \lvert\{ b
      \mid R(a,b)\in \Amf \mbox{ and } b \in C^{\Amf}\}\rvert =
      n\}}
  \end{array}
\]

\noindent
Then $\Amf$ \emph{satisfies} a
concept inclusion $C \sqsubseteq D$ if $C^\Amf \subseteq D^\Amf$. Alternatively, one
can define the semantics of DLs by translating them into FO;
the following table gives such a translation:
\[
  \small
  \begin{array}{@{}r@{\;}c@{\;}lr@{\;}c@{\;}l}
    \top^*(x) &=& \top &   \bot^*(x) &=& \bot
    \\[1mm]
    A^*(x) &=& A(x) &   (\neg C)^*(x) &=& \neg (C^*(x))
    \\[1mm]
    (C \sqcap D)^*(x) &=& C^*(x) \wedge D^*(x) &
    (C \sqcup D)^*(x) &=& C^*(x) \vee D^*(x) \\[1mm]
    (\exists R . C)^*(x) &=& \multicolumn{4}{l}{\exists y \, (R(x,y)
      \wedge C^*(y))} \\[1mm]
    (\forall R . C)^*(x) &=& \multicolumn{4}{l}{\forall y \, (R(x,y) \rightarrow C^*(y))}\\[1mm]
    (\geq n\;R\;C)^*(x) &=& \multicolumn{4}{l}{\exists^{\geq
      n}y(R(x,y)\land C^*(y))}
  \end{array}
\]
We observe the following relationships between DLs and fragments of the guarded
fragment. For a DL $\Lmc$ and fragment $\Lmc'$ of the guarded fragment we say that an \emph{$\Lmc$ ontology $\Omc$
can be written as an $\Lmc'$ ontology} if the translation given above translates $\Omc$
into an $\Lmc'$ ontology.
\begin{lemma} The following inclusions hold:
\begin{enumerate}
\item Every $\mathcal{ALCHI}$ ontology can be written as a uGF$_{2}$ ontology.
If the ontology has depth~2, then it can be written as a uGF$^{-}_{2}(2)$ ontology.
\item Every $\mathcal{ALCHIF}$ ontology can be written as a uGF$_{2}^{-}(f)$ ontology.
\item Every $\mathcal{ALCHIQ}$ ontology can be written as a uGC$_{2}$ ontology. If the ontology has depth~1,
then it can be written as a uGC$^{-}_{2}(1)$ ontology.
\end{enumerate}
\end{lemma}
\section{Introduction to Datalog}
We give a brief introduction to the notation used for Datalog. A \emph{datalog$^{\not=}$ rule} $\rho$ takes the form
$$
S(\vec{x}) \leftarrow R_{1}(\vec{x}_{1})\wedge \cdots \wedge R_{m}(\vec{x}_{m})
$$
where $S$ is a relation symbol, $m\geq 1$, and $R_{1},\ldots, R_{m}$ are either relation symbols or the symbol $\not=$ for inequality.
We call $S(\vec{x})$ the \emph{head} of $\rho$ and $R_{1}(\vec{x}_{1})\wedge \cdots \wedge R_{m}(\vec{x}_{m})$ its \emph{body}.
Every variable in the head of $\rho$ is required to occur in its body. We call a datalog$^{\not=}$ rule that does
not use inequality a \emph{datalog rule}. A \emph{Datalog$^{\not=}$ program} is a finite set $\Pi$ of datalog$^{\not=}$ rules
with a selected \emph{goal relation symbol} \text{goal} that does not occur in rule bodies in $\Pi$ and only in \emph{goal rules}
of the form $\text{goal}(\vec{x})\leftarrow R_{1}(\vec{x}_{1})\wedge \cdots \wedge R_{m}(\vec{x}_{m})$.
The \emph{arity} of $\Pi$ is the arity of its goal relation. A \emph{Datalog program} is a Datalog$^{\not=}$ program
not using inequality.

For every instance $\Dmf$ and Datalog$^{\not=}$ program $\Pi$, we call a model $\Amf$ of $\Dmf$
a \emph{model of $\Pi$} if $\Amf$ is a model of all FO sentences
$\forall \vec{x}\forall \vec{x}_{1}\cdots \forall \vec{x}_{m}(R_{1}(\vec{x}_{1})\wedge \cdots \wedge R_{m}(\vec{x}_{m})\rightarrow S(\vec{x}))$
with $S(\vec{x}) \leftarrow R_{1}(\vec{x}_{1})\wedge \cdots \wedge R_{m}(\vec{x}_{m})\in \Pi$.
We set $\Dmf\models \Pi(\vec{a})$ if $\text{goal}(\vec{a})\in \Amf$ for all models $\Amf$ of $\Dmf$ and $\Pi$.


\clearpage

\section{Proofs for Section 2}
Following the notation introduced after Theorem~\ref{thm:inv} we denote the guarded fragment
with equality by GF$(=)$ and uGF with equality in non-guard positions by uGF$(=)$. In contrast,
GF and uGF denote the corresponding languages without equality in non-guard positions. We use
this notation in the formulation of Theorem~\ref{thm:inv} below.
\begin{trivlist}\item
  \textbf{Theorem~\ref{thm:inv} (restated)}~\itshape
(1) A sentence in GF$(=)$ is invariant under disjoint unions iff it is equivalent to a sentence in uGF$(=)$.

(2) A sentence in GF is invariant under disjoint unions iff it is equivalent to a sentence in uGF.
\end{trivlist}

\begin{proof}
We proof Point~(1). The proof of Point~(2) is essentially the same.
We first show that every sentence in GF$(=)$ is equivalent to a Boolean combination of sentences in uGF$(=)$:
assume a sentence $\varphi$ in GF$(=)$ is given.
First, we may assume that all subformulas~$\exists y(x=y\land
  \psi(x,y))$ of~$\varphi$ are rewritten as~$\neg \forall y(x=y \to
  \overline{\psi}(x,y))$. Then let~$\psi = \forall y (x=y\to \psi'(x,y))$ be a
  subformula of~$\varphi$ with free variable $x$. It follows
  that~$\psi \equiv \psi'(x,x)$. As $\psi'(x,x)$ is a GF$(=)$ formula, any
  occurrence of~$\psi$ in~$\varphi$ can be replaced
  by~$\psi'(x,x)$. Let~$\varphi'$ be the result of substituting all
  subformulas of the form of~$\psi$ in~$\varphi$ with the
  corresponding equivalent GF$(=)$ formula. It follows that~$\varphi \equiv
  \varphi'$ (note that~$\forall x\forall y (x=y \to \psi'(x,y))$ can
  be rewritten as $\forall x (x=x \to \forall y(x=y \to
  \psi'(x,y)))$).

  A sentence~$\psi$ in $\varphi$ is a \emph{simple sentence} if there
  is no strict subformula $\psi'$ of $\psi$ which is a sentence.
  For the second step, let $\varphi_0 = \varphi'$, and let
$$
\varphi_i
  = (\varphi_{i-1}[\psi/\top]\land \psi) \lor
  (\varphi_{i-i}[\psi/\bot]\land \neg\psi)
$$
for some simple sentence~$\psi$ of~$\varphi_{i-1}$ not already substituted in any
  step~$j < i-1$. As at each step any simple sentence is a GF$(=)$
  sentence with guard $R(\vec{x})$ or~$x=x$ and the guarded formula is
  an openGF formula, then any simple sentence is either the
  negation of a uGF$(=)$ sentence or a uGF$(=)$ sentence. Let~$\varphi_n$ be the
  resulting sentence. It follows that~$\varphi \equiv \varphi_n$, and
  $\varphi_n$ is a Boolean combination of uGF$(=)$ sentences.

%

We return to the proof of Theorem~\ref{thm:inv}. It is easy to see that every sentence in uGF$(=)$
is invariant under disjoint unions. For the converse direction assume that a GF$(=)$ sentence $\varphi$ is
invariant under disjoint unions. Then $\varphi$ is equivalent to a Boolean combination of uGF$(=)$ sentences.
Let $\text{cons}(\varphi)$ be the set of all sentences $\chi$ in uGF$(=)$ with $\varphi \models \chi$.
By compactness of FO it is sufficient to show that $\text{cons}(\varphi)\models \varphi$.
If this is not the case, take a model $\Amf_{0}$ of $\text{cons}(\varphi)$ refuting $\varphi$
and take for any sentence $\psi$ in uGF$(=)$ that is not in $\text{cons}(\varphi)$ a model
$\Amf_{\neg\psi}$ satisfying $\varphi$ and refuting $\psi$. Let $\Amf_{1}$ be the disjoint union
of all $\Amf_{\neg \psi}$. By preservation of $\varphi$ under disjoint unions, $\Amf_{1}$ satisfies
$\varphi$. By reflection of $\varphi$ under disjoint unions, the disjoint union $\Amf$ of $\Amf_{0}$
and $\Amf_{1}$ does not satisfy $\varphi$. Thus $\Amf$ satisfies $\varphi$ and $\Amf_{1}$ does not satisfy $\varphi$ but
by construction $\Amf$ and $\Amf_{1}$ satisfy the same sentences in uGF$(=)$. This is impossible since $\varphi$ is equivalent
to a Boolean combination of uGF$(=)$ sentences.
\end{proof}
For the proof of Lemma~\ref{lem:forestmodel} we require the notion of guarded bisimulations
that characterizes the expressive power of GF$(=)$~\cite{DBLP:books/daglib/p/Gradel014}.
To apply guarded bisimulations to characterize the fragment
uGF$(=)$ of GF$(=)$ we modify the notion slightly by demanding that in the back-and-forth condition only
overlapping guarded sets are used. To cover uGC$_{2}(=)$ we also consider guarded bisimulations preserving
the number of guarded sets of the same type that contain a fixed singleton set.
Assume $\Amf$ and $\Bmf$ are interpretations.
A set $I$ of partial isomorphisms $p: \vec{a} \mapsto \vec{b}$ between guarded tuples $\vec{a}$ and $\vec{b}$
in $\Amf$ and $\Bmf$, respectively, is called a \emph{connected guarded bisimulations} if the following
holds for all $p: \vec{a} \mapsto \vec{b}\in I$:
\begin{itemize}
\item for every guarded tuple $\vec{a}'$ in $\Amf$ which has a non-empty intersection with $\vec{a}$ there exists
a guarded tuple $\vec{b}'$ in $\Bmf$ and $p': \vec{a}'\mapsto \vec{b}'\in I$ that coincides with $p$ on the
intersection of $\vec{a}$ and $\vec{a}'$.
\item for every guarded tuple $\vec{b}'$ in $\Bmf$ which has a non-empty intersection with $\vec{b}$ there exists
a guarded tuple $\vec{a}'$ in $\Amf$ and $p': \vec{a}'\mapsto \vec{b}'\in I$ that coincides with $p$ on the
intersection of $\vec{b}$ and $\vec{b}'$.
\end{itemize}
We say that $(\Amf,\vec{a})$ and $(\Bmf,\vec{b})$ are \emph{connected guarded bisimilar} if there exists
a connected guarded bisimulation between $\Amf$ and $\Bmf$ containing $\vec{a} \mapsto \vec{b}$.
\begin{theorem}\label{thm:guardedbisim}
If $(\Amf,\vec{a})$ and $(\Bmf,\vec{b})$ are connected guarded bisimilar and $\varphi(\vec{x})$ is a formula
in openGF, then $\Amf\models \varphi(\vec{a})$ iff $\Bmf\models \varphi(\vec{b})$.
\end{theorem}
For uGC$_{2}(=)$ the guarded sets have cardinality at most two. To preserve counting we require the following modified
version of the definition above. A set $I$ of partial isomorphisms $p: \vec{a} \mapsto \vec{b}$ between guarded tuples
$\vec{a}=(a_{1},a_{2})$ and $\vec{b}=(b_{1},b_{2})$ in $\Amf$ and $\Bmf$, respectively, is called a
\emph{counting connected guarded bisimulations} if the following
holds for all $p: (a_{1},a_{2}) \mapsto (b_{1},b_{2})\in I$:
\begin{itemize}
\item for every finite number of distinct guarded tuples $(a_{1},a_{2}')$ in $\Amf$ there exist
(at least) the same number of guarded tuples $(b_{1},b_{2}')$ in $\Bmf$ and
$p': (a_{1},a_{2}')\mapsto (b_{1},b_{2}')\in I$.
\item for every finite number of distinct guarded tuples $(a_{1}',a_{2})$ in $\Amf$ there exist
(at least) the same number of guarded tuples $(b_{1}',b_{2})$ in $\Bmf$ and $p': (a_{1}',a_{2})\mapsto (b_{1}',b_{2})\in I$.
\item for every finite number of distinct guarded tuples $(b_{1},b_{2})$ in $\Bmf$ there exist
(at least) the same number of guarded tuples $(a_{1},a_{2}')$ in $\Amf$ and $p': (a_{1},a_{2}')\mapsto (b_{1},b_{2}')\in I$.
\item for every finite number of distinct guarded tuples $(b_{1}',b_{2})$ in $\Bmf$ there exist
(at least) the same number of guarded tuples $(a_{1}',a_{2})$ in $\Amf$ and $p': (a_{1}',a_{2})\mapsto (b_{1}',b_{2})\in I$.
\end{itemize}
We say that $(\Amf,\vec{a})$ and $(\Bmf,\vec{b})$ are \emph{counting connected guarded bisimilar} if there exists
a counting connected guarded bisimulation between $\Amf$ and $\Bmf$ containing $\vec{a} \mapsto \vec{b}$.
\begin{theorem}\label{thm:countingguardedbisim}
If $(\Amf,\vec{a})$ and $(\Bmf,\vec{b})$ are counting connected guarded bisimilar and $\varphi(\vec{x})$ is a formula
in openGC$_{2}$, then $\Amf\models \varphi(\vec{a})$ iff $\Bmf\models \varphi(\vec{b})$.
\end{theorem}
\begin{trivlist}\item
  \textbf{Lemma~\ref{lem:forestmodel} (restated)}~\itshape
Let $\Omc$ be a uGF$(=)$ or uGC$_{2}(=)$ ontology, $\Dmf$ a possibly infinite instance, and $\Amf$ a model of $\Dmf$ and $\Omc$.
Then there exists a forest-model $\Bmf$ of $\Dmf$ and $\Omc$ such that there
exists a homomorphism $h$ from $\Bmf$ to $\Amf$ that preserves $\text{dom}(\Dmf)$.
%
\end{trivlist}
\begin{proof}
Assume first a uGF$(=)$ ontology $\Omc$, an instance $\Dmf$, and a model $\Amf$ of $\Omc$ and $\Dmf$ are
given. Let $G$ be a maximal guarded set in $\Dmf$. We unfold $\Amf$ at $G$ into a cg-tree decomposable interpretation
$\Bmf_{G}$ as follows: let $T(G)$ be the undirected tree with nodes
$t=G_{0}G_{1}\cdots G_{n}$, where $G_{i}$, $0\leq i \leq n$,
are maximal guarded sets in $\Amf$, $G_{0}=G$, $G_{1}\not\subseteq \text{dom}(\Dmf)$, and
(a) $G_{i}\not= G_{i+1}$, (b) $G_{i}\cap G_{i+1}\not=\emptyset$, and (c) $G_{i-1}\not= G_{i+1}$.
Let $(t,t')\in E$ if $t'=tF$ for some maximal guarded $F$. Set $\text{tail}(G_{0}\cdots G_{n})=G_{n}$.
Take an infinite supply of copies of any $d\in \text{dom}(\Amf)$. We assume all
copies are in $\Delta_{D}$. We set $e^{\uparrow}=d$ if $e$ is
a copy of $d$. Define instances $\text{bag}(t)$ for $t\in T(G)$ inductively as follows.
$\text{bag}(G)$ is an instance whose domain is a set of copies of elements $d\in G$ such that the
mapping $e \mapsto e^{\uparrow}$ is an isomorphism from $\text{bag}(G)$ onto $\Amf_{|G}$. Assume
$\text{bag}(t)$ has been defined, $\text{tail}(t)=F$, and $\text{bag}(t')$ has not yet been defined for
$t'=tF'\in T(G)$. Then take for any $d\in F'\setminus F$ a fresh copy $d'$ of $d$ and
define $\text{bag}(t')$ with domain $\{ d' \mid d\in F'\setminus F\}\cup
\{ e \in \text{bag}(t) \mid e^{\uparrow} \in F' \cap F)\}$
such that the mapping $e\mapsto e^{\uparrow}$ is an isomorphism from $\text{bag}(t')$ onto $\Amf_{|F'}$.

Let $\Bmf_{G}= \bigcup_{t\in T(G)}\text{bag}(t)$. Now hook $\Bmf_{G}$ to $\Dmf$ at $G$ by identifying
$\text{dom}(\text{bag}(G))$ with $G$ using the isomorphism $e\mapsto e^{\uparrow}$ and let $\Bmf$ be the union of
all $\Bmf_{G}$ with $G$ a maximal guarded set in $\Dmf$.
It is not difficult to prove the following properties of $\Bmf$:
\begin{enumerate}
\item $(T(G),E,\text{bag})$ is a cg-tree decomposition of $\Bmf_{G}$, for every maximal guarded
set $G$ in $\Dmf$;
\item The mapping $h:e\mapsto e^{\uparrow}$ is a homomorphism from $\Bmf$ to $\Amf$ which preserves $\text{dom}(\Dmf)$ and is
an isomorphism on each $\text{bag}(t)$ with $t\in T(G)$ for some maximal guarded set $G$ in $\Dmf$.
\item For any maximal guarded set $G=\{a_{1},\ldots,a_{k}\}$ in $\Dmf$, there is a connected guarded bisimulation between
$(\Bmf,(a_{1},\ldots,a_{k}))$ and $(\Amf,(a_{1},\ldots,a_{k}))$.
\end{enumerate}
It follows from Theorem~\ref{thm:guardedbisim} that $\Bmf$ is a model of $\Dmf$ and $\Omc$ and thus as required.

\medskip

Assume now that $\Omc$ is a uGC$_{2}(=)$ ontology, that $\Dmf$ is an instance, and that $\Amf$ is a model of $\Omc$
and $\Dmf$. We can assume that $\Dmf$ and $\Amf$ use unary and binary relation symbols only.
We have to preserve the number of guarded sets of a given type intersecting with a singleton
and therefore the unfolding is slightly different from the case without guarded counting quantifiers.
Let $c\in \text{dom}(\Dmf)$. We unfold $\Amf$ at $c$ into a cg-tree decomposable $\Bmf_{\{c\}}$.
Let $T(c)$ be the undirected tree with nodes $t=\{c\}G_{1}\cdots G_{n}$, where $G_{i}$, $1\leq i \leq n$,
are maximal guarded sets in $\Amf$, $c\in G_{1}$ and $G_{1}\not\subseteq \text{dom}(\Dmf)$, and
(a) $G_{i}\not= G_{i+1}$, (b) $G_{i}\cap G_{i+1}\not=\emptyset$, and (c) $G_{i-1}\cap G_{i}\not= G_{i+1}\cap G_{i}$.
Let $(t,t')\in E$ if $t'=tF$ for some maximal guarded $F$ in $\Dmf$.
Observe that Condition~(c) is stronger than the corresponding Condition~(c) in
the construction of forest models for uGF ontologies. The new condition ensures that
we do not introduce more successors of the same type when we unfold an interpretation.
Take an infinite supply of copies of any $d\in \text{dom}(\Amf)$. We set, as before, $e^{\uparrow}=d$ if $e$ is
a copy of $d$. Define instances $\text{bag}(t)$ for $t\in T(c)$ inductively as follows.
Let $c'$ be a copy of $c$ and set $\text{bag}(\{c\})= \{P(c') \mid P(c)\in \Amf\}$.
Assume $\text{bag}(t)$ has been defined, $\text{tail}(t)=F$, and $\text{bag}(t')$ has not yet been defined for
$t'=tF'\in T(c)$. Take for the unique $d\in F'\setminus F$ a fresh copy $d'$ of $d$ and
define $\text{bag}(t')$ with domain $\{ d' \}\cup \{ e \in \text{bag}(t) \mid e^{\uparrow} \in F' \cap F)\}$
such that the mapping $e\mapsto e^{\uparrow}$ is an isomorphism from $\text{bag}(t')$ onto $\Amf_{|F'}$.
Let $\Bmf_{\{c\}}= \bigcup_{t\in T(c)}\text{bag}(t)$. Now hook $\Bmf_{\{c\}}$ to $\Dmf$ at $\{c\}$
by identifying $c'$ with $c$ and let
$$
\Bmf = \{ R(\vec{a}) \in \Amf \mid \vec{a}\subseteq \text{dom}(\Dmf)\} \cup \bigcup_{c\in \text{dom}(\Dmf)}\Bmf_{\{c\}}.
$$
We can, of course, present the interpretation $\Bmf$ as an interpretation obtained from hooking interpretations
$\Bmf_{G}$ with $G$ maximal guarded sets to $\Dmf$ by choosing for each $c$ a single maximal guarded set $f(c)$
with $c\in f(c)$ and setting for $G=\{c_{1},c_{2}\}$,
$$
\Bmf_{G}= \{ R(\vec{a})\in \Amf \mid \vec{a} \subseteq G\} \cup \bigcup_{f(c)= G}\Bmf_{\{c\}}
$$
It is not difficult to prove the following properties of $\Bmf$:
\begin{enumerate}
\item $(T(c),E,\text{bag})$ is a guarded tree decomposition of $\Bmf_{\{c\}}$, for every $c\in \text{dom}(\Dmf)$;
\item The mapping $h:e\mapsto e^{\uparrow}$ is a homomorphism from $\Bmf$ to $\Amf$ which preserves $\text{dom}(\Dmf)$ and is
an isomorphism on each $\text{bag}(t)$ with $t\in T(c)$ for some $c\in \text{dom}(\Dmf)$.
\item For any maximal guarded set $G=\{a_{1},a_{2}\}$ in $\Dmf$, there is a counting connected guarded bisimulation between
$(\Bmf,(a_{1},a_{2}))$ and $(\Amf,(a_{1},a_{2}))$.
\end{enumerate}
It follows from Theorem~\ref{thm:countingguardedbisim} that $\Bmf$ is a model of $\Omc$ and $\Dmf$ and thus
as required.
\end{proof}

\section{Proofs for Section 3}
\begin{trivlist}\item
  \textbf{Theorem~\ref{thm:materializabilityeq} (restated)}~\itshape
Let \Omc be a uGF$(=)$ or uGC$_{2}(=)$ ontology and \Mmc a class of
instances. Then the following conditions are equivalent:
\begin{enumerate}
\item $\Omc$ is rAQ-materializable for \Mmc;
\item $\Omc$ is CQ-materializable for \Mmc;
\item $\Omc$ is UCQ-materializable for \Mmc.
\end{enumerate}
\end{trivlist}
\begin{proof}
The equivalence (2) $\Leftrightarrow$ (3) is straightforward. We show (1) $\Rightarrow$ (2).
Assume $\Omc$ is rAQ-materializable for $\Mmc$ and assume the instance $\Dmf\in \Mmc$ is consistent w.r.t.~$\Omc$.
Let $\Amf$ be a rAQ-materialization of $\Omc$ and $\Dmf$. Consider a forest model $\Bmf$ of $\Dmf$ and $\Omc$
such that there is a homomorphism from $\Bmf$ to $\Amf$ preserving $\text{dom}(\Dmf)$ (Lemma~\ref{lem:forestmodel}).
Recall that $\Bmf$ is obtained from $\Dmf$ by hooking cg-tree decomposable $\Bmf_{G}$ to
any maximal guarded set $G$ in $\Dmf$. We show that $\Bmf$ is a CQ-materialization of $\Omc$ and $\Dmf$.
To this end it is sufficient to prove that for any finite subinterpretation $\Bmf'$ of $\Bmf$ and any model
$\Amf'$ of $\Omc$ and $\Dmf$ there exists a homomorphism $h$ from $\Bmf'$ to $\Amf'$ that preserves
$\text{dom}(\Dmf) \cap \text{dom}(\Bmf')$. We may assume that for any maximal guarded set $G$ in $\Dmf$, if
$\text{dom}(\Bmf') \cap G\not=\emptyset$, then
$G \subseteq \text{dom}(\Bmf')$ and that $\Bmf'\cap \Bmf_{G}$ is connected if nonempty.
But then we can regard every $\Bmf' \cap \Bmf_{G}$ as an rAQ $q_{G}$ with answer variables
$G$. From $\Bmf\models q_{G}(\vec{b})$ for $\vec{b}$ a suitable tuple containing all $b\in G$
we obtain $\Amf\models q_{G}(\vec{b})$ since $\Bmf$ is a rAQ-materialization of
$\Dmf$ and $\Omc$. Let $h_{G}$ be the homomorphism that witnesses $\Bmf\models q_{G}(\vec{b})$.
Then $h_{G}$ is a homomorphism from $\Bmf' \cap \Bmf_{G}$ to $\Amf$ that preserves $G$.
The union $h$ of all $h_{G}$, $G$ a maximal guarded set in $\Dmf$, is the desired homomorphism
from $\Bmf'$ to $\Amf'$ preserving $\text{dom}(\Dmf) \cap \text{dom}(\Bmf')$.
\end{proof}

\begin{trivlist}\item
  \textbf{Lemma~\ref{lem:canhom} (restated)}~\itshape
A uGC$_{2}(=)$ ontology is materializable iff it admits hom-universal
models. This does not hold for uGF$(2)$ ontologies.
\end{trivlist}
\begin{proof}
The direction ``$\Leftarrow$'' is straightforward.
Conversely, assume that $\Omc$ is materializable. Let $\Dmf$ be an instance
that is consistent w.r.t.~$\Omc$. By Lemma~\ref{lem:forestmodel} there exists
a forest model $\Bmf$ of $\Dmf$ and $\Omc$ that is a CQ-materialization of $\Omc$
and $\Dmf$. Using a straightforward selective filtration procedure one can show that
there exists $\Bmf'\subseteq \Bmf$ such that $\Bmf'$ is a model of $\Omc$ and $\Dmf$
and for any guarded set $G$ the number of guarded sets
$G'$ with $G\cap G'\not=\emptyset$ is finite. We show that for any model $\Amf$
of $\Omc$ and $\Dmf$ there is a homomorphism from $\Bmf'$ into $\Amf$
that preserves $\mn{dom}(\Dmf)$. Let $\Amf$ be a model of $\Omc$ and $\Dmf$. Again we may
assume that $\Amf$ is a forest model such that for any guarded set $G$ the number of guarded sets
$G'$ with $G\cap G'\not=\emptyset$ is finite. Now, for any finite subset $F$ of $\text{dom}(\Bmf')$
there is a homomorphism $h_{F}$ from $\Bmf'_{|F}$ to $\Amf$ preserving $\text{dom}(\Dmf)\cap F$.
Let $F_{m}$ be the set of all $d\in \text{dom}(\Bmf')$ such that there is a sequence
of at most $m$ guarded sets $G_{0},\ldots,G_{m}$ with $G_{i}\cap G_{i+1}\not=\emptyset$ for $i<m$,
$G_{0}\cap \text{dom}(\Dmf)\not=\emptyset$,
and $d\in G_{m}$. Then each $F_{m}$ is finite and $\bigcup_{m\geq 0}F_{m}= \text{dom}(\Bmf')$.
Now we can construct the homomorphism $h$ from $\Bmf'$ to $\Amf$ as the limit of homomorphisms $h_{F_{n_{m}}}$
from $\Bmf'_{|F_{n_{m}}}$ to $\Amf$ preserving $\text{dom}(\Dmf)$ for an infinite sequence $n_{0}<n_{1}\cdots$
using a standard pigeonhole argument.

\medskip
We construct an ontology $\Omc$ that expresses that every element of a unary relation
$C$ is in the centre of a `partial wheel' represented by a ternary relation $W$.
The wheel can be generated by turning right or left. The two resulting models are
homomorphically incomparable but cannot be distinguished by answers to CQs. Thus, neither of the two
models is hom-universal but one can ensure that $\Omc$ is materializable.
As a first attempt to construct $\Omc$ we take unary relations $L$ (turn left) and $R$ (turn right),
the following sentence that says that one has to choose between $L$ and $R$ when generating the
wheel with cente $C$:
$$
\forall x \big(C(x) \rightarrow ((L(x) \vee R(x)) \wedge \exists y_{1},y_{2} W(x,y_{1},y_{2}))\big)
$$
and the following sentences that generates the wheel accordingly:
$$
\forall x,y,z \big(W(x,y,z) \rightarrow (L(x) \rightarrow  \exists z' W(x,z,z'))\big)
$$
$$
\forall x,y,z \big(W(x,y,z) \rightarrow (R(x) \rightarrow  \exists y' W(x,y',y))\big)
$$
Clearly this ontology does not admit hom-universal models. Note, however, that is also not
materializable since for the instance $\Dmf=\{C(a)\}$
we have $\Omc,\Dmf\models L(a) \vee R(a)$ but we neither have $\Omc,\Dmf\models L(a)$ nor
$\Omc,\Dmf\models R(a)$. The first step to solve this problem is to replace $L(x)$ by $\exists y
(\text{gen}(x,y) \wedge \neg L(x))$ and $R(x)$ by $\exists y (\text{gen}(x,y) \wedge \neg R(x))$
in the sentences above and also add $\forall x \exists y (\text{gen}(x,y) \wedge L(x))$ and
$\forall x \exists y (\text{gen}(x,y) \wedge R(x))$ to $\Omc$. Then a CQ `cannot detect' whether one
satisfies the disjunct $\exists y (\text{gen}(x,y) \wedge \neg R(x))$ or the disjunct
$\exists y (\text{gen}(x,y) \wedge \neg L(x))$ at a given node $x$ in $C$. Unfortunately, the
resulting ontology is still not materializable: if $W(a,b,c)\in \Dmf$ then CQs can detect whether
one introduces a $c'$ with $W(a,c,c')\in \Amf$ or a $b'$ with $W(a,b',b)\in \Amf$. To solve this problem
we ensure that a wheel has to be generated from $W(a,b,c)$ only if $b,c$ are not in the instance $\Dmf$.
In detail, we construct $\Omc$ as follows.
We use unary relations $A$, $L$, and $R$ and binary relations $\text{aux}$ and $\text{gen}$.
For a binary relation $S$ and unary relation $B$ we abbreviate
\begin{eqnarray*}
\forall S.B & := & \forall y (S(x,y) \rightarrow B(y))\\
\exists S.\neg B(x) & := & \exists y (S(x,y) \wedge \neg B(y))
\end{eqnarray*}
Now let $\Omc$ state that every node has $\text{aux}$-successors in $A$,
and $\text{gen}$-successors in $L$ and $R$:
$$
\forall x \exists y (\text{aux}(x,y) \wedge A(y))
$$
and
$$
\forall x \exists y (\text{gen}(x,y) \wedge L(y))
\quad \forall x \exists y (\text{gen}(x,y) \wedge R(y))
$$
We abbreviate
$$
W'(x,y,z) := W(x,y,z) \wedge \forall \text{aux}.A(y) \wedge \forall \text{aux}.A(z)
$$
Next we introduce the disjunction that determines whether one generates the wheel
by turning left or right, as indicated above:
$$
\forall x \big(C(x) \rightarrow (\exists \text{gen}.\neg L(x) \vee \exists \text{gen}.\neg R(x))\big)
$$
The following axiom starts the generation of the wheel. Observe that we use $W'(x,y_{1},y_{2})$ rather than $W(x,y_{1},y_{2})$
to ensure that new individuals are created for $y_{1}$ and $y_{2}$:
$$
\forall x \big(C(x) \rightarrow \exists y_{1},y_{2} W'(x,y_{1},y_{2})\big)
$$
Finally, we turn either left or right:
$$
\forall x,y,z \big(W'(x,y,z) \rightarrow (\exists \text{gen}.\neg L(x) \rightarrow  \exists z' W'(x,z,z'))\big)
$$
$$
\forall x,y,z \big(W'(x,y,z) \rightarrow (\exists \text{gen}.\neg R(x) \rightarrow  \exists y' W'(x,y',y))\big)
$$
This finishes the definition of $\Omc$. $\Omc$ is a uGF$(2)$ ontology. To show that it is materializable
observe that for any instance $\Dmf$ we can ensure that in a model $\Amf$ of $\Dmf$ and $\Omc$ for all $a\in \text{dom}(\Dmf)$ there
exists $b$ with $\text{aux}(a,b)\in \Amf$ and $A(b)\not\in \Amf$.
Thus, the wheel generation sentences apply only to $W(a,b,c)$ with $b,c\not\in \text{dom}(\Amf)$.
It is now easy to prove that CQ evaluation w.r.t.~$\Omc$ is in {\sc PTime}.
On the other hand, for $\Dmf=\{C(a)\}$ there does not exist a model $\Bmf$ of $\Omc$ and $\Dmf$
such that there is a homomorphism from \Bmf into every model of \Omc and \Dmf that preserves $\mn{dom}(\Dmf)$.
\end{proof}

In the proof of Lemma~\ref{thm:nomatlower} (and related proofs below), it will be more convenient to work with a certain
disjunction property instead of with materializability. We now
introduce this property and show the equivalence of the two notions.
Let $\Qmc$ be a class of non-Boolean connected CQs.
An ontology $\Omc$ has the \emph{$\Qmc$-disjunction property} if for all
instances $\Dmf$, queries
$q_{1}(\vec{x}_{1}),\ldots,q_{n}(\vec{x}_{n})\in \Qmc$ and $\vec{d}_{1},\ldots,\vec{d}_{n}$
in $\Dmf$ with $\vec{d}_{i}$ of the same length as $\vec{x}_{i}$:
if $\Omc,\Dmf\models q_{1}(\vec{d}_{1}) \vee \ldots \vee q_{n}(\vec{d}_{n})$,
then there exists $i\leq n$ such that $\Omc,\Dmf\models q_{i}(\vec{d}_{i})$.
\begin{theorem}
\label{thm:disjproperty}
Let $\Qmc$ be a class CQs and \Omc an FO$(=)$-ontology. Then \Omc is
$\Qmc$-materializable iff \Omc has the $\Qmc$-disjunction property.
\end{theorem}
\begin{proof}
  For the nontrivial ``$\Leftarrow$'' direction, let $\Dmf$ be an intance
  that is consistent w.r.t.~$\Omc$ and such that there is no
  $\Qmc$-materialization of $\Omc$ and $\Dmf$.  Then the set of FO-sentences $\Omc\cup \Dmf \cup
  \Gamma$ is not satisfiable, where
$$
    \Gamma = \{ \neg q(\vec{d}) \mid \Omc,\Dmf\not\models q(\vec{d}), \vec{d}\subseteq \text{dom}(\Dmf),
              q\in \Qmc\}
$$
In fact, any satisfying interpretation would be a $\Qmc$-materialization of $\Omc$.
By compactness, there is a finite subset $\Gamma'$ of $\Gamma$
such that $\Omc \cup \Dmf \cup \Gamma'$ is not satisfiable, i.e.
$$
\Omc,\Dmf \models \bigvee_{\neg q(\vec{d}) \in \Gamma'} q(\vec{d})
$$
%
By definition of $\Gamma'$, $\Omc,\Dmf\not\models q(\vec{d})$ for all $q(\vec{d}) \in \Gamma'$. Thus,
$\Omc$ lacks the $\Qmc$-disjunction property.
\end{proof}

\begin{trivlist}\item\textbf{Theorem~\ref{thm:nomatlower} (restated)}~\itshape
Let \Omc be an FO$(=)$-ontology that is invariant under disjoint unions.
If \Omc is not materializable, then rAQ-evaluation w.r.t.\ \Omc is
{\sc coNP}-hard.
\end{trivlist}
\begin{proof}[sketch]
  It was proved in \cite{KR12-csp} that if an \ALC ontology \Omc is not
  ELIQ-materializable, then ELIQ-evaluation w.r.t.\ \Omc is {\sc
    coNP}-hard, where an ELIQ is a rAQ $q(\vec{x})$ such that the
  associated instance $\Dmf_{q(\vec{x})}$ viewed as an undirected
  graph is a tree (instead of cg-tree decomposable) with a single
  answer variable at the root.\footnote{In the context of \ALC,
    relation symbols are at most binary and thus it should be clear
    what `tree' means.}  The proof is by reduction from 2+2-SAT, the
  variant of propositional satisfiability where the input is a set of
  clauses of the form $(p_1 \vee p_2 \vee \neg n_1 \vee \neg n_2)$,
  each $p_1,p_2,n_1,n_2$ a propositional letter or a truth constant
  \cite{Schaerf-93}. The proof of Theorem~\ref{thm:nomatlower} can be
  obtained from the proof in \cite{KR12-csp} by minor modifications, which
  we sketch in the following.

  The proof crucially exploits that if \Omc is not rAQ-materializable,
  then by Theorem~\ref{thm:disjproperty} it does not have the
  rAQ-disjunction property. In fact, we take an instance \Dmf, rAQs
  (resp.\ ELIQs) $q_{1}(\vec{x}_{1}),\ldots,q_{n}(\vec{x}_{n})\in
  \Qmc$, and elements $\vec{d}_{1},\ldots,\vec{d}_{n}$ of $\Dmf$ that
  witness failure of the disjunction property, copy them an
  appropriate number of times, and use the resulting set of gadgets to
  choose a truth value for the variables in the input 2+2-SAT
  formula. The fact that \Omc is invariant under disjoint unions
  ensures that the choice of truth values for different variables is
  independent.  A main difference between ELIQs and rAQs is that rAQs
  can have more than one answer variable. A straightforward way to
  handle this is to replace certain binary relations from the
  reduction in \cite{KR12-csp} with relations of higher arity (these are
  `fresh' relations introduced in the reduction, that is, they do not
  occur in \Omc). To deal with a rAQ of arity $k$, one would use a
  $k+1$-ary relation. However, with a tiny bit of extra effort (and if
  desired), one can replace these relations with $k$ binary
  relations.
  %
\end{proof}

We now turn to the proof of Theorem~\ref{thm:PTime}.

\begin{trivlist}\item\textbf{Theorem~\ref{thm:PTime} (restated)}~\itshape
  Let $\Omc$ be a $\uGF$ or $\uGC$ ontology.
  The following are equivalent:
  \begin{enumerate}
  \item\label{thm:PTime/PTime}
    UCQ-evaluation w.r.t.~$\Omc$ is in \PTime
    iff CQ-evaluation w.r.t.~$\Omc$ is in \PTime
    iff rAQ-evaluation w.r.t.~$\Omc$ is in \PTime.
  \item\label{thm:PTime/Datalog}
    UCQ-evaluation w.r.t.~$\Omc$ is Datalog$^{\neq}$-rewritable
    iff CQ-evaluation w.r.t.~$\Omc$ is Datalog$^{\neq}$-rewritable
    iff rAQ-evaluation w.r.t.~$\Omc$ is Datalog$^{\neq}$-rewritable.
    If \Omc is a uGF ontology,
    then the same is true for Datalog-rewritability.
  \item\label{thm:PTime/coNPhard}
    UCQ-evaluation w.r.t.~$\Omc$ is \coNP-hard
    iff CQ-evaluation w.r.t.~$\Omc$ is \coNP-hard
    iff rAQ-evaluation w.r.t.~$\Omc$ is \coNP-hard.
  \end{enumerate}
\end{trivlist}

The ``only if'' directions of the first two parts of Theorem~\ref{thm:PTime}
and the ``if'' direction of the third part are trivial.
For the non-trivial ``if'' directions of part~\ref{thm:PTime/PTime}
and~\ref{thm:PTime/Datalog}
we decompose UCQs into finite disjunctions of queries $q \wedge \bigwedge_i q_i$
where $q$ is a ``core CQ''
that only needs to be evaluated over the input instance $\Dmf$
(ignoring labeled nulls)
and each $q_i$ is a rAQ.
This is similar to squid decompositions
in~\cite{DBLP:journals/jair/CaliGK13}, but more subtle due to the
presence of subqueries that are not connected to any answer variable
of $q$.
The same decomposition of UCQs is also used in the ``only if'' direction
of part~\ref{thm:PTime/coNPhard}.
The following definition formally introduces decompositions of UCQs.

\begin{definition}\label{def:decomposition}\upshape
  Let $\Omc$ be an ontology and let $q(\vec{x})$ be a UCQ.
  A \emph{decomposition} of $q(\vec{x})$ w.r.t.\ $\Omc$
  is a finite set $\Qmc$ of pairs $(\phi(\vec{y}),\Cmc)$,
  where
  \begin{itemize}
  \item
    $\phi(\vec{y})$ is a conjunction of atoms (possibly equality atoms)
    such that $\vec{y}$ contains all variables in $\vec{x}$;
  \item
    $\Cmc$ is a finite set of CQs $q(\vec{z})$,
    where $\Dmf_q$ is cg-tree decomposable;
  \end{itemize}
  such that the following are equivalent
  for each instance $\Dmf$ and tuple $\vec{a} \in \dom(\Dmf)^{\len{\vec{x}}}$:
  \begin{enumerate}
  \item
    $\Omc,\Dmf \models q(\vec{a})$.
  \item
    For each model $\Amf$ of $\Dmf$ and $\Omc$
    there exists a pair $(\phi(\vec{y}),\Cmc)$ in $\Qmc$
    and an assignment $\pi$ of elements in $\dom(\Dmf)$
    to the variables in $\vec{y}$
    such that
    \begin{itemize}
    \item
      $\pi(\vec{x}) = \vec{a}$,
    \item
      $\Dmf \models \phi(\pi(\vec{y}))$,
      and
    \item
      $\Amf \models q'(\pi(\vec{z}))$ for each $q'(\vec{z})$ in $\Cmc$.
  \end{itemize}
  \end{enumerate}
  The decomposition is \emph{strong} if for each pair $(\phi(\vec{y}),\Cmc)$,
  the set $\Cmc$ consists of rAQs.
\end{definition}

We aim to prove that for every $\uGF$ or $\uGC$ ontology $\Omc$
and each UCQ $q(\vec{x})$
there exists a strong decomposition of $q(\vec{x})$ w.r.t.\ $\Omc$.
As an intermediate step, we prove:

\begin{lemma}\label{weak-decomposition}
  For each $\uGF$ or $\uGC$ ontology $\Omc$ and each UCQ $q(\vec{x})$
  there is a decomposition $\Qmc$ of $q(\vec{x})$ w.r.t.\ $\Omc$.
  Moreover, if $\Omc$ is an $\uGC$ ontology,
  then each rAQ that occurs in a pair of $\Qmc$ has a single free variable.
\end{lemma}

\begin{proof}
  The construction is similar to that of squid decompositions
  in~\cite{DBLP:journals/jair/CaliGK13}.
  We first construct $\Qmc$
  and then show that it is a decomposition of $q(\vec{x})$ w.r.t.\ $\Omc$.
  To simplify the presentation,
  we will assume that each element that occurs in $\Dmf_{q'}$ for a CQ $q'$
  is identified with the variable it represents.
  Consider the set $\Imc$ of all tuples $(q',f,V,H,(T_i)_{i=1}^k)$,
  where
  \begin{itemize}
  \item
    $q'$ is a disjunct of $q$;
  \item
    $f\colon \dom(\Dmf_{q'}) \to \dom(\Dmf_{q'})$ is an idempotent mapping
    such that for each answer variable $x$ of $q'$
    we have that $f(x)$ is an answer variable;
  \item
    $V \subseteq f(\dom(\Dmf_{q'}))$
    and $V$ contains all variables in $f(\vec{x})$;
  \item
    $H,T_1,\dotsc,T_n$ form a partition of $f(\Dmf_{q'})$,
    and each atom in $H$ only contains variables in $V$;
  \item
    each $T_i$ is cg-tree decomposable;
  \item
    if $\Omc$ is a $\uGC$ ontology,
    then each $\dom(T_i)$ contains at most one variable from $V$;
  \item
    $k$ is at most the number of atoms in $q'$.
  \end{itemize}
  Clearly, $\Imc$ is finite.
  Each tuple $I \in \Imc$ contributes exactly one pair $p_I$ to $\Qmc$,
  which we describe next.

  Let $I = (q',f,V,H,(T_i)_{i=1}^k) \in \Imc$.
  Define $\phi(\vec{y})$ to be the conjunction of all atoms of $H$
  and all equality atoms $x = f(x)$ for each variable $x$ in $\vec{x}$
  with $f(x) \neq x$.
  Furthermore, for each $i \in \set{1,\dotsc,k}$,
  let $q_i(\vec{z}_i)$ be the CQ whose atoms are the atoms in $T_i$,
  and whose tuple $\vec{z}_i$ of answer variables
  consists of all variables in $\dom(T_i)$ that also occur in $V$.
  Let
  \[
    p_I = (\phi(\vec{y}),\set{q_i(\vec{z}_i) \mid 1 \leq i \leq k}).
  \]
  We show that $\Qmc = \set{p_I \mid I \in \Imc}$ 
  is a decomposition of $q(\vec{x})$ w.r.t.\ $\Omc$.

  First note that each pair in $\Qmc$ has the form $(\phi(\vec{y}),\Cmc)$,
  where $\phi(\vec{y})$ is a conjunction of atomic formulas
  that contains all variables of $\vec{x}$,
  and $\Cmc$ is a finite set of CQs $q'(\vec{z})$
  whose instance $\Dmf_{q'}$ is cg-tree decomposable.
  Furthermore, if $\Omc$ is formulated in $\uGC$,
  then each CQ in $\Cmc$ has at most one answer variable.

  Next we establish the equivalence between conditions 1 and 2
  from Definition~\ref{def:decomposition}.
  For the direction from 2 to 1,
  suppose that for each model $\Amf$ of $\Dmf$ and $\Omc$
  there exists a pair $(\phi(\vec{y}),\Cmc) \in \Qmc$
  and an assignment $\pi$ of elements in $\dom(\Dmf)$
  to the variables in $\vec{y}$
  such that $\pi(\vec{x}) = \vec{a}$,
  $\Dmf \models \phi(\pi(\vec{y}))$,
  and $\Amf \models q'(\pi(\vec{z}))$ for each $q'(\vec{z}) \in \Cmc$.
  By construction of $(\phi(\vec{y}),\Cmc)$
  this implies $\Omc,\Dmf \models q$.

  For the direction from 2 to 1, suppose that $\Omc,\Dmf \models q(\vec{a})$,
  and consider a model $\Amf$ of $\Dmf$ and $\Omc$.
  By Lemma~\ref{lem:forestmodel},
  there is a forest-model $\Bmf$ of $\Dmf$ and $\Omc$
  and a homomorphism $h$ from $\Bmf$ to $\Amf$ that preserves $\dom(\Dmf)$.
  Let
  \[
    \Bmf = \Dmf \cup \bigcup_{G \in \Gmc} \Bmf_G,
  \]
  where $\Gmc$ is the set of all maximal guarded subsets of $\Dmf$,
  and for each $\Bmf_G$ there is a cg-tree decomposition of $\Bmf_G$
  whose root is labeled with a bag with domain $G$.
  Now, $\Omc,\Dmf \models q(\vec{a})$ implies $\Bmf \models q(\vec{a})$,
  so there is a disjunct $q'(\vec{x})$ of $q(\vec{x})$
  and a homomorphism $\pi$ from $\Dmf_{q'}$ to $\Bmf$
  with $\pi(\vec{x}) = \vec{a}$.
  Note that $\pi$ is the composition of
  \begin{itemize}
  \item
    an idempotent mapping $f\colon \dom(\Dmf_{q'}) \to \dom(\Dmf_{q'})$
    such that $f(\vec{x})$ contains only variables from $\vec{x}$,
    and
  \item
    an injective mapping $g$ from $f(\dom(\Dmf_{q'}))$ to $\Bmf$.
  \end{itemize}
  Let
  \begin{itemize}
  \item
    $V \isdef
    \set{f(x) \in \dom(\Dmf_{q'}) \mid g(f(x)) \in \dom(\Dmf)}$;
  \item
    $H \isdef
    \set{f(\alpha) \mid \alpha \in \Dmf_{q'},\,
      g(f(\alpha)) \in \Dmf \setminus \bigcup_{G \in \Gmc} \Bmf_G}$;
  \item
    $T_G \isdef
    \set{f(\alpha) \mid \alpha \in \Dmf_{q'},\, g(f(\alpha)) \in \Bmf_G}$
    for $G \in \Gmc$.
  \end{itemize}
  Here, for each atom $\alpha = R(x_1,\dotsc,x_k)$
  and mapping $h$ with domain $\{x_1,\dotsc,x_k\}$
  we use the notation $h(\alpha)$
  to denote the atom $R(h(x_1),\dotsc,h(x_k))$.
  Let $T_1,\dotsc,T_k$ be an enumeration of all connected components
  of the sets $T_G$, $G \in \Gmc$.
  Then each $T_i$ is cg-tree decomposable,
  and $k$ is bounded by the number of atoms in $q'$.
  Now,
  \[
    I = (q'(\vec{x}),f,V,H,(T_i)_{i=1}^k) \in \Imc,
  \]
  and $p_I = (\phi(\vec{y}),\Cmc) \in \Qmc$.
  It is straightforward to verify that $\Dmf \models \phi(\pi(\vec{y}))$
  and $\Bmf \models q'(\pi(\vec{z}))$ for each $q'(\vec{z}) \in \Cmc$.
  Since $h$ is a homomorphism from $\Bmf$ to $\Amf$ that preserves $\dom(\Dmf)$,
  we obtain that $\pi' \isdef h \circ \pi$ is an assignment
  of elements in $\dom(\Dmf)$ to the variables in $\vec{y}$
  with $\pi'(\vec{x}) = \vec{a}$,
  $\Dmf \models \phi(\pi'(\vec{y}))$,
  and $\Amf \models q'(\pi'(\vec{z}))$ for each $q'(\vec{z}) \in \Cmc$.

  Altogether,
  this completes the proof that $\Qmc$ is a decomposition of $q(\vec{x})$
  w.r.t.\ $\Omc$.
\end{proof}

Next we show that for $\uGF$ or $\uGC$ ontologies $\Omc$
we can transform any decomposition of a UCQ $q(\vec{x})$ w.r.t.\ $\Omc$
into a strong decomposition of $q(\vec{x})$ w.r.t.\ $\Omc$.
This transformation is based on Lemma~\ref{bounded-depth} below.
Given a $\uGF$ or $\uGC$ ontology $\Omc$
and a decomposition $\Qmc$ of a UCQ w.r.t.\ $\Omc$,
the lemma tells us that for each model $\Amf$ of $\Dmf$ and $\Omc$,
each pair $(\phi(\vec{y}),\Cmc) \in \Qmc$,
each assignment $\pi$ of elements in $\dom(\Dmf)$
to the variables in $\vec{y}$,
and each $q'(\vec{z}) \in \Cmc$,
it suffices to consider homomorphisms from $\Dmf_{q'}$ to $\Amf$
whose image contains an element at a bounded distance
from an element in $\dom(\Dmf)$.
Let us first make more precise what we mean by the distance
between elements in an instance.

\begin{definition}\upshape
  Let $\Amf$ be an instance.
  \begin{itemize}
  \item
    The \emph{Gaifman graph} of an instance $\Amf$ is the undirected graph
    whose vertex set is $\dom(\Amf)$
    and which has an edge between any two distinct $a,b \in \dom(\Amf)$
    if $a$ and $b$ occur together in an atom of $\Amf$.
  \item
    Given $a,b \in \dom(\Amf)$,
    the \emph{distance} from $a$ to $b$ in $\Amf$, denoted by $\dist_\Amf(a,b)$,
    is the length of a shortest path from $a$ to $b$
    in the Gaifman graph of $\Amf$,
    or $\infty$ if there exists no such path.
  \item
    Given $A,B \subseteq \dom(\Amf)$,
    the \emph{distance} from $A$ to $B$ in $\Amf$
    is defined as $\dist_\Amf(A,B) := \min_{a \in A,b \in B} \dist_\Amf(a,b)$.
  \end{itemize}
\end{definition}

\begin{lemma}\label{bounded-depth}
  Let $\Omc$ be a $\uGF$ or $\uGC$ ontology,
  and let $\Qmc$ be a decomposition of a UCQ $q_0(\vec{x})$ w.r.t.\ $\Omc$.
  Then there exists an integer $d \geq 0$
  such that the following is true for all instances $\Dmf$
  and all tuples $\vec{a} \in \dom(\Dmf)^{\len{\vec{x}}}$.

  If $\Omc,\Dmf \models q_0(\vec{a})$,
  then for each model $\Amf$ of $\Dmf$ and $\Omc$
  there exists a pair $(\phi(\vec{y}),\Cmc)$ in $\Qmc$
  and an assignment $\pi$
  of elements in $\dom(\Dmf)$ to the variables in $\vec{y}$ such that
  \begin{itemize}
  \item
    $\pi(\vec{x}) = \vec{a}$,
  \item
    $\Dmf \models \phi(\pi(\vec{y}))$, and
  \item
    for each $q(\vec{z}) \in \Cmc$
    there is a homomorphism $h$ from $\Dmf_q$ to $\Amf$
    such that $h(\vec{z}) = \pi(\vec{z})$
    and the distance from $\dom(h(\Dmf_q))$ to $\dom(\Dmf)$ in $\Amf$
    is at most $d$.
  \end{itemize}
\end{lemma}

\begin{proof}
  For each model $\Amf$ of $\Dmf$ and $\Omc$ and each positive integer $d$,
  let $\Qmc_d(\Amf)$ be the set of all triples $(\phi(\vec{y}),\Cmc,\pi)$
  such that $(\phi(\vec{y}),\Cmc) \in \Qmc$
  and $\pi$ is an assignment of elements in $\dom(\Dmf)$
  to the variables in $\vec{y}$ such that
  \begin{itemize}
  \item
    $\pi(\vec{x}) = \vec{a}$,
  \item
    $\Dmf \models \phi(\pi(\vec{y}))$, and
  \item
    for each $q(\vec{z}) \in \Cmc$
    there is a homomorphism $h$ from $\Dmf_q$ to $\Amf$
    such that $h(\vec{z}) = \pi(\vec{z})$
    and the distance from $\dom(h(\Dmf_q))$ to $\dom(\Dmf)$ in $\Amf$
    is at most $d$.
  \end{itemize}
  We show that there exists a constant $d \geq 0$ with the following property:
  if $\Omc,\Dmf \models q_0(\vec{a})$,
  then for each model $\Amf$ of $\Dmf$ and $\Omc$
  we have $\Qmc_d(\Amc) \neq \emptyset$.

  The constant $d$ depends on the number of $C$-types, which we define next.
  Let $\Qmc'$ be the set of all \emph{Boolean} CQs that occur in the set $\Cmc$
  of some pair $(\phi(\vec{y}),\Cmc) \in \Qmc$.
  For the definition of $C$-types we assume that each $q \in \Qmc'$
  is a GF sentence (if $\Omc$ is formulated in ${\uGF}$),
  or a GC$_2$ sentence (if $\Omc$ is formulated in ${\uGC}$).
  We can assume this w.l.o.g.\ because for each $q \in \Qmc'$
  the instance $\Dmf_q$ has a cg-tree decomposition,
  which can be used to construct a GF or GC$_2$ sentence
  that is equivalent to $q$.
  Now, let $\cl(\Omc,q)$ be the smallest set satisfying:
  \begin{itemize}
  \item
    $\Omc \cup \Qmc' \subseteq \cl(\Omc,q)$;
  \item
    $\cl(\Omc,q)$ contains one atomic formula $R(\vec{x})$
    for each relation symbol $R$ that occurs in $\Omc$ or $q$,
    where $\vec{x}$ is a tuple of distinct variables;
  \item
    $\cl(\Omc,q)$ is closed under subformulas and single negation,
    where the subformulas
    of a formula $\phi = Q \vec{x}\, \psi$ with a quantifier $Q$
    are $\phi$ itself and $\psi$.
  \end{itemize}
  Given a set $C \subseteq \Delta_D$,
  a \emph{$C$-type} is a set of formulas $\phi(\vec{c})$,
  where $\phi(\vec{x}) \in \cl(\Omc,q)$
  and $\vec{c}$ is a tuple of constants in $C$ of the appropriate length.
  Let $w$ be the maximum arity of a relation symbol in $\Omc$ or $q$,
  and define $\tau$ to be the number of $C$-types for a set $C$ of size $2w$.
  Note that $\tau$ depends only on $\Omc$ and $q$,
  and is bounded by $2^{\card{\cl(\Omc,q)} \cdot (2w)^w}$.
  Fix
  \[
    d := s + d^* \quad \text{with} \quad d^* \isdef \tau^2+1,
  \]
  where $s$ is the maximum number of atoms in a non-Boolean CQ
  that occurs in $\Cmc$ for some pair $(\phi(\vec{y},\Cmc) \in \Qmc$.
  This finishes the definition of the constant $d$.

  Suppose now that $\Omc,\Dmf \models q_0(\vec{a})$.
  We aim to show that each model $\Amf$ of $\Dmf$ and $\Omc$
  satisfies $\Qmc_d(\Amc) \neq \emptyset$.
  Let $\Amf$ be a model of $\Dmf$ and $\Omc$.
  By Lemma~\ref{lem:forestmodel},
  there is a forest-model $\Bmf$ of $\Dmf$ and $\Omc$
  and a homomorphism $g$ from $\Bmf$ to $\Amf$ that preserves $\dom(\Dmf)$.
  We show that $\Qmc_d(\Bmc) \neq \emptyset$.
  Since $g$ preserves CQs and does not increase distances
  (i.e., for all $a,b \in \dom(\Bmf)$
   we have $\dist_\Amf(g(a),g(b)) \leq \dist_\Bmf(a,b)$),
  this implies $\Qmc_d(\Amc) \neq \emptyset$, as desired.

  For a contradiction, suppose that $\Qmc_d(\Bmf) = \emptyset$.
  We use $\Bmf$ to construct a model of $\Dmf$ and $\Omc$
  that does not satisfy $q_0(\vec{a})$.
  This contradicts $\Omc,\Dmf \models q_0(\vec{a})$,
  and finishes the proof.
  The following is the core of the construction.

  \begin{trivlist}\item\textit{Claim.~}\itshape
    Let $\Cmf$ be a forest-model of $\Dmf$ and $\Omc$
    and let $e \geq d$ be such that $\Qmc_e(\Cmf) = \emptyset$.
    Then there is a forest-model $\Cmf'$ of $\Dmf$ and $\Omc$
    with $\Qmc_{e+1}(\Cmf') = \emptyset$,
    and $\Cmf'$ agrees with $\Cmf$ on all elements
    at distance at most $e-d^*+1$ from $\dom(\Dmf)$.
  \end{trivlist}

  \begin{trivlist}\item\textit{Proof.}
    Recall the definition of the set $\Qmc'$
    given at the beginning of the proof of Lemma~\ref{bounded-depth}.
    Let $X$ be the set of all $q \in Q'$
    such that for each homomorphism $h$ from $\Dmf_q$ to $\Cmf$,
    \[
      \dist_\Cmf(\dom(h(\Dmf_q)),\dom(\Dmf)) \geq e+1.
    \]
    For each instance $\Cmf'$ and each integer $e'$,
    let $H_{e'}(\Cmf')$ be the set of all pairs $(q,h)$
    such that $q \in X$
    and $h$ is a homomorphism from $\Dmf_q$ to $\Cmf'$
    with
    \[
      \dist_{\Cmf'}(\dom(h(\Dmf_q)),\dom(\Dmf)) \leq e'.
    \]
    Pick a forest-model $\Cmf'$ of $\Dmf$ and $\Omc$ such that
    \begin{enumerate}
    \item
      $H_e(\Cmf')$ is empty;
    \item
      $H_{e+1}(\Cmf')$ is inclusion-minimal;
    \item
      for each $(\phi(\vec{y}),\Cmc) \in \Qmc$,
      each assignment $\pi$ with $\pi(\vec{x}) = \vec{a}$
      and $\Dmf \models \phi(\pi(\vec{y}))$,
      and each non-Boolean $q(\vec{z}) \in \Cmc$
      we have $\Cmf \models q(\pi(\vec{z}))$
      iff $\Cmf' \models q(\pi(\vec{z}))$;
    \item
      $\Cmf'$ agrees with $\Cmf$ on all elements
      at distance at most $e-d+1$ from $\dom(\Dmf)$.
    \end{enumerate}
    Such a model exists since $\Cmf$ is a forest-model of $\Dmf$ and $\Omc$
    that satisfies conditions~1 and~3.
    We show that $H_{e+1}(\Cmf') = \emptyset$.
    This implies $\Qmc_{e+1}(\Cmf') = \emptyset$ as follows.
    Suppose, for a contradiction, that $H_{e+1}(\Cmf') = \emptyset$
    and $\Qmc_{e+1}(\Cmf') \neq \emptyset$.
    Then any $(\phi(\vec{y}),\Cmc,\pi) \in \Qmc_{e+1}(\Cmf')$
    satisfies $\pi(\vec{x}) = \vec{a}$, $\Dmf \models \phi(\pi(\vec{y}))$,
    and $\Cmf \models q(\pi(\vec{z}))$ for all non-Boolean $q(\vec{z}) \in \Cmc$
    (by point~3).
    But then, there is a Boolean $q \in \Cmc$ with $q \in X$
    (as otherwise $(\phi(\vec{y}),\Cmc,\pi) \in \Qmc_{e}(\Cmf)$).
    Since by $(\phi(\vec{y}),\Cmc,\pi) \in \Qmc_{e+1}(\Cmf')$
    there exists a homomorphism $h$ from $\Dmf_q$ to $\Cmf'$
    such that the distance from $\dom(h(\Dmf_q))$ to $\dom(\Dmf)$
    is at most $e+1$,
    we obtain $(q,h) \in H_{e+1}(\Cmf') = \emptyset$,
    a contradiction.
    Hence, it suffices to show that $H_{e+1}(\Cmf') = \emptyset$.

    For a contradiction, suppose that $H_{e+1}(\Cmf') \neq \emptyset$.
    We construct a new forest-model $\Cmf''$ of $\Dmf$ and $\Omc$
    that satisfies conditions~1,~3, and~4,
    but with $H_{e+1}(\Cmf'') \subsetneq H_{e+1}(\Cmf')$.
    This will contradict that $H_{e+1}(\Cmf')$ is inclusion-minimal.
    The construction is based on a pumping argument.

    By definition we have $\Cmf' = \Dmf \cup \bigcup_{G \in \Gmc} \Cmf'_G$,
    where $\Gmc$ is the set of all maximal guarded subsets of $\Dmf$,
    and for each $G \in \Gmc$
    there is a cg-tree decomposition $(T_G,E_G,\bag_G)$ of $\Cmf'_G$
    whose root $r_G$ satisfies $\dom(\bag_G(r_G)) = G$.
    For each node $t \in T_G$, let $\bag_G^*(t)$ be the union of $\bag_G(t')$,
    where $t'$ ranges over the descendants of $t$ in $(T_G,E_G)$,
    including $t$.

    Observe that for each $(q,h) \in H_{e+1}(\Cmf')$
    there is a $G \in \Gmc$ and a node $t \in T_G$
    at depth at least $e$ in the tree $(T_G,E_G)$
    such that $h(\Dmf_q) \subseteq \bag^*_G(t)$.
    Indeed, since
    \begin{align}
      \label{bounded-depth/1}
      \dist_{\Cmf'}(\dom(h(\Dmf_q)),\dom(\Dmf)) \geq e+1,
    \end{align}
    the set $h(\Dmf_q)$ does not contain any constant from $\dom(\Dmf)$.
    Moreover, since $\Dmf_q$ is connected,
    there must be a $G \in \Gmc$ and a node $t \in T_G$
    such that $h(\Dmf_q) \subseteq \bag^*_G(t)$.
    This node $t$ can be chosen to have depth at least $e$ in $(T_G,E_G)$,
    because otherwise we could construct a path of length at most $e$
    from an element in $\dom(h(\Dmf_q))$ to an element in $\dom(\Dmf)$
    in the Gaifman graph of $\Cmf'$,
    contradicting~\eqref{bounded-depth/1}.

    Pick any $G \in \Gmc$ and $t \in T_G$ at depth $e$ in $(T_G,E_G)$
    such that for some $(q,h) \in H_{e+1}(\Cmf')$
    we have $h(\Dmf_q) \subseteq \bag^*_G(t)$.
    Let $t_0,t_1,\dotsc,t_{d^*}$ be the last $d^*+1$ nodes
    on the path from $r_G$ to $t$ in $(T_G,E_G)$,
    and define $C_i \isdef \dom(\bag_G(t_i))$
    as well as $C_i^+ \isdef C_{i-1} \union C_i$.
    For each $i \in \set{1,\dotsc,d^*}$,
    define the type and extended type of $t_i$ as follows:
    \begin{align*}
      \theta_i
      & \isdef \set{\phi(\vec{a}) \mid \phi(\vec{x}) \in \cl(\Omc,q),\,
          \vec{a} \in C_i^{\len{\vec{x}}},\, \Cmf' \models \phi(\vec{a})},
      \\
      \theta_i^+
      & \isdef \set{\phi(\vec{a}) \mid \phi(\vec{x}) \in \cl(\Omc,q),\,
          \vec{a} \in (C_i^+)^{\len{\vec{x}}},\, \Cmf' \models \phi(\vec{a})}.
    \end{align*}
    Note that $\theta_i$ is a $C_i$-type and $\theta_i^+$ is a $C_i^+$-type.
    By our choice of $d^*$,
    there are $1 \leq i < j \leq d^*$
    and a bijective mapping $f\colon C_i^+ \to C_j^+$
    with $f(\theta_i^+) = \theta_j^+$ and $f(\theta_i) = \theta_j$,
    where $f(\theta_i^+)$ is obtained from $\theta_i^+$
    by replacing each constant $a$ that occurs in it by $f(a)$,
    and likewise for $f(\theta_i)$.
    In particular, $f$ is an isomorphism from $\bag_G(t_i)$ to $\bag_G(t_j)$.
    We extend $f$ to an injective mapping with domain $\dom(\bag_G^*(t_i))$
    such that each element that does not occur in $C_i$
    is mapped to an element in $\Delta_D \setminus \dom(\Cmf')$.

    We are now ready to construct $\Cmf''$.
    First, define
    \[
      \Cmf''_G \isdef (\Cmf'_G \setminus \bag^*_G(t_j)) \union f(\bag^*_G(t_i)).
    \]
    Note that a cg-tree decomposition of $\Cmf''_G$
    can be obtained from $(T_G,E_G,\bag_G)$
    by removing the subtree of $(T_G,E_G)$ rooted at $t_j$,
    adding a fresh copy of the subtree rooted at $t_i$,
    making its root a child of $t_{j-1}$,
    and renaming each element $a$
    that occurs in the bag of a node in the new subtree to $f(a)$.
    The root of this cg-tree decomposition is the root of $(T_G,E_G,\bag_G)$,
    and is therefore labeled with a bag with domain $G$.
    We let
    \[
      \Cmf'' \isdef \Dmf \union \Cmf''_G \union
      \bigcup_{G' \in \Gmc \setminus \set{G}} \Cmf'_{G'}.
    \]
    It is straightforward to verify that $\Cmf''$
    is a forest-model of $\Dmf$ and $\Omc$
    that has the desired properties.
    For point 4, note that $\Cmf''$ agrees with $\Cmf'$
    on all elements whose distance from $\dom(\Dmf)$ in $\Cmf''$
    is at most the depth of $t_j$ in $(T_G,E_G)$,
    which is at least $e-d^*+1$.
    Note that point 3 follows from point 4,
    because $e-d^*+1 \geq s+1$
    and $s$ was chosen in such a way that any non-Boolean CQ
    that occurs in $\Cmc$ for some $(\phi(\vec{y}),\Cmc) \in \Qmc$
    and that has a match into $\Cmf''$
    has a match into the subinstance of $\Cmf''$
    induced by all elements at distance at most $s$ from $\dom(\Dmf)$.
    \hfill$\lrcorner$
  \end{trivlist}

  To complete the proof of Lemma~\ref{bounded-depth},
  we apply the claim repeatedly to $\Bmf$
  to obtain a sequence $\Bmf_0,\Bmf_1,\Bmf_2,\dotsc$
  of forest models $\Bmf_i$ of $\Dmf$ and $\Omc$
  such that for each $i \geq 0$,
  \begin{itemize}
  \item
    $\Qmc_{d+i}(\Bmf_i) = \emptyset$;
  \item
    $\Bmf_i$ and $\Bmf_{i+1}$ agree on all elements
    at distance at most $s+i+1$ from elements in $\dom(\Dmf)$.
  \end{itemize}
  By compactness, we obtain a forest-model $\Bmf'$ of $\Dmf$ and $\Omc$
  with $\Qmc_{e}(\Bmf') = \emptyset$ for all integers $e \geq 0$.
  Since $\Qmc$ is a decomposition of $q_0(\vec{x})$ w.r.t.\ $\Omc$,
  this implies $\Bmf' \not\models q_0(\vec{a})$, as desired.
\end{proof}

We now use Lemma~\ref{bounded-depth}
to show that any decomposition of a UCQ $q(\vec{x})$
w.r.t.\ a $\uGF$ or $\uGC$ ontology $\Omc$
can be turned into a strong decomposition of $q(\vec{x})$ w.r.t.\ $\Omc$.

\begin{lemma}\label{weak-to-strong-decomposition}
  Let $\Omc$ be a $\uGF$ or $\uGC$ ontology,
  and let $\Qmc$ be a decomposition of a UCQ $q(\vec{x})$ w.r.t.\ $\Omc$.
  Then there is a strong decomposition of $q(\vec{x})$ w.r.t.\ $\Omc$.
\end{lemma}

\begin{proof}
  Fix the constant $d$ from Lemma~\ref{bounded-depth}.
  Consider a pair $p = (\phi(\vec{x}),\Cmc) \in \Qmc$,
  and let $q_1(\vec{z}_1),\dotsc,q_n(\vec{z}_n)$ be an enumeration
  of all CQs in $\Cmc$.
  For each $i \in \set{1,\dotsc,n}$,
  we define a set $\Cmc_i$ of rAQs as follows.
  If $q_i(\vec{z}_i)$ is non-Boolean, then $q_i(\vec{x})$ is an rAQ
  and we define $\Cmc_i \isdef \set{q_i(\vec{x})}$.
  Otherwise, if $q_i$ is Boolean, let $q_i \gets \phi$.
  Then, $\Cmc_i$ consists of all rAQs of the form
  \[
    q'(x) \gets R_1(\vec{y}_1) \land \dotsb \land R_e(\vec{y}_e) \land \phi,
  \]
  for $e \leq d$,
  where
  \begin{itemize}
  \item
    $x$ occurs in $\vec{y}_1$;
  \item
    each $R_i$ is an at least binary relation symbol in $\Omc$;
  \item
    $\vec{y}_e$ contains a variable $y$ in one of the atoms of $\phi$,
    but no other variable from $\phi$ occurs in any of the $\vec{y}_i$.
  \end{itemize}
  Let $\Qmc'_p$ be the set of all pairs
  $(\phi(\vec{x}),\set{q'_i(\vec{z}_i) \mid 1 \leq i \leq n})$,
  where $q'_i(\vec{z}_i) \in \Cmc_i$ for each $i \in \set{1,\dotsc,n}$.
  By the construction of $\Qmc'_p$ and our choice of $d$,
  it follows that the following are equivalent
  for each model $\Amf$ of $\Dmf$ and $\Omc$:
  \begin{enumerate}
  \item
    there exists an assignment $\pi$ such that $\pi(\vec{x}) = \vec{a}$,
    $\Dmf \models \phi(\pi(\vec{y}))$,
    and $\Amf \models q_i(\pi(\vec{z}_i))$ for each $1 \leq i \leq n$;
  \item
    there exists a pair $(\phi(\vec{x}),\Cmc') \in \Qmc'_p$
    and an assignment $\pi$ such that $\pi(\vec{x}) = \vec{a}$,
    $\Dmf \models \phi(\pi(\vec{y}))$,
    and $\Amf \models q'(\pi(\vec{z}))$ for each $q'(\vec{z}) \in \Cmc'$.
  \end{enumerate}
  Altogether, this implies that $\Qmc' \isdef \bigcup_{p \in \Qmc} \Qmc'_p$
  is a strong decomposition of $q(\vec{x})$ w.r.t.\ $\Omc$.
\end{proof}

\begin{corollary}\label{strong-decomposition}
  For each $\uGF$ or $\uGC$ ontology $\Omc$ and each UCQ $q(\vec{x})$
  there is a strong decomposition $\Qmc$ of $q(\vec{x})$ w.r.t.\ $\Omc$.
  Moreover, if $\Omc$ is an $\uGC$ ontology,
  then each rAQ that occurs in a pair of $\Qmc$ has a single free variable.
\end{corollary}

We are now ready to give the proof of Theorem~\ref{thm:PTime}.

\begin{proof}[of Theorem~\ref{thm:PTime}]
  \textsc{Part~\ref{thm:PTime/PTime} and~\ref{thm:PTime/Datalog}:}
  Since the ``only if'' directions are trivial,
  it suffices to focus on the direction from rAQ-evaluation to UCQ-evaluation.
  Suppose that rAQ-evaluation w.r.t.~$\Omc$ is in \PTime
  (resp., Datalog$^{\neq}$-rewritable),
  and let $q(\vec{x})$ be a UCQ.
  We show that evaluating $q(\vec{x})$ w.r.t.~$\Omc$ is in \PTime
  (resp., Datalog$^{\neq}$-rewritable).

  Let $\Qmc$ be a strong decomposition of $q(\vec{x})$ w.r.t.\ $\Omc$,
  which exists by Corollary~\ref{strong-decomposition}.
  Then for all instances $\Dmf$ and all tuples $\vec{a}$ over $\dom(\Dmf)$
  of length $\lvert\vec{x}\rvert$
  we have $\Omc,\Dmf \models q(\vec{a})$ iff the following is true:
  \begin{align}
    \label{eq:PTime/1}
    \parbox{0.85\linewidth}{
      There exists $(\phi(\vec{y}),\Cmc) \in \Qmc$
      and an assignment $\pi$ of elements in $\dom(\Dmf)$
      to the variables in $\vec{y}$ with
      \begin{enumerate}
      \item
        $\pi(\vec{x}) = \vec{a}$,
      \item
        $\Dmf \models \phi(\pi(\vec{y}))$, and
      \item
        $\Omc,\Dmf \models q'(\pi(\vec{z}))$
        for each rAQ $q'(\vec{z})$ in $\Cmc$.
      \end{enumerate}
    }
  \end{align}
  Indeed, since rAQ-evaluation w.r.t.~$\Omc$ is in \PTime
  (this also holds if rAQ-evaluation w.r.t.~$\Omc$
   is Datalog$^{\neq}$-rewritable),
  we may assume that $\Omc$ is rAQ-materializable
  (by Theorem~\ref{thm:nomatlower}).
  Pick a rAQ-materialization $\Amf$ of $\Omc$ and $\Dmf$.
  Then condition~2 of Definition~\ref{def:decomposition}
  means that there exists a pair $(\phi(\vec{y}),\Cmc) \in \Qmc$
  and an assignment $\pi$ of elements in $\dom(\Dmf)$
  to the variables in $\vec{y}$
  such that $\pi(\vec{x}) = \vec{a}$, $\Dmf \models \phi(\pi(\vec{y}))$,
  and $\Amf \models q'(\pi(\vec{z}))$ for each rAQ $q'(\vec{z})$ in $\Cmc$.
  Since $\Amf$ is a rAQ-materialization of $\Omc$ and $\Dmf$,
  we can replace $\Amf \models q'(\pi(\vec{z}))$ in the third condition
  by $\Omc,\Dmf \models q'(\pi(\vec{z}))$,
  which yields~\eqref{eq:PTime/1}.

  If rAQ-evaluation w.r.t.~$\Omc$ is in \PTime,
  then~\eqref{eq:PTime/1} yields a polynomial time procedure
  for evaluating $q$ w.r.t.~$\Omc$.

  Moreover, if rAQ-evaluation w.r.t.~$\Omc$ is Datalog$^{\neq}$-re\-writ\-able,
  then we can construct a Datalog$^{\neq}$ program
  for evaluating $q$ w.r.t.~$\Omc$ as follows.
  Let $\Qmc'$ be the set of rAQs that occur in some pair of $\Qmc$.
  For each $q' \in \Qmc'$,
  let $\Pi_{q'}$ be a Datalog$^{\neq}$ program that evaluates $q'$ w.r.t.~$\Omc$.
  Without loss of generality we assume that the intensional predicates
  used in different programmes $\Pi_{q'}$ and $\Pi_{q''}$ are disjoint,
  and that the goal predicate of $\Pi_{q'}$ is $\mathsf{goal}_{q'}$.
  Now let $\Pi$ be the Datalog$^{\neq}$ program
  containing the rules of all programs $\Pi_{q'}$, for $q' \in \Qmc'$,
  and the following rule for each $(\phi(\vec{y}),\Cmc) \in \Qmc$:
  \[
    \mathsf{goal}(\vec{x})
    \gets
    \phi(\vec{y})
    \land
    \bigwedge_{q'(\vec{z}) \in \Cmc}\mathsf{goal}_{q'}(\vec{z}).
  \]
  Note that if each $\Pi_{q'}$ is a Datalog program,
  then $\Pi$ is a Datalog program as well.
  Using~\eqref{eq:PTime/1}
  it is straightforward to verify that for all instances $\Dmf$
  we have $\Dmf \models \Pi(\vec{a})$ iff $\Omc,\Dmf \models q(\vec{a})$.

  \medskip\noindent
  \textsc{Part~\ref{thm:PTime/coNPhard}:}
  The ``if'' directions are trivial,
  so we focus on the direction from UCQ-evaluation to rAQ-evaluation.
  Suppose that UCQ-evaluation w.r.t.~$\Omc$ is \coNP-hard,
  and let $q(\vec{x})$ be a UCQ that witnesses this.
  We show that rAQ-evaluation w.r.t.~$\Omc$ is \coNP-hard
  via a polynomial-time reduction from evaluating $q$ w.r.t.~$\Omc$.

  We start by describing a translation of instances $\Dmf$ and tuples $\vec{a}$,
  and will show afterwards that this translation is a polynomial-time reduction
  from evaluating $q$ w.r.t.~$\Omc$
  to evaluating a suitably chosen rAQ w.r.t.~$\Omc$.

  Let $\Qmc$ be a strong decomposition of $q(\vec{x})$ w.r.t.\ $\Omc$.
  Fix an enumeration $q'_1(\vec{z}_1),\dotsc,q'_m(\vec{z}_m)$
  of all rAQs that occur in some pair of $\Qmc$,
  and let $k_i$ be the length of $\vec{z}_i$.
  Without loss of generality, we can assume that each $\Dmf_{q'_i}$
  is consistent w.r.t.\ $\Omc$.
  We use fresh relation symbols $R$, $S$, and $T_i$ ($1 \leq i \leq m$),
  where $R$ and $S$ are binary, and $T_i$ is $(k_i+1)$-ary.
  Note that each of these relation symbols is at most binary
  in the case that $\Omc$ is a $\uGC$ ontology.

  Given an instance $\Dmf$
  and a tuple $\vec{a} \in \dom(\Dmf)^{\lvert\vec{x}\rvert}$,
  we compute a new instance $\tilde{\Dmf}$
  by adding the following atoms to $\Dmf$.
  First, we add all atoms of $\Dmf_{q'_i}$, for each $i \in \{1,\dotsc,m\}$,
  where we assume w.l.o.g.\ that the domain of $\Dmf_{q'_i}$
  is disjoint from that of $\Dmf$ and $\Dmf_{q'_j}$ for $j \neq i$.
  Let $\vec{c}_i$ be the tuple of elements in $\Dmf_{q'_i}$
  that represents the tuple $\vec{z}_i$ of the answer variables of $q'_i$.
  Next, we add the following atoms
  for each pair $p = (\phi(\vec{y}),\Cmc) \in \Qmc$,
  and for each assignment $\pi$ of elements in $\dom(\Dmf)$
  to the variables in $\vec{y}$
  that satisfies $\pi(\vec{x}) = \vec{a}$
  and $\Dmf \models \phi(\pi(\vec{y}))$:
  \begin{itemize}
  \item
    $R(a_0,a_p)$;
  \item
    $S(a_p,a_{p,\pi})$;
  \item
    $T_i(a_{p,\pi},\pi(\vec{z}_i))$
    for each $i \in \{1,\dotsc,m\}$ with $q'_i \in \Cmc$;
  \item
    $T_i(a_{p,\pi},\vec{c}_i)$
    for each $i \in \{1,\dotsc,m\}$ with $q'_i \notin \Cmc$. 
  \end{itemize}
  Since $\Qmc$ is of constant size,
  the instance $\tilde{\Dmf}$ can be computed
  in time polynomial in the size of $\Dmf$.

  It is now straightforward to verify that $\Omc,\Dmf \models q(\vec{a})$
  holds iff $\tilde{\Dmf},\Omc \models \tilde{q}(a_0)$,
  where $\tilde{q}(x)$ is the rAQ
  \begin{align*}
    \tilde{q}(x)
    & \gets R(x,y) \land S(y,z) \land
    \bigwedge_{i=1}^m \bigl(
      T_i(z,\vec{u}_i) \land q'_i(\vec{u}_i)
    \bigr).
  \end{align*}
  Since evaluating $q(\vec{x})$ w.r.t.~$\Omc$ is \coNP-hard,
  we conclude that evaluating $\tilde{q}(x)$ w.r.t.~$\Omc$ is \coNP-hard.
\end{proof}


\section{Proofs for Section 4}

\begin{trivlist}\item\textbf{Theorem~\ref{thm:datalog} (restated)~}\itshape
   For all $\uGF$ and $\uGC$ ontologies \Omc, unravelling tolerance of
  \Omc implies that rAQ-evaluation w.r.t.\ \Omc is
  Datalog$^{\neq}$-rewritable (and Datalog-rewritable if \Omc is
  formulated in uGF).
\end{trivlist}

\begin{proof}
  Let $q$ be an rAQ.
  We show that answering $q$ w.r.t.\ $\Omc$ is Datalog$^{\neq}$-rewritable.
  We provide separate proofs
  for the case that $\Omc$ is a $\uGF$ ontology (Part 1)
  and for the case that $\Omc$ is a $\uGC$ ontology (Part 2).

  \medskip\noindent
  \textit{Part 1: $\Omc$ is a $\uGF$ ontology.}
  To simplify the presentation, we will regard $q$ as an $\openGF$ formula,
  which we can do w.l.o.g.\
  because $q$ has a connected guarded tree decomposition
  whose root is labeled by the answer variables of $q$.
  Let $\cl(\Omc,q)$ be the smallest set satisfying:
  \begin{itemize}
  \item
    $\Omc \cup \{q\} \subseteq \cl(\Omc,q)$;
  \item
    $\cl(\Omc,q)$ contains one atomic formula $R(\vec{x})$
    for each relation symbol $R$ that occurs in $\Omc$ or $q$,
    where $\vec{x}$ is a tuple of distinct variables;
  \item
    $\cl(\Omc,q)$ contains a single equality atom $x=y$,
    where $x$ and $y$ are distinct variables;
  \item
    $\cl(\Omc,q)$ is closed under subformulas and single negation,
    where the subformulas of a formula $\phi = Q \vec{x}\, \psi$
    with a quantifier $Q$
    are $\phi$ itself and $\psi$.
  \end{itemize}
  Let $w$ be the maximum arity of a relation symbol in $\Omc$ or $q$,
  and let $x_1,\dotsc,x_{2w-1}$ be distinct variables
  that do not occur in $\Omc$ or $q$.
  Given a tuple $\vec{x}$ over $X = \set{x_1,\dotsc,x_{2w-1}}$,
  an $\vec{x}$-\emph{type} $\theta$ is a maximal consistent set of formulas,
  where each formula in $\theta$
  is obtained from a formula $\phi(\vec{y}) \in \cl(\Omc,q)$
  by substituting a variable from $\vec{x}$ for each variable in $\vec{y}$,
  and one formula in $\theta$ is a relational atomic formula
  containing all the variables in $\theta$.
  If $\theta_i$ is an $\vec{x}_i$-type for each $i \in \set{1,2}$,
  then $\theta_1$ and $\theta_2$ are \emph{compatible}
  if they agree on all formulas
  containing only the variables that occur both in $\vec{x}_1$ and $\vec{x}_2$.
  An $\vec{x}$-type $\theta$ is \emph{realizable} in an interpretation $\Amf$
  if there is an assignment $\pi$ of elements in $\dom(\Amf)$
  to the variables in $\vec{x}$
  such that $\Amf \models \phi(\pi(\vec{y}))$
  for each $\phi(\vec{y}) \in \theta$.
  In this case we also say that $\vec{a} = \pi(\vec{x})$
  \emph{realizes} $\theta$ in $\Amf$.
  Let $\tp(\vec{x})$ be the set of all $\vec{x}$-types
  that are realizable in some model of $\Omc$.
  Since $\Omc$ and $q$ are fixed,
  each $\tp(\vec{x})$ is of constant size
  and can be computed in constant time
  using a standard satisfiability procedure
  for the guarded fragment~\cite{DBLP:journals/jsyml/Gradel99}.

  We now describe a Datalog$^{\neq}$ program $\Pi$
  for evaluating $q$ w.r.t.\ $\Omc$.
  For each $l \in \set{1,\dotsc,w}$,
  each $l$-tuple $\vec{x}$ over $X$,
  and each $\Theta \subseteq \tp(\vec{x})$,
  the program uses an intensional $l$-ary relation symbol $P^l_\Theta$.
  Intuitively, $P^l_\Theta(\vec{a})$ encodes an assignment
  of possible types (namely those in $\Theta$)
  for $\vec{a}$ in a model of $\Omc$.
  In the description below,
  $l_1$ and $l_2$ range over integers in $\set{1,\dotsc,w}$,
  and $\vec{x}_1$ and $\vec{x}_2$ range over tuples over $X$
  of length $l_1$ and $l_2$, respectively,
  such that $\vec{x}_2$ contains at least one variable from $\vec{x}_1$.
  Let $k$ be the arity of $q$.
  The program contains the following rules:
  \begin{enumerate}
  \item
    $P^{l_1}_\Theta(\vec{x}_1) \gets R(\vec{y}) \land \alpha(\vec{z})$,
    where $\Theta$ is the set of all $\theta \in \tp(\vec{x}_1)$
    that contain $R(\vec{y})$ and $\alpha(\vec{z})$,
    $\alpha$ is an atomic formula (possibly an equality)
    or an inequality,
    and $\vec{y}$ contains exactly the variables from $\vec{x}_1$;
  \item
    $P^{l_1}_\Theta(\vec{x}_1)
    \gets P^{l_1}_{\Theta_1}(\vec{x}_1) \land P^{l_2}_{\Theta_2}(\vec{x}_2)$,
    where $\Theta_i \subseteq \tp(\vec{x}_i)$,
    and $\Theta$ is the set of all $\theta_1 \in \Theta_1$
    such that there exists a $\theta_2 \in \Theta_2$
    that is compatible with $\theta_1$;
  \item
    $P^{l_1}_{\Theta_1 \isect \Theta_2}(\vec{x}_1)
    \gets P^{l_1}_{\Theta_1}(\vec{x}_1) \land P^{l_1}_{\Theta_2}(\vec{x}_1)$,
    where $\Theta_i \subseteq \tp(\vec{x}_1)$;
  \item
    $\mathsf{goal}(x_{i_1},\dotsc,x_{i_k}) \gets P^{l_1}_\Theta(\vec{x}_1)$,
    where $\Theta \subseteq \tp(\vec{x}_1)$,
    and $q(x_{i_1},\dotsc,x_{i_k}) \in \theta$ for each $\theta \in \Theta$;
  \item
    $\mathsf{goal}(x_1,\dotsc,x_k) \gets P^l_\emptyset(\vec{y})$,
    where $l \leq w$
    and $\vec{y}$ is a tuple of $l$ variables from $X$
    distinct from $x_1,\dotsc,x_k$.
  \end{enumerate}

  It remains to show that $\Pi$ is a Datalog$^{\neq}$-rewriting
  for $q$ w.r.t.\ $\Omc$.
  To this end, let $\Dmf$ be an instance.
  We show that $\Omc,\Dmf \not\models q(\vec{a})$
  iff $\Dmf \not\models \Pi(\vec{a})$.

  For the ``only if'' direction, assume $\Omc,\Dmf \not\models q(\vec{a})$,
  and let $\Bmf$ be a model of $\Dmf$ and $\Omc$
  with $\Bmf \not\models q(\vec{a})$.
  Consider a tuple $\vec{b} = (b_1,\dotsc,b_l) \in \dom(\Dmf)^l$
  for some $l \in \set{1,\dotsc,w}$.
  Note that if $\vec{b}$ is not guarded in $\Bmf$,
  then for all tuples $\vec{x} \in X^l$ and all $\Theta \subseteq \tp(\vec{x})$
  the program $\Pi$ does not derive $P^l_\Theta(\vec{b})$ on input $\Dmf$.
  This is because of the rules in line 1,
  which require $\vec{b}$ to be guarded
  in order to derive $P^l_\Theta(\vec{b})$.
  For each tuple $\vec{y} = (y_1,\dotsc,y_l) \in X^l$,
  define $\theta_{\vec{y}}(\vec{b})$ to be the set of all formulas
  obtained from a formula $\phi(z_1,\dotsc,z_n) \in \cl(\Omc,q)$
  and $1 \leq i_1,\dotsc,i_n \leq l$
  with $\Bmf \models \phi(b_{i_1},\dotsc,b_{i_n})$
  by substituting $y_{i_j}$ for each $z_j$.
  Since $\vec{b}$ is guarded in $\Dmf$,
  $\theta_{\vec{y}}(\vec{b})$ contains a relational atomic formula
  that contains all variables from $\vec{y}$,
  hence $\theta_{\vec{y}}(\vec{b})$ is a $\vec{y}$-type.
  Furthermore, $\theta_{\vec{y}}(\vec{b})$ is realizable in $\Bmf$,
  which implies $\theta_{\vec{y}}(\vec{b}) \in \tp(\vec{y})$.
  By induction on rule applications of $\Pi$ one can show that
  for each guarded tuple $\vec{b}$ in $\dom(\Bmf)$ of length $l$,
  each tuple $\vec{y} \in X^l$,
  and each $\Theta \subseteq \tp(\vec{y})$
  such that $P^l_\Theta(\vec{b})$ is derivable by $\Pi$ on $\Dmf$
  we have $\theta_{\vec{y}}(\vec{a}) \in \Theta$.
  Since $\Bmf \not\models q(\vec{a})$,
  for each guarded tuple $\vec{b} = (b_1,\dotsc,b_l)$ in $\Dmf$
  with $\vec{a} = (b_{i_1},\dotsc,b_{i_k})$
  and each $\vec{y} = (y_1,\dotsc,y_l) \in X^l$
  we have $q(y_{i_1},\dotsc,y_{i_k}) \notin \theta_{\vec{y}}(\vec{b})$.
  Altogether, this implies that $\Dmf \not\models \Pi(\vec{a})$.

  For the ``if'' direction, assume $\Dmf \not\models \Pi(\vec{a})$.
  Let $G_{\vec{a}}$ be a maximal guarded set of $\Dmf$
  containing the elements of $\vec{a}$.
  Recall the definition of the $\uGF$-unravelling of $\Dmf$
  from Section~\ref{sec:unravelling-tolerance}.
  We show that $\Omc,\Dmf^u \not\models q(\vec{b})$,
  where $\vec{b}$ is the copy of $\vec{a}$ in $\text{bag}(G_{\vec{a}})$,
  which implies $\Omc,\Dmf \not\models q(\vec{a})$
  by unravelling tolerance of $\Omc$.
  More precisely, we construct a model $\Bmf$ of $\Dmf^u$ and $\Omc$
  with $\Bmf \not\models q(\vec{b})$.

  Observe that for each $l \in \set{1,\dotsc,w}$,
  each $\vec{c} \in \dom(\Dmf)^l$,
  each $\Theta \subseteq \tp(\vec{x})$,
  and each bijective mapping $f\colon X \to X$,
  the program $\Pi$ derives $P^l_{\Theta}(\vec{c})$
  iff it derives $P^l_{f(\Theta)}(\vec{c})$,
  where $f(\Theta)$ is the $f(\vec{x})$-type obtained from $\Theta$
  by substituting $f(x)$ for each variable $x$ in $\vec{x}$.
  In other words,
  the sets $\Theta$ of types that $\Pi$ derives for each tuple $\vec{c}$
  do not depend on the choice of the variables $\vec{x}$.
  Observe also that for each guarded tuple $\vec{c}$ in $\Dmf$
  there is a unique minimal $\Theta(\vec{c}) \subseteq \tp(x_1,\dotsc,x_l)$
  such that $\Pi$ derives $P^l_{\Theta(\vec{c})}(\vec{c})$.
  Since $\Pi$ does not derive $\mathsf{goal}(\vec{a})$,
  we must have $\Theta(\vec{c}) \neq \emptyset$
  for all guarded tuples $\vec{c}$ in $\Dmf$.
  Furthermore, if $\vec{c} = (c_1,\dotsc,c_l)$
  and $\vec{a} = (c_{i_1},\dotsc,c_{i_k})$,
  then $q(x_{i_1},\dotsc,x_{i_k}) \notin \theta$
  for some $\theta \in \Theta(\vec{c})$.
  Let us denote the set of such types in $\Theta(\vec{c})$
  by $\Theta^{\lnot q}(\vec{c})$.

  To construct the desired model $\Bmf$,
  we first assign to each maximally guarded tuple $\vec{c}$ of $\Dmf^u$
  a type in $\Theta(\vec{c})$ as follows.
  Recall the definition of the tree $T(\Dmf)$
  and the bags $\bag(t)$, $t \in T(\Dmf)$,
  used in the definition of the uGF-unravelling
  in Section~\ref{sec:unravelling-tolerance}.
  For each node $t$ of $T(\Dmf)$,
  let $G_t \isdef \dom(\bag(t))$
  and $\vec{c}_t$ a $\card{G_t}$-tuple of all elements in $G_t$.
  We inductively assign to each node $t$ of $T(\Dmf)$
  a $(x_1,\dotsc,x_{\card{G_t}})$-type $\theta_t$ as follows.
  If $t = G$ for a maximal guarded subset $G$ of $\Dmf$,
  pick a type $\theta_t \in \Theta^{\lnot q}(\vec{c}_t)$.
  For the induction step,
  suppose that $t = t' G$ and $\text{tail}(t') = G'$.
  Let $\vec{c}_{t'} = (c'_1,\dotsc,c'_m)$, $\vec{c}_t = (c_1,\dotsc,c_n)$,
  and define a bijective mapping $f\colon X \to X$
  such that for all $i \in \set{1,\dotsc,n}$,
  \begin{itemize}
  \item
    if $c_i = c'_j$, then $f(x_i) = x_j$;
  \item
    if $c_i \notin \set{c'_1,\dotsc,c'_m}$,
    then $f(x_i) \notin \set{x_1,\dotsc,x_m}$.
  \end{itemize}
  Since $\theta_{t'} \in \Theta(\vec{c}_{t'})$,
  there exists a type $\theta_t \in \Theta(\vec{c}_t)$
  such that $\theta_{t'}$ and $f(\theta_t)$ are compatible.
  This follows from the rules in line 2 of the definition of $\Pi$.

  Finally, for each node $t$ of $T(\Dmf)$,
  pick a model $\Bmf_t$ of $\Omc$ such that $\vec{c}_t$ realizes $\theta_t$
  in $\Bmf_t$.
  Without loss of generality we assume that
  $\dom(\Bmf_t) \isect \dom(\Dmf^u) = G_t$ for each node $t$ of $T(\Dmf)$,
  and $\dom(\Bmf_t) \isect \dom(\Bmf_{t'}) = G_t \cap G_{t'}$
  for every two distinct $t,t' \in T(\Dmf)$
  Let $\Bmf$ be the interpretation obtained from $\Dmf^u$
  by hooking $\Bmf_t$ to $\Dmf^u$ for each node $t$ of $T(\Dmf)$:
  \[
    \Bmf \ :=\ \Dmf^u \cup \bigcup_{t \in T(\Dmf)} \Bmf_t.
  \]
  Clearly, $\Bmf$ is a model of $\Dmf^u$.
  It remains to show that $\Bmf$ is a model of $\Omc$
  with $\Bmf \not\models q(\vec{b})$.

  \begin{trivlist}
  \item\textit{Claim.}
    \textit{For all $\openGF$ formulas $\phi(\vec{x})$,
      all guarded tuples $\vec{c}$ of $\Bmf$,
      and all nodes $t \in T(\Dmf)$ with $\vec{c} \subseteq \dom(\Bmf_t)$,
      \begin{align}
        \label{eq:datalog/main}
        \Bmf_t \models \phi(\vec{c})
        \iff
        \Bmf \models \phi(\vec{c}).
      \end{align}}
  \end{trivlist}

  \begin{trivlist}
  \item\textit{Proof.}
    By induction on the structure of $\phi$.
    For the base case assume that $\phi$ is an atomic formula.
    Let $\vec{c}$ be a guarded tuple of $\Bmf$
    and let $t \in T(\Dmf)$ be such that $\vec{c} \subseteq \dom(\Bmf_t)$.
    Then, $\Bmf_t \models \phi(\vec{c})$ implies $\Bmf \models \phi(\vec{c})$
    because $\Bmf_t \subseteq \Bmf$.
    For the converse assume that $\Bmf \models \phi(\vec{c})$.
    Since $\phi$ is atomic,
    this implies $\Bmf_{t'} \models \phi(\vec{c})$ for some $t' \in T(\Dmf)$
    with $\vec{c} \subseteq \dom(\Bmf_{t'})$.
    By the choice of the types $\theta_{t}$ and $\theta_{t'}$,
    we obtain $\Bmf_t \models \phi(\vec{c})$.

    For the inductive step, we distinguish the following cases:

    \medskip\noindent\textit{Case 1:
      $\phi(\vec{x}) = \lnot \phi'(\vec{x})$.}
    For each guarded tuple $\vec{c}$ of $\Bmf$
    and each node $t \in T(\Dmf)$ with $\vec{c} \subseteq \dom(\Bmf_t)$,
    we have
    \begin{align*}
      \Bmf_t \models \phi(\vec{c})
      & \iff \Bmf_t \not\models \phi'(\vec{c}) \\
      & \iff \Bmf \not\models \phi'(\vec{c})
        \iff \Bmf \models \phi(\vec{c}),
    \end{align*}
    where the second equivalence follows from the induction hypothesis.

    \medskip\noindent\textit{Case 2:
      $\phi(\vec{x}) = \phi_1(\vec{x}_1) \land \phi_2(\vec{x}_2)$.}
    Consider a guarded tuple $\vec{c}$ of $\Bmf$
    and a node $t \in T(\Dmf)$ with $\vec{c} \subseteq \dom(\Bmf_t)$.
    Let $\vec{c}_i$ be the projection of $\vec{c}$
    onto the positions of $\vec{x}$ that contain a variable from $\vec{x}_i$.
    Then,
    \begin{align*}
      \Bmf_t \models \phi(\vec{c})
      & \iff \Bmf_t \models \phi_i(\vec{c}_i)
        \ \text{for each $i \in \{1,2\}$} \\
      & \iff \Bmf \models \phi_i(\vec{c}_i)
        \ \text{for each $i \in \{1,2\}$} \\
      & \iff \Bmf \models \phi(\vec{c}),
    \end{align*}
    where the second equivalence follows from the induction hypothesis.

    \medskip\noindent\textit{Case 3:
      \(
        \phi(\vec{x}) =
        \forall \vec{y} (\alpha(\vec{x},\vec{y}) \to \psi(\vec{x},\vec{y})).
      \)}
    Consider a guarded tuple $\vec{c}$ of $\Bmf$
    and a $t \in T(\Dmf)$ with $\vec{c} \subseteq \dom(\Bmf_t)$.

    We first prove that $\Bmf_t \models \phi(\vec{c})$
    implies $\Bmf \models \phi(\vec{c})$.
    Suppose that $\Bmf_t \models \phi(\vec{c})$,
    and let $\Bmf \models \alpha(\vec{c},\vec{d})$ for some tuple $\vec{d}$.
    Since $\alpha$ is atomic,
    we have $\Bmf_{t'} \models \alpha(\vec{c},\vec{d})$
    for some node $t' \in T(\Dmf)$
    with $\vec{c},\vec{d} \subseteq \dom(\Bmf_{t'})$.
    Note that $\vec{c}$ contains only values
    in $\dom(\Bmf_t) \cap \dom(\Bmf_{t'})$
    and that $\vec{c}$ is non-empty (because $\phi$ is open).
    By the choice of the types $\theta_t$ and $\theta_{t'}$
    and since $\Bmf_t \models \phi(\vec{c})$,
    we obtain $\Bmf_{t'} \models \phi(\vec{c})$.
    Hence, $\Bmf_{t'} \models \psi(\vec{c},\vec{d})$.
    The induction hypothesis now yields $\Bmf \models \psi(\vec{c},\vec{d})$.
    We have thus shown that $\Bmf \models \phi(\vec{c})$.

    For the converse,
    assume $\Bmf \models \phi(\vec{c})$,
    and let $\Bmf_t \models \alpha(\vec{c},\vec{d})$ for some tuple $\vec{d}$.
    Since $\alpha$ is atomic and $\Bmf_t \subseteq \Bmf$,
    we have $\Bmf \models \alpha(\vec{c},\vec{d})$,
    and therefore $\Bmf \models \psi(\vec{c},\vec{d})$.
    Since $\vec{c},\vec{d}$ is guarded and contained in $\dom(\Bmf_t)$,
    the induction hypothesis implies $\Bmf_t \models \psi(\vec{c},\vec{d})$.
    This shows that $\Bmf_t \models \phi(\vec{c})$.

    \medskip\noindent\textit{Case 4:
      \(
        \phi(\vec{x}) =
        \exists \vec{y} (
          \alpha(\vec{x},\vec{y}) \land \psi(\vec{x},\vec{y})
        ).
      \)}
    Consider a guarded tuple $\vec{c}$ of $\Bmf$
    and a node $t \in T(\Dmf)$ with $\vec{c} \subseteq \dom(\Bmf_t)$.

    If $\Bmf_t \models \phi(\vec{c})$,
    then there exists a tuple $\vec{d} \subseteq \dom(\Bmf_t)$
    such that $\Bmf_t \models \alpha(\vec{c},\vec{d})$
    and $\Bmf_t \models \psi(\vec{c},\vec{d})$.
    By the induction hypothesis,
    this implies $\Bmf \models \alpha(\vec{c},\vec{d})$
    and $\Bmf \models \psi(\vec{c},\vec{d})$,
    and hence $\Bmf \models \phi(\vec{c})$.

    Conversely, if $\Bmf \models \phi(\vec{c})$,
    then there exists a tuple $\vec{d}$
    such that $\Bmf \models \alpha(\vec{c},\vec{d})$
    and $\Bmf \models \psi(\vec{c},\vec{d})$.
    Since $\alpha$ is atomic,
    there exists a node $t' \in T(\Dmf)$
    with $\vec{c},\vec{d} \subseteq \dom(\Bmf_{t'})$,
    and hence by the induction hypothesis we obtain
    $\Bmf_{t'} \models \alpha(\vec{c},\vec{d})$
    and $\Bmf_{t'} \models \psi(\vec{c},\vec{d})$,
    and therefore $\Bmf_{t'} \models \phi(\vec{c})$.
    By the choice of $\theta_t$ and $\theta_{t'}$,
    we obtain $\Bmf_t \models \phi(\vec{c})$.
    \hfill$\lrcorner$
  \end{trivlist}

  Using the claim we now show that $\Bmf$ is a model of $\Omc$
  with $\Bmf \not\models q(\vec{b})$.
  By construction we have $\theta_t \in \Theta^{\lnot q}(\vec{c}_t)$,
  where $t = G_{\vec{a}}$.
  This implies $\Bmf_t \not\models q(\vec{b})$,
  and using the claim we obtain $\Bmf \not\models q(\vec{b})$.

  To prove that $\Bmf$ is a model of $\Omc$,
  let $\psi = \forall \vec{x} (\alpha(\vec{x}) \to \phi(\vec{x}))$
  be a sentence in $\Omc$,
  and let $\vec{c}$ be a tuple with $\Bmf \models \alpha(\vec{c})$.
  Since $\alpha$ is atomic, $\vec{c}$ is a guarded tuple in $\Bmf$.
  In particular, since $\Bmf \models \alpha(\vec{c})$,
  there must be a node $t_0 \in T(\Dmf)$
  with $\vec{c} \subseteq \dom(\Bmf_{t_0})$
  and $\Bmf_{t_0} \models \alpha(\vec{c})$.
  In fact, by the choice of the types assigned to the
  maximally guarded tuples of $\Dmf^u$,
  every node $t \in T(\Dmf)$ with $\vec{c} \subseteq \dom(\Bmf_t)$
  must satisfy $\Bmf_t \models \alpha(\vec{c})$.
  Now let $t$ be a node in $T(\Dmf)$ such that $\vec{c} \subseteq \dom(\Bmf_t)$.
  Since $\Bmf_t \models \alpha(\vec{c})$ and $\Bmf_t$ is a model of $\Omc$,
  we get $\Bmf_t \models \phi(\vec{c})$,
  so by~\eqref{eq:datalog/main} we have $\Bmf \models \phi(\vec{c})$.
  Therefore, $\Bmf \models \psi$.
  This holds for all sentences in $\Omc$,
  hence $\Bmf$ is a model of $\Omc$.

  If $\Omc$ is a uGF-ontology, we obtain a Datalog-rewriting from $\Pi$
  by removing inequalities from types and the rules in line 1.

  \medskip\noindent
  \textit{Part 2: $\Omc$ is a $\uGC$ ontology.}
  Analogous to part~1 we regard $q$ as an $\openGC$ formula,
  which we can do w.l.o.g.\
  because $q$ has a cg-tree decomposition
  where the domain of each bag consists of at most two elements
  and whose root bag has the answer variables of $q$ as its domain.
  We define $\cl(\Omc,q)$, types,
  and the corresponding notion of realization as in part~1.
  Since in the context of $\uGC$ we work over an at most binary signature,
  we will only use types over one or two variables.
  The sets $\tp(x)$ and $\tp(x,y)$ of these types
  can be computed in constant time
  using a satisfiability procedure
  for the two-variable guarded counting fragment~%
  \cite{DBLP:journals/logcom/Pratt-Hartmann07}.
  %
  Given a $\vec{x}_0$-type $\theta_0$
  and a set $\Theta_i \subseteq \tp(\vec{x}_i)$
  for each $i \in \set{1,\dotsc,\ell}$
  we write $\theta_0 \rightsquigarrow (\Theta_i)_{i=1}^\ell$
  iff there exists a $\theta_i \in \Theta_i$
  for each $i \in \set{1,\dotsc,\ell}$
  such that $\bigcup_{i=0}^\ell \theta_i$ is realizable in a model of $\Omc$.
  Furthermore, if $(x,y)$ and $(u,v)$ are pairs of distinct variables
  and $\Theta \subseteq \tp(u,v)$,
  then we write $\Theta_{x,y}$ for the set of all types in $\tp(x,y)$
  that can be obtained from a type in $\Theta$
  by renaming $u$ to $x$ and $v$ to $y$.

  We now turn to the description of the Datalog$^{\neq}$ program $\Pi$
  for evaluating $q$ w.r.t.\ $\Omc$.
  The program uses distinct variables $x,y,z_0,z_1,z_2,\dotsc,z_{N \tau 2^\tau}$.
  Here, $\tau$ is one more than the number of types of two distinct variables,
  and $N$ is the largest integer $n$
  such that a formula of the form $\exists^{\geq n} x\, \phi$
  occurs in $\cl(\Omc,q)$,
  or $1$ if there is no such formula in $\cl(\Omc,q)$.
  For each $i \in \set{1,2}$ and $\Theta \subseteq \tp(u,v)$,
  $\Pi$ uses an intensional binary relation symbol $P_\Theta$
  with the same intended interpretation as the symbols $P_\Theta^l$ in part~1.
  Let $\mathsf{neq}_\ell \isdef \bigwedge_{0 \leq i < j \leq \ell} z_i \neq z_j$.
  Then $\Pi$ consists of the following rules:
  \begin{enumerate}
  \item
    $P_\Theta(x,y) \gets R(\vec{v}) \land \alpha(\vec{w})$,
    where $\Theta$ is the set of all types in $\tp(x,y)$
    containing $R(\vec{v})$ and $\alpha(\vec{w})$,
    $\alpha$ is an atomic formula (including an equality)
    or an inequality,
    and $\vec{v}$ contains both $x$ and $y$;
  \item
    \(
      P_\Theta(x,z_0) \gets
      \bigwedge_{i=0}^l P_{\Theta_i}(x,z_i) \land \mathsf{neq}_\ell,
    \)
    where $\ell \leq N \tau 2^\tau$,
    $\Theta_i \subseteq \tp(x,z_i)$ for $i = 0,1,\dotsc,\ell$,
    and $\Theta$ consists of all $\theta_0 \in \Theta_0$
    with $\theta_0 \rightsquigarrow (\Theta_i)_{i=1}^\ell$;
  \item
    $P_{\Theta_1 \isect \Theta_2}(x,y) \gets \bigwedge_{i=1}^2 P_{\Theta_i}(x,y)$
    with $\Theta_i \subseteq \tp(x,y)$;
  \item
    $\mathsf{goal}(\vec{v}) \gets P_\Theta(x,y)$,
    where $\Theta \subseteq \tp(x,y)$,
    and $q(\vec{v}) \in \theta$ for each $\theta \in \Theta$;
  \item
    $\mathsf{goal}(\vec{v}) \gets P_\emptyset(x)$.
  \end{enumerate}
  In addition, for each set $\Theta \subseteq \tp(x,y)$,
  the program contains $P_{\Theta_{x,z_i}}(x,y) \gets P_\Theta(x,y)$
  and $P_\Theta(x,z_0) \gets P_{\Theta_{x,z_0}}(x,z_0)$
  to rename variables in types,
  and $P_{\Theta_{y,x}}(y,x) \gets P_\Theta(x,y)$ to swap positions of variables.

  Intuitively,
  $\Pi$ computes for each guarded tuple $(a,b)$ of an instance $\Dmf$
  a set $\Theta(a,b)$ of $(x,y)$-types $\theta_0$
  that contain all information
  about the atomic formulas that hold in $\Dmf_{|\set{a,b}}$,
  and for each collection $(a,c_1),\dotsc,(a,c_\ell)$
  of $\ell \leq N \tau 2^\tau$ guarded tuples of $\Dmf$
  that have a non-empty intersection with $(a,b)$
  there are types $\theta_i \in \Theta(a,c_i)$
  such that $\bigcup_{i=0}^\ell (\theta_i)_{x,z_i}$
  is realizable in a model of $\Omc$.
  The renaming of the variables in each $\theta_i$ is necessary
  to account for the overlap of the tuples $(a,b)$ and $\vec{c}_i$.
  $\Pi$ derives $\mathsf{goal}(\vec{c})$ if some $\Theta(a,b)$ is empty
  (which will happen if $\Dmf$ is inconsistent w.r.t.\ $\Omc$)
  or if $\Theta(a,b)$ contains $q(\vec{v})$ for some $(a,b)$,
  where the $i$th variable in $\vec{v}$ is $x$
  if the $i$th position of $\vec{a}$ is $a$,
  and it is $y$ otherwise.

  It remains to show that $\Pi$ is a Datalog$^{\neq}$-rewriting
  for $q$ w.r.t.\ $\Omc$.
  To this end, let $\Dmf$ be an instance.
  We show that $\Omc,\Dmf \not\models q(\vec{a})$
  iff $\Dmf \not\models \Pi(\vec{a})$.

  The ``only if'' direction is similar to part 1.
  For the ``if'' direction, assume $\Dmf \not\models \Pi(\vec{a})$.
  Recall the definition of the uGF$_2$-unravelling of $\Dmf$
  from Section~\ref{sec:unravelling-tolerance},
  and let $G_{\vec{a}}$ be a maximal guarded set of $\Dmf$
  containing the elements of $\vec{a}$.
  As in part 1, it suffices to construct a model $\Bmf$ of $\Dmf^u$ and $\Omc$
  with $\Bmf \not\models q(\vec{b})$,
  where $\vec{b}$ is the copy of $\vec{a}$ in $\text{bag}(G_{\vec{a}})$.

  Observe that for each guarded tuple $\vec{c}$ in $\Dmf$
  there is a unique minimal set $\Theta(\vec{c}) \subseteq \tp(x,y)$
  such that $\Pi$ derives $P_{\Theta(\vec{c})}(\vec{c})$ on input $\Dmf$.
  Since $\Pi$ does not derive $\mathsf{goal}(\vec{a})$,
  we must have $\Theta(\vec{c}) \neq \emptyset$
  for all guarded tuples $\vec{c}$ in $\Dmf$.
  Furthermore, if $\vec{a} = f(\vec{v})$
  for a mapping from the variables in $\vec{v}$ to the constants in $\vec{c}$,
  then $q(\vec{v}) \notin \theta$ for some $\theta \in \Theta(\vec{c})$.
  Let us denote the set of such types in $\Theta(\vec{c})$
  by $\Theta^{\lnot q}(\vec{c})$.

  To construct the desired model $\Bmf$,
  we first assign to each maximally guarded tuple $\vec{c}$ of $\Dmf^u$
  a type in $\Theta(\vec{c}^\uparrow)$ and a model $\Amf_t$ of $\Omc$ as follows.
  Recall the definition of the tree $T(\Dmf)$
  and the bags $\bag(t)$, for $t \in T(\Dmf)$,
  used in the definition of the uGF$_2$-unravelling
  in Section~\ref{sec:unravelling-tolerance}.
  We inductively assign to each node $t$ of $T(\Dmf)$
  a $(x,y)$-type $\theta_t$ as follows.
  If $t = \set{a,b}$ for a maximal guarded set $\set{a,b}$ in $\Dmf$,
  let $\theta_t$ be any type in $\Theta^{\lnot q}(a,b)$.
  For the induction step,
  assume that we have assigned a type $\theta_t$ to $t \in T(\Dmf)$,
  but that $\theta_{t'}$ is undefined for each node $t' = t G$ in $T(\Dmf)$.
  Let $\text{tail}(t) = \set{a,b}$.
  We distinguish the following two cases.

  \medskip\noindent
  \textit{Case 1: There exists a $t_0 \in T(\Dmf)$ with $t = t_0 \set{a,b}$.}
  In this case, there is only one kind of successor node of $t$ in $T(\Dmf)$:
  nodes of the form $t \set{a,c}$,
  where $a$ does not occur in $\text{tail}(t_0)$.
  Let $c_1,\dotsc,c_n$ be an enumeration of all elements of $\dom(\Dmf)$
  such that $t_i \isdef t \set{a,c_i}$ is a node in $T(\Dmf)$.
  Denote by $\sim$ the equivalence relation on $\set{c_1,\dotsc,c_n}$
  with $c_i \sim c_j$ iff $\Theta(a,c_i) = \Theta(a,c_j)$.
  There are $s \leq 2^\tau$ equivalence classes w.r.t.\ $\sim$.
  Let $E_1,\dotsc,E_s$ be an enumeration of these classes.
  For each $i \in \set{1,\dotsc,s}$,
  pick an arbitrary subset $E'_i$ of $E_i$ of size $N\tau$,
  or the entire set if $E_i$ has fewer than $N\tau$ elements.
  Let $c'_1,\dotsc,c'_\ell$ be an enumeration
  of the elements in $E'_1 \cup \dotsb \cup E'_s$.
  Then, $\ell \leq N \tau 2^\tau$.

  Since $\theta_t \in \Theta(a,b)$,
  the rules in line~2 of the definition of $\Pi$
  ensure $\theta_t \rightsquigarrow (\Theta_i)_{i=1}^\ell$,
  where $\Theta_i \isdef \set{\theta_{x,z_i} \mid \theta \in \Theta(a,c'_i)}$.
  Hence, for each $i \in \set{1,\dotsc,\ell}$
  there is a type $\theta_i \in \Theta(a,c'_i)$
  such that $\theta_t \cup \bigcup_{i=1}^\ell (\theta_i)_{x,z_i}$
  is realizable in a model of $\Omc$.
  Let $\theta_{t\set{a,c'_i}} \isdef \theta_i$
  for each $i \in \set{1,\dotsc,\ell}$.

  It remains to assign types to the elements in $E_i \setminus E'_i$,
  for each $i \in \set{1,\dotsc,s}$.
  By construction, we have $E_i \setminus E'_i \neq \emptyset$
  only if $\card{E'_i} = N\tau$.
  Thus there must be at least one type $\theta_i^*$
  that is assigned to at least $N$ of the nodes $t \set{a,c_i}$.
  We define $\theta_{t\set{a,c}} \isdef \theta_i^*$
  for each $c \in E_i \setminus E'_i$.

  This finishes the assignment of types $\theta_{t'}$
  to all successor nodes $t'$ of $t$.
  Note that by our choice of $N$ and the type $\theta_i^*$,
  the set $\theta \isdef \theta_t \cup \bigcup_{i=1}^n (\theta_{t_i})_{x,z_i}$
  is realizable in a model $\Amf_t$ of $\Omc$.
  In fact, $\Amf_t$ can be chosen such that
  the assignment $\pi_t$ with $\pi_t(x)^\uparrow = a$
  and $\pi_t(z_i)^\uparrow = c_i$
  realizes $\theta$ in $\Amf_t$.

  \medskip\noindent
  \textit{Case 2: $t = \set{a,b}$.}
  In this case, there are two kinds of successor nodes of $t$:
  nodes of the form $t \set{a,c}$
  and nodes of the form $t \set{b,c}$.
  Let $t_1,\dotsc,t_n$ be an enumeration of all nodes of the form $t \set{a,c}$,
  and let $u_1,\dotsc,u_m$ be an enumeration
  of all nodes of the form $t \set{b,c}$.
  We assign types $\theta_{t_i}$ to each node $t_i$ as in Case 1.
  We do the same for all successors nodes $u_j$,
  but here we have to use the rule $P_{\Theta_{y,x}}(y,x) \gets P_\Theta(x,y)$
  which implies that
  $\Theta(b,a) = \set{\theta_{y,x} \mid \theta \in \Theta(a,b)}$.
  One can now show that
  \(
    \theta \isdef
    \theta_t \cup
    \bigcup_{i=1}^n (\theta_{t_i})_{x,z_i} \cup
    \bigcup_{i=1}^m (\theta_{u_i})_{y,z'_i}
  \)
  is realizable in a model $\Amf_t$ of $\Omc$,
  and that this model can be chosen in such a way
  that the assignment $\pi_t$ with $\pi_t(x)^\uparrow = a$,
  $\pi_t(z_i)^\uparrow \in \text{tail}(t_i) \setminus \set{a}$
  and $\pi_t(u_j)^\uparrow \in \text{tail}(u_j) \setminus \set{b}$
  realizes $\theta$ in $\Amf_t$.

  \medskip\noindent
  This concludes the induction step.

  \medskip
  We are now ready to construct $\Bmf$.
  Note that any two $\Amf_t$ and $\Amf_{t'}$
  with $t'$ a neighbor of $t$ in $T(\Dmf)$
  coincide on all atoms that only involve elements of $\dom(\Dmf^u)$.
  Let $\Bmf_0$ be the collection of all atoms that occur in some $\Amf_t$
  and only involve elements of $\dom(\Dmf_u)$.
  For each $a \in \dom(\Dmf^u)$
  we now attach the following structure $\Bmf_{\set{a}}$ to $\Bmf_0$.
  Pick any $t \in T(\Dmf)$ such that $a \in \dom(\bag(t))$.
  Using the construction
  in the second part of the proof of Lemma~\ref{lem:forestmodel},
  we unfold $\Amf_t$ at $a$ into a cg-decomposable model $\Bmf_{\set{a}}$,
  where we identify $a$ and its neighbors $b$
  with the copies of $a$ and $b$ in the bag of the root and its neighbors.
  We assume that all other elements have been renamed
  so that they do not occur in $\dom(\Dmf^u)$.
  Let
  \[
    \Bmf \,\isdef\, \Bmf_0 \cup \bigcup_{a \in \dom(\Dmf^u)} \Bmf_{\set{a}}.
  \]
  It is not difficult to prove that $\Bmf$ is a model of $\Dmf^u$ and $\Omc$
  with $\Bmf \not\models q(\vec{b})$.
\end{proof}

\section{Proofs for Section 5}
\begin{trivlist}\item
\textbf{Theorem~\ref{thm:infiniteunravel} (restated)}~\itshape
Let $\Omc$ be  either a uGF$(1)$, uGF$^{-}(1,=)$,
uGF$^{-}_{2}(2)$, uGC$^{-}_{2}(1,=)$, or $\mathcal{ALCHIF}$ ontology
of depth~2.
If $\Omc$ is materializable for (possibly infinite)
cg-tree decomposable instances $\Dmf$ with $\text{sig}(\Dmf) \subseteq \text{sig}(\Omc)$,
then $\Omc$ is unravelling tolerant.
\end{trivlist}
\begin{proof}

We first give the proof for the ontology languages uGF$(1)$, uGF$_{2}^{-}(2)$ and uGF$^{-}(1,=)$.
Assume such an ontology $\Omc$ is given. Let $\Dmf$ be an instance and $\Dmf^{u}$ its uGF-unravelling.
Let $G_{0}$ be a maximal guarded set in $\Dmf$, $\vec{a}$ in $G_{0}$, $\vec{b}$ the copy of $\vec{a}$ in $\text{bag}(G_{0})$,
and $q$ an rAQ. We have to show that $\Omc,\Dmf\models q(\vec{a})$ implies $\Omc,\Dmf^{u}\models q(\vec{b})$.

Define an equivalence relation $\sim$ on $T(\Dmf)$ by setting $t\sim t'$ if $\text{tail}(t)=\text{tail}(t')$.
Recall that for any $t,t'\in T(\Dmf)$ with $t\sim t'$ the mapping $h_{t,t'}$ that sends
every $e\in \text{dom}(\text{bag}(t))$ to the unique $f\in \text{dom}(\text{bag}(t'))$ with
$e^{\uparrow}= f^{\uparrow}$ is an isomorphism from $\text{bag}(t)$ to $\text{bag}(t')$. Call
$h_{t,t'}$ the \emph{canonical isomorphism} from $\text{bag}(t)$ onto $\text{bag}(t')$. By construction
of the undirected graph $(T(\Dmf),E)$,
for any $t,t'\in T(\Dmf)$ with $t\sim t'$ there is an
automorphism $i_{t,t'}$ of $(T(\Dmf),E)$ such that $i_{t,t'}(t)=t'$ and $i_{t,t'}(s)\sim s$ for every
$s\in T(\Dmf)$. $i_{t,t'}$ is uniquely determined on the connected component of $t$ in $(T(\Dmf),E)$ and induces
the \emph{extended canonical automorphism} $\hat{h}_{t,t'}$ of $\Dmf^{u}$ by setting
$\hat{h}_{t,t'} = \bigcup_{s\in T(\Dmf)}h_{s,i_{t,t'}(s)}$.

\medskip
\noindent
{\bf Fact~1.} If $t,t'\in T(\Dmf)$ such that $t\sim t'$ then $\hat{h}_{t,t'}$ is an automorphism of $\Dmf^{u}$.

\bigskip

\noindent
Our aim now is to construct a materialization $\Bmf$ of $\Omc$ and $\Dmf^{u}$ that is a forest model
such that the automorphism $\hat{h}_{t,t'}$ can be lifted to an automorphism of $\Bmf$. Once this is done
it is straightforward to construct a materialization $\Bmf'$ of $\Omc$ and $\Dmf$ by hooking to any maximal guarded set $G$ of $\Dmf$ a
copy of $\Bmf_{\text{bag}(G)}$, the guarded tree decomposable subinstances of $\Bmf$ hooked to $\text{bag}(G)$ in $\Bmf$.
Then $(\Bmf',\vec{a})$ and $(\Bmf,\vec{b})$ are guarded bisimilar and if $\Omc,\Dmf^{u}\not\models q(\vec{b})$ then
$\Omc,\Dmf\not\models q(\vec{a})$ follows.

Let $\Bmf$ be a materialization of $\Omc$ and $\Dmf^{u}$.
(To show that $\Bmf$ exists let $\text{red}(\Dmf^{u})$ be the $\text{sig}(\Omc)$-reduct of
$\Dmf$. As $\Omc$ is invariant under disjoint unions and materializable
for the class of (possibly infinite) cg-tree decomposable instances $\Dmf$ with
$\text{sig}(\Dmf) \subseteq \text{sig}(\Omc)$ there exists a materialization
$\Bmf_{\text{red}}$ of $\text{red}(\Dmf^{u})$. Clearly
$\{ R \mid R(\vec{a})\in \Bmf_{\text{red}}\}\subseteq \text{sig}(\Omc)$.
Now let
$$
\Bmf= \Bmf_{\text{red}} \cup \{R(\vec{a}) \in \Dmf^{u}\mid R\not\in \text{sig}(\Omc)\}
$$
One can show that $\Bmf$ is a materialization of $\Dmf^{u}$ and $\Omc$.)

Fact~1 entails the following.
  %
  \\[2mm]
  {\bf Fact~2.}
   For all $t,t'$ with $t\sim t'$ and $\vec{e}$ in $\text{dom}(\text{bag}(t))$ and
   any rAQ $q(\vec{x})$ such that $\vec{x}$ has the same length as $\vec{e}$:
$$
\Bmf\models q(\vec{e}) \quad \Leftrightarrow \quad \Bmf\models q(h_{t,t'}(\vec{e}))
$$
We now distinguish two cases.

\medskip

\noindent
\emph{Case 1.} $\Omc$ is a uGF$(1)$ or a uGF$_{2}^{-}(2)$ ontology. The following observation is crucial for the proof
(and does not hold for languages with equality such as uGF$^{-}(1,=)$).

\medskip
\noindent
{\bf Observation 1.} If there is a homomorphism $h$ from an instance $\Dmf$ to
an instance $\Dmf'$ then $\Omc,\Dmf\models q(\vec{a})$ implies $\Omc,\Dmf'\models q(h(\vec{a}))$
for any CQ $q$ and $\vec{a}$ in $\text{dom}(\Dmf)$.

\medskip
\noindent

We hook in $\Dmf^{u}$ to any $\text{bag}(t)$ with $t\in T(\Dmf)$ a copy of any rAQ $q$ (regarded as an instance)
from which there is a homomorphism into $\Bmf$
that is injective on the root bag of $q$ and maps it into $\text{dom}(\text{bag}(t))$.
In more detail, let $X_{t}$ be the set of all pairs $(\pi,\Dmf_{q})$ such that $\Dmf_{q}$ is an instance corresponding
to an rAQ $q$ and $\pi$ is a homomorphism from $\Dmf_{q}$ to $\Bmf$ mapping the constants $d_{x}$ corresponding to
answer variables $x$ of $q$ to distinct $\pi(d_{x})\in \text{dom}(\text{bag}(t))$.
By renaming constants in $\Dmf_{q}$ we obtain an instance $\Dmf_{q,\pi}$ isomorphic to $\Dmf_{q}$ such that
$\pi(d_{x})= d_{x}$ and such that $\text{dom}(\Dmf_{q,\pi})\cap \text{dom}(\Dmf^{u})$ is the set of all $d_{x}$
with $x$ an answer variable of $q$. Now let
$$
\Dmf_{t}=\bigcup_{(q,\pi)\in X_{t}}\Dmf_{q,\pi},
\quad
\Dmf^{u+}= \Dmf^{u} \cup \bigcup_{t\in T(\Dmf)}\Dmf_{t}
$$
The following properties of $\Dmf^{u+}$ follow directly from the definition:
\begin{enumerate}
\item For any $t,t'\in T(\Dmf)$ with $t\sim t'$ there is an isomorphism from $\Dmf_{t}$ onto $\Dmf_{t'}$
that extends the canonical isomorphism $h_{t,t'}$;
\item there is a homomorphism from $\Dmf^{u+}$ into $\Bmf$ preserving $\text{dom}(\Dmf^{u})$.
Thus, by Observation~1, any materialization of $\Dmf^{u+}$ is a materialization of $\Dmf^{u}$ and
for every CQ $q(\vec{x})$ and $\vec{d}$ in $\Dmf^{u}$ of the same length as $\vec{x}$:
$$
\Omc,\Dmf^{u+}\models q(\vec{d}) \quad \Leftrightarrow \quad \Omc,\Dmf^{u}\models q(\vec{d})
$$
\end{enumerate}
Now take a materialization $\Bmf^{u+}$ of $\Dmf^{u+}$ and $\Omc$ that is a forest model of $\Dmf^{u+}$
and $\Omc$.
Thus $\Bmf^{u+}$ is obtained from $\Dmf^{u+}$ by hooking cg-tree decomposable
models $\Bmf_{t}^{u+}$ of $\Dmf_{t}$ to every $\text{bag}(t)$ with $t\in T(\Dmf)$.
By Point~1 and Fact~1, we have the following:

\medskip
\noindent
{\bf Fact~3.} For any $t,t'\in T(\Dmf)$ with $t\sim t'$ the canonical automorphism $\hat{h}_{t,t'}$
extends to an isomorphism from $\Bmf^{u+}_{|\text{dom}(\Dmf_{t})}$ onto $\Bmf^{u+}_{|\text{dom}(\Dmf_{t'})}$.

\medskip
\noindent
%
{\bf Fact 4.} For any $t,t'\in T(\Dmf)$ with $t\sim t'$ and any finite subinstance $\Amf$ of $\Bmf^{u+}_{t}$
there exists an isomorphic embedding of $\Amf$ into $\Bmf^{u+}_{|\text{dom}(\Dmf_{t'})}$ extending
the canonical automorphism $\hat{h}_{t,t'}$.

\medskip
\noindent
We prove Fact~4. By Fact~3 it suffices to prove that for any $t\in T(\Dmf)$ and any finite subinstance
$\Amf$ of $\Bmf^{u+}_{t}$ there exists an isomorphic embedding of $\Amf$ into $\Bmf^{u+}_{|\text{dom}(\Dmf_{t})}$
preserving $\text{dom}(\text{bag}(t))$. But this follows from the
fact that $\Bmf^{u+}$ is a materialization of $\Dmf^{u}$ (Point~(2)): then there is an isomorphism
from $\Amf$ to some $\Dmf_{\pi,q}$ used in the construction of $\Dmf^{u+}$ which preserves $\text{dom}(\text{bag}(t))$.
Fix such a $\Dmf_{\pi,q}$. It remains to be proved that there does not exist any $R(\vec{a})$
with $\vec{a}$ in $\text{dom}(\Dmf_{\pi,q})$ such that $R(\vec{a})\in \Bmf^{u+}\setminus \Dmf_{\pi,q}$.
But using the fact that $\Amf$ is a subinstance of the model $\Bmf^{u+}$ of $\Omc$ and $\Dmf^{u+}$
isomorphic to $\Dmf_{\pi,q}$ one can easily construct a model of $\Dmf^{u+}$ and $\Omc$ that contains no
$R(\vec{a})\not\in \Dmf_{\pi,q}$ with $\vec{a}$ in $\text{dom}(\Dmf_{\pi,q})$. Thus $\Bmf^{u+}$ contains
no such $R(\vec{a})$ since $\Bmf^{u+}$ is a materialization of $\Omc$ and $\Dmf^{u+}$.
This finishes the proof of Fact~4.

Fact~4 easily generalizes to $\Bmf^{u+}$ (by Fact~3): for any $t,t'\in T(\Dmf)$ with $t\sim t'$ and any
finite subinstance $\Amf$ of $\Bmf^{u+}$ there exists an isomorphic embedding of $\Amf$ into $\Bmf^{u+}$ extending
the canonical automorphism $\hat{h}_{t,t'}$.
%
%
We now uniformize $\Bmf^{u+}$. For each $t\in T(\Dmf)$ with $t$ not a guarded set in $\Dmf$ we hook at $\text{bag}(t)$ to $\Dmf^{u}$
an isomorphic copy $\Bmf^{u\ast}_{t}$ of the interpretation $\Bmf^{u+}_{\text{bag}(G)}$ with $t\sim G$ and remove $\Bmf^{u+}_{t}$.
Denote the resulting model by $\Bmf^{u\ast}$. $\Bmf^{u\ast}$ is uniform in the sense that for any two $t,t'\in T(\Dmf)$,
the cg-tree decomposable
models hooked to $\text{bag}(t)$ and $\text{bag}(t')$ in $\Bmf^{u\ast}$ are isomorphic. We now show that $\Bmf^{u\ast}$ is a
materialization of $\Dmf$ and $\Omc$. For $a\in \text{dom}(\Bmf^{u\ast}_{t})$ and $t\sim G$ denote by $a^{\sim}$ the
corresponding element of $\Bmf^{u+}_{\text{bag}(G)}$ such that for $a\in \text{bag}(t)$ we have $a^{\uparrow}= (a^{\sim})^{\uparrow}$.

\medskip
\noindent
{\bf Fact 5.} $\Bmf^{u\ast}$ is a materialization of $\Omc$ and $\Dmf^{u}$.

\medskip
We show that $\Bmf^{u\ast}$ is a model of $\Omc$. Then $\Bmf^{u\ast}$ is a materialization
of $\Omc$ and $\Dmf^{u}$ since it is a model of $\Dmf^{u}$ and since $\Bmf^{u\ast}$ is obtained
from the materialization $\Bmf^{u+}$ of $\Omc$ and $\Dmf^{u}$ by replacing certain interpretations
that are hooked to $\text{bag}(t)$ by interpretations that are hooked to $\text{bag}(t')$ for $t,t'\in T(\Dmf)$
with $t\sim t'$ which preserves the answers to rAQs (use Fact~1).


Consider first a sentence $\varphi$ of the form $\forall \vec{y} (R(\vec{y}) \rightarrow \psi(\vec{y}))$
in $\Omc$, where $\psi$ is a formula in openGF of depth one. We show that
$\Bmf^{u\ast}\models \varphi$.
Let $\Bmf^{u\ast}\models R(\vec{a})$ for some $\vec{a}=(a_{1},\ldots,a_{k})$
in $\text{dom}(\Bmf^{u\ast})$. Then $a_{1},\ldots,a_{k}$ are contained in $\Bmf^{u\ast}_{t}$ for some $t\in T(\Dmf)$.
We show that $\Bmf^{u\ast}\models \psi(\vec{a})$ iff $\Bmf^{u+}\models \psi(\vec{a}^{\sim})$ where
$\vec{a}^{\sim}=(a_{1}^{\sim},\ldots,a_{k}^{\sim})$. Assume $t\sim G$. If $a_{1},\ldots,a_{k} \not\in \text{dom}(\Dmf^{u})$, then
this is clear by construction since the truth of $\psi(\vec{a})$ then only depends on the subinterpretation
$\Bmf_{t}^{u\ast}$ of $\Bmf^{u\ast}$ and this is isomorphic to the subinterpretation $\Bmf_{\text{bag}(G)}^{u+}$ of the model
$\Bmf^{u+}$ of $\Omc$. Now assume that $\{a_{1},\ldots,a_{k}\}\cap \text{dom}(\Dmf^{u})=Z\not=\emptyset$. By Fact~4,
for any guarded set $F$ in $\Bmf^{u\ast}$ with $G'\cap Z\not=\emptyset$ there exists a guarded set $F'$ in $\Bmf^{u\ast}$
such that there exists an isomorphism $h$ from $\Bmf^{u\ast}_{|F}$ onto $\Bmf^{u\ast}_{|F'}$ that extends the canonical
automorphism $\hat{h}_{t,G}$. The converse direction holds as well: let $Z'=\{a_{1}^{\sim},\ldots,a_{k}^{\sim}\}\cap \text{dom}(\Dmf^{u})$.
By Fact~4,
for any guarded set $F'$ in $\Bmf^{u\ast}$ with $F'\cap Z'\not=\emptyset$ there exists a guarded set $F$ in $\Bmf^{u\ast}$
such that there exists an isomorphism $h$ from $\Bmf^{u\ast}_{|F'}$ onto $\Bmf^{u\ast}_{|F}$ that extends the canonical
automorphism $\hat{h}_{G,t}$.Thus $\Bmf^{u\ast}\models \psi(\vec{a})$ iff
$\Bmf^{u+}\models \psi(\vec{a}^{\sim})$ follows.

For sentences $\varphi$ of the form $\forall x\psi(x)$ in $\Omc$, where $\psi(x)$ is an openGF formula of
depth two using at most binary relations the argument is similar using the fact that guarded sets $\{a,b\}$ are either
completely contained in $\text{dom}(\Dmf^{u})$ or contain at most one element from $\text{dom}(\Dmf^{u})$.
This finishes the proof of Fact~5.

\medskip
Finally we hook for any maximal guarded $G$ in $\Dmf$ the interpretation $\Bmf^{u+}_{\text{bag}(G)}$ to $G$ in the original
instance $\Dmf$ and obtain a forest model $\Bmf^{+}$. It is straightforward to prove that for any maximal
guarded set $G$ in $\Dmf$, any tuple $\vec{e}$ containing all elements of $G$, and the copy $\vec{f}$ of $\vec{e}$ in
$\text{bag}(G)$ there is a connected guarded bisimulation between $(\Bmf^{u\ast},\vec{f})$ and
$(\Bmf^{+},\vec{e})$. It follows that $\Bmf^{+}$ is a model of $\Omc$ and $\Dmf$ and that $\Bmf^{+}\not\models q(\vec{a})$
if $\Bmf^{u\ast}\not\models q(\vec{b})$.

\medskip

\noindent
\emph{Case 2.} $\Omc$ is a uGF$^{-}(1,=)$ ontology. In this case the construction is
simpler as we do not modify $\Bmf$ further. There is no need to manipulate
$\Bmf$ as we are in a fragment of depth one in which the outermost universal quantifier is guarded by an equality.
Define $\Bmf^{+}$ by adding to $\Dmf$
\begin{itemize}
\item all atoms $R(a_{1}^{\uparrow},\ldots,a_{k}^{\uparrow})$ such that $R(a_{1},\ldots,a_{k})\in \Bmf$ and
$a_{1},\ldots,a_{k} \in \text{dom}(\Dmf^{u})$;
\item for any $a\in \text{dom}(\Dmf)$ for a fixed copy $a'$ of $a$ in $\Dmf^{u}$ a copy $\Bmf_{a'}$ of $\Bmf$ that is
hooked to $\Dmf$ at $a$ by identifying $a'$ and $a$.
\end{itemize}
Using Fact~1 and the condition that $\Omc$ is a uGF$^{-}(1,=)$ ontology it is straightforward
to prove that $\Bmf^{+}$ is a model of $\Omc$ and $\Dmf$. Also by Fact~1 and construction $\Bmf^{+}\not\models q(\vec{a})$
if $\Bmf\not\models q(\vec{b})$.

\medskip

We now assume that $\Omc$ is a uGC$^{-}_{2}(1,=)$ or a $\mathcal{ALCHIF}$ ontology
of depth~2. Let $\Dmf$ be an instance, $G_{0}$ be a maximal guarded set in $\Dmf$, $\vec{a}$ in $G_{0}$, $\vec{b}$ the copy of
$\vec{a}$ in $\text{bag}(G_{0})$, and $q$ an rAQ. We have to show that $\Omc,\Dmf\models q(\vec{a})$ implies
$\Omc,\Dmf^{u}\models q(\vec{b})$. Observe that now we have to consider the uGC$_{2}$-unravelling
rather than the uGF-unravelling of $\Dmf$. First
we establish again the existence of certain automorphisms of $\Dmf^{u}$. In this case, however, they are \emph{not}
induced by automorphisms of the tree $(T(\Dmf),E)$ but are determined directly on the interpretation.
Recall that for any $t,t'\in T(\Dmf)$ with $t\sim t'$ the mapping $h_{t,t'}$ that sends
every $e\in \text{dom}(\text{bag}(t))$ to the unique $f\in \text{dom}(\text{bag}(t'))$ with
$e^{\uparrow}= f^{\uparrow}$ is an isomorphism from $\text{bag}(t)$ to $\text{bag}(t')$. Call
$h_{t,t'}$ the \emph{canonical isomorphism} from $\text{bag}(t)$ onto $\text{bag}(t')$.
One can easily prove the existence of an \emph{extended canonical automorphism}
$\hat{h}_{t,t'}$ of $\Dmf^{u}$ that extends $h_{t,t'}$.

\medskip
\noindent
{\bf Fact~1.} If $t,t'\in T(\Dmf)$ such that $t\sim t'$ there is an automorphism $\hat{h}_{t,t'}$ of $\Dmf^{u}$ that
extends $h_{t,t'}$.

\bigskip

\noindent
Let $\Bmf$ be a materialization of $\Omc$ and $\Dmf^{u}$. Fact~1 entails the following.
  %
  \\[2mm]
  {\bf Fact~2.}
   For all $t,t'$ with $t\sim t'$ and $\vec{e}$ in $\text{dom}(\text{bag}(t))$ and
   any rAQ $q(\vec{x})$ such that $\vec{x}$ has the same length as $\vec{e}$:
$$
\Bmf\models q(\vec{e}) \quad \Leftrightarrow \quad \Bmf\models q(h_{t,t'}(\vec{e}))
$$
We now distinguish two cases.

\medskip\noindent
\emph{Case 3.} $\Omc$ is a uGC$_{2}^{-}(1,=)$ ontology. This case is similar to Case~2. No further modification of $\Bmf$ is needed.
We may assume that $\Bmf$ is obtained from $\Dmf^{u}$ by hooking cg-tree decomposable interpretations $\Bmf_{c}$
to $c$ for every $c\in \text{dom}(\Dmf^{u})$ and by adding atoms $R(c,d)$ to $\Dmf^{u}$ for distinct $c,d\in \text{dom}(\Dmf^{u})$.
Define $\Bmf^{+}$ by adding to $\Dmf$
\begin{itemize}
\item all atoms $R(a_{1}^{\uparrow},\ldots,a_{k}^{\uparrow})$ such that $R(a_{1},\ldots,a_{k})\in \Bmf$ and
$a_{1},\ldots,a_{k} \in \text{dom}(\Dmf^{u})$;
\item for any $a\in \text{dom}(\Dmf)$ for a fixed copy $a'$ of $a$ in $\Dmf^{u}$ a copy of $\Bmf_{a'}$ that is
hooked to $\Dmf$ at $a$ by identifying $a'$ and $a$.
\end{itemize}
Using Fact~1 and the condition that $\Omc$ is a uGC$_{2}^{-}(1,=)$ ontology it is straightforward
to prove that $\Bmf^{+}$ is a model of $\Omc$ and $\Dmf$. Also by Fact~1 and construction $\Bmf^{+}\not\models q(\vec{a})$
if $\Bmf\not\models q(\vec{b})$.

\medskip
\noindent
\emph{Case 4.} $\Omc$ is a $\mathcal{ALCHIF}$ ontology of depth~2. The proof that follows
is similar to Case~1, but one cannot hook to any $\text{bag}(t)$ a copy of any rAQ from which there is a
homomorphism into $\Bmf$ that is injective on the root bag of $q$ and maps it into $\text{dom}(\text{bag}(t))$ as this can
lead to violations of functionality. Two modifications are needed: firstly,
we do not independently hook interpretations to bags $\text{bag}(t)$ with two elements in
$\Dmf^{u}$. This is to avoid violations of functionality due to guarded sets $G_{1}$ and $G_{2}$ with $G_{1}\cap G_{2}=\{d\}$
when the interpretations we hook to $G_{1}$ and $G_{2}$ independently add an $R$-successor to $d$ for a function
$R$ (this has already been done in Case~3). Secondly, we cannot hook arbitrarily many rAQs to bags as this will
lead to violations of functionality as well.

We may again assume that $\Bmf$ is obtained from $\Dmf^{u}$ by hooking cg-tree decomposable interpretations $\Bmf_{c}$
to $c$ for every $c\in \text{dom}(\Dmf^{u})$ and by adding atoms $R(c,d)$ to $\Dmf^{u}$ for distinct $c,d\in \text{dom}(\Dmf^{u})$.
Observe that $G_{\Bmf_{c}}= \{ \{a,b\} \mid R(a,b)\in \Bmf_{c},a\not=b\}$ is an undirected tree. We call $c$ its root and in
this proof call such an interpretation a \emph{tree interpretation with root $c$}.
We have to modify $\Dmf^{u}$ to be able to uniformize. For any $c\in \text{dom}(\Dmf^{u})$ we define the tree instance
$\Dmf_{c}$ with root $c$ as follows.
Let $\Dmf_{q}$ be the instance corresponding to an rAQ $q = q(x)\leftarrow \phi$ with a single answer
variable $x$ and a single additional variable $y$ such that there is an injective
homomorphism $h$ from $\Dmf_{q}$ to $\Bmf$ mapping $x$ to some $c$ in $\text{dom}(\Dmf^{u})$
and such that neither $R(h(x),h(y))\in \Bmf$ nor $R(h(y),h(x))\in \Bmf$ for any $R$ that is functional in $\Omc$.
Then $\Dmf_{c}$ contains a copy of $\Dmf_{q}$ obtained by identifying the variable $x$ with $c$. Set
$$
\Dmf^{u+} = \{ R(\vec{a}) \in \Bmf \mid \vec{a}\subseteq \text{dom}(\Dmf^{u})\} \cup
\bigcup_{c\in \text{dom}(\Dmf^{u})}\Dmf_{c}.
$$
The following properties of $\Dmf^{u+}$ follow directly from the definition and standard properties
of $\mathcal{ALCHIF}$:
\begin{enumerate}
\item For any $c,d\in \text{dom}(\Dmf^{u})$ with $c^{\uparrow}=d^{\uparrow}$ there is an isomorphism from $\Dmf_{c}$ onto $\Dmf_{d}$
mapping $c$ to $d$;
\item there is a homomorphism from $\Dmf^{u+}$ into $\Bmf$ preserving $\text{dom}(\Dmf^{u})$ and functionality.
Thus, in particular, any materialization of $\Dmf^{u+}$ is a materialization of $\Dmf^{u}$ and
for every CQ $q(\vec{x})$ and $\vec{d}$ in $\Dmf^{u}$ of the same length as $\vec{x}$:
$$
\Omc,\Dmf^{u+}\models q(\vec{d}) \quad \Leftrightarrow \quad \Omc,\Dmf^{u}\models q(\vec{d})
$$
\end{enumerate}
Now take a materialization $\Bmf^{u+}$ of $\Dmf^{u+}$ and $\Omc$.
Thus $\Bmf^{u+}$ is obtained from $\Dmf^{u+}$ by hooking tree-interpretations $\Bmf_{c}^{u+}$ that are models of $\Dmf_{c}$
to every $c\in \text{dom}(\Dmf^{u})$. By Point~1 and Fact~1 and the properties of $\mathcal{ALCHIF}$ we have the following:

\medskip
\noindent
{\bf Fact~3.} For any $c,d\in \text{dom}(\Dmf^{u})$ with $c^{\uparrow}=d^{\uparrow}$, $c\in \text{dom}(\text{bag}(t))$,
and $d\in \text{dom}(\text{bag}(t'))$ such that $t\sim t'$ the canonical automorphism $\hat{h}_{t,t'}$
extends to an isomorphism from $\Bmf^{u+}_{|\text{dom}(\Dmf_{c})}$ onto $\Bmf^{u+}_{|\text{dom}(\Dmf_{d})}$.

\medskip
\noindent
The following fact can now be proved by modifying in a straightforward way the proof of Fact~4 above.
%

\medskip
\noindent
{\bf Fact 4.} For any $c,d\in \text{dom}(\Dmf^{u})$ with $c^{\uparrow}=d^{\uparrow}$ and any $e\in \text{dom}(\Bmf^{u+}_{c})$
there exists an isomorphic embedding of $\Bmf^{u+}_{|\{c,e\}}$ into $\Bmf^{u+}_{|\text{dom}(\Dmf_{d})}$.

\medskip
\noindent
Let for $c,d\in \text{dom}(\Dmf^{u})$, $c\sim d$ if $c^{\uparrow}=d^{\uparrow}$. Fix for any equivalence class $[c]=\{ d \mid d\sim c\}$
a unique $c_{\sim}\in [c]$. We now uniformize $\Bmf^{u+}$. For each $d\in \text{dom}(\Dmf^{u})$ we hook at $d$ to $\Dmf^{u}$
an isomorphic copy $\Bmf^{u\ast}_{d}$ of the interpretation $\Bmf^{u+}_{c_{\sim}}$ with $c_{\sim}^{\uparrow}=d^{\uparrow}$
and remove $\Bmf^{u+}_{d}$. Denote the resulting model by $\Bmf^{u\ast}$. $\Bmf^{u\ast}$ is uniform in the sense that for any
$c\sim d$ the interpretations $\Bmf_{c}^{u\ast}$ hooked to $c$ in $\Dmf^{u}$
and $\Bmf_{d}^{u\ast}$ hooked to $d$ in $\Dmf^{u}$ are isomorphic. One can now prove similary to the proof of
Fact~5 above the following:

\medskip
\noindent
{\bf Fact 5.} $\Bmf^{u\ast}$ is a materialization of $\Omc$ and $\Dmf^{u}$.

\medskip
\noindent
It remains construct the materialization $\Bmf^{+}$ of $\Dmf$. To this end we hook to any $c\in \text{dom}(\Dmf)$
the tree interpretation $\Bmf^{u\ast}_{d}$ with $d^{\uparrow}=c$. In addition we include in $\Bmf^{+}$ all
$R(a_{1}^{\uparrow},\ldots,a_{k}^{\uparrow})$ such that $R(a_{1},\ldots,a_{k})\in \Bmf^{u+}$ and
$a_{1},\ldots,a_{k} \in \text{dom}(\Dmf^{u})$. Using Fact~5 one can show that $\Bmf^{+}$ is a model of $\Omc$
and $\Dmf$ such that $\Bmf^{+}\not\models q(\vec{a})$ if $\Bmf\not\models q(\vec{b})$.
\end{proof}

\section{Proofs for Section 6}

\begin{trivlist}\item
  \textbf{Theorem~\ref{thm:csphard} (restated)}~\itshape
 For any of the following ontology languages, CQ-evaluation w.r.t.\
  \Lmc is CSP-hard: uGF$_{2}(1,=)$, uGF$_{2}(2)$, uGF$_2(1,f)$, and
  the class of $\mathcal{ALCF}_{\ell}$ ontologies of depth~2.
\end{trivlist}
\begin{proof}
We provide additional details of the proof for uGF$_{2}(1,=)$.
Recall that $\Omc$ contains
$$
\begin{array}{ll}
\multicolumn{2}{l}{\displaystyle\forall x (\bigwedge_{a\not=a'}\neg(\varphi_{a}^{\not=}(x)\wedge \varphi_{a'}^{\not=}(x))
\wedge \bigvee_{a}\varphi_{a}^{\not=}(x))}\\
\forall x (A(x) \rightarrow \neg \varphi_{a}^{\not=}(x)) & \text{when $A(a)\not\in \Amf$}\\
\forall xy (R(x,y) \rightarrow \neg (\varphi_{a}^{\not=}(x) \wedge
\varphi_{a'}^{\not=}(y))) & \text{when $R(a,a')\not\in \Amf$}\\[0.6mm]
\forall x \varphi_{a}^{=}(x) & \text{for all $a\in \text{dom}(\Amf)$}
\end{array}
$$
where $A$ and $R$ range over symbols in $\text{sig}(\Amf)$ of the
respective arity. We first show that coCSP$(\Amf)$ polynomially reduces to the query evaluation problem for
$(\Omc,q\leftarrow N(x))$. Assume $\Dmf$ with $\text{sig}(\Dmf)\subseteq \text{sig}(\Amf)$ is given and let
$$
\Dmf'=\Dmf \cup \{R_{a}(d,d') \mid P_{a}(d)\in \Amf\},
$$
where the relations $P_{a}\in \text{sig}(\Amf)$ determine the precoloring and $d'$ is a fresh labelled
null for each $d\in \text{dom}(\Dmf)$. We show that $\Dmf\rightarrow \Amf$ iff $\Omc,\Dmf'\not\models q$.
First let $h$ be a homomorphism from
$\Dmf$ to $\Amf$. Define a model $\Bmf$ of $\Dmf'$ and $\Omc$ by adding to $\Dmf'$ for any $d\in \text{dom}(\Dmf)$ with
$h(d) = a$ and infinite chain
$$
R_{a}(d_{0,d},d_{1,d}),R_{a}(d_{1,d},d_{2,d}),\ldots
$$
with $d_{0,d}=d$ and fresh labelled nulls $d_{i,d}$ for $i>0$. Also add
$R_{a}(d,d)$ to $\Dmf$ for all $d\in \text{dom}(\Dmf)$ and all labelled nulls used in the chains.
Using the definition of $\Omc$ it is not difficult to show that $\Bmf$ is a model of $\Omc$ and $\Dmf'$.
Thus $\Omc,\Dmf'\not\models q$, as required.
Now assume that $\Omc,\Dmf'\not\models q$. Then there is a model $\Bmf$ of
$\Omc$ and $\Dmf'$ such that $\Bmf\not\models q$. Define a mapping $h$
from $\Dmf$ to $\Amf$ by setting $h(d)=a$ iff there exists $d'$ with $d'\not=d$ and $R_{a}(d,d')\in \Bmf$.
Using the definition of $\Omc$ it is not difficult to show that $h$ is well defined and a homomorphism.
This finishes the proof of the polynomial reduction of coCSP$(\Amf)$ to the query evaluation problem
for $(\Omc,q\leftarrow N(x))$.

Now we show that for any rAQ $q$ the query evaluation problem for $(\Omc,q)$ can be polynomially reduced to
coCSP$(\Amf)$. We first show that there is a polynomial reduction of the problem whether an instance $\Dmf$ is
consistent w.r.t.~$\Omc$ to CSP$(\Amf)$. Assume $\Dmf$ is given. Let $\Dmf^{\bullet}$ be the
$\text{sig}(\Amf)$-reduct of $\Dmf$ extended with $P_{a}(d)$ for any $d$ with $R_{a}(d,d')\in \Dmf$ for some $d'\not=d$.
Using the definition of $\Omc$ one can show that $\Dmf$ is consistent w.r.t.~$\Omc$ iff $\Dmf^{\bullet}\rightarrow \Amf$.

Now $\Omc,\Dmf\models q(\vec{d})$ iff $\Dmf$ is not consistent w.r.t.~$\Omc$ or $\Dmf'\models q(\vec{d})$
where $\Dmf'= \Dmf\cup \{R_{a}(d,d)\mid a\in \text{dom}(\Amf),d\in \text{dom}(\Dmf)\}$. The latter problem is
in {\sc PTime}.
\end{proof}

\section{Proofs for Section 7}

\begin{trivlist}\item
  \textbf{Theorem~\ref{thm:undecidability} (restated)}~\itshape
For the ontology languages uGF$^{-}_{2}(2,f)$ and
$\mathcal{ALCIF}_{\ell}$ of depth 2, it is undecidable whether for a
given ontology \Omc,
\begin{enumerate}

\item query evaluation w.r.t.\ \Omc is in \PTime,
  Datalog$^{\not=}$-rewritable, or {\sc coNP}-hard
  (unless $\text{\sc PTime}=\text{\sc NP}$);

\item \Omc is materializable.

\end{enumerate}
\end{trivlist}
As discussed in the main part of the paper,
we prove Theorem~\ref{thm:undecidability} in two steps: we first
construct an ontology $\Omc_{\text{cell}}$ that marks lower left corners
of cells and then we construct an ontology $\Omc_{\mathfrak{P}}$ that marks
the lower left corner of grids that represent a solution to a rectangle
tiling problem $\mathfrak{P}$. We construct the ontologies in $\mathcal{ALCIF}_{\ell}$.
Thus, in addition to $\mathcal{ALCI}$ concepts we use concepts of the form $(\leq 1 R)$, $(= 1 R)$, and $(\geq 2 R)$.
The proof is given using DL notation.
%

\bigskip
\noindent
{\bf Marking the lower left corner of grid cells.}
Let $X$ and $Y$ be binary relations and $X^{-},Y^{-}$ their inverses
in $\mathcal{ALCI}$. Using the sentences
$$
\top
\sqsubseteq (\leq 1 Z)
$$
for all $Z\in \{X,Y,X^{-},Y^{-}\}$ we ensure
that in any instance $\Dmf$ that is consistent w.r.t.~our ontology the
relations $X$ and $Y$ as well as their inverses $X^{-}$ and $Y^{-}$
are functional in $\Dmf$ in the sense that $R(d,d'),R(d,d'')\in \Dmf$
implies $d'=d''$ for all $R\in \{X,Y,X^{-},Y^{-}\}$. For an instance
$\Dmf$ and $d\in \text{dom}(\Dmf)$ we write $\Dmf\models
\text{cell}(d)$ if there exist $d_{1},d_{2},d_{3}$ with $X(d,d_{1})$,
$Y(d_{1},d_{3})$, $Y(d,d_{2})$, $X(d_{2},d_{3})\in \Dmf$. Since $X$
and $Y$ are functional in $\Dmf$, $\Dmf\models \text{cell}(d)$ implies
$d_{3}=d_{4}$ if $X(d,d_{1})$, $Y(d_{1},d_{3})$, $X(d,d_{2}),
Y(d_{2},d_{4})\in \Dmf$. As a marker for all $d$ such that
$\Dmf\models \text{cell}(d)$ we use the concept $(=1P)$
for a binary relation $P$. For $P$ and all binary relations $R$ introduced below
we add the inclusion $\top \sqsubseteq \exists R.\top$
to our ontology so that when building models one can only choose between having exactly one $R$-successor
or at least two $R$-successors. To do the marking we use concepts $(=1 R_{1})$ and $(=1 R_{2})$ with binary
relation symbols $R_{1},R_{2}$ as `second-order variables', ensure
that all nodes in $\Dmf$ are contained in $(=1 R_{1}) \sqcup (=1
R_{2})$, and then state (as a first attempt) that
$$
\bigsqcup_{i=1,2}\exists X. \exists Y. (=1 R_{i}) \sqcap \exists
Y. \exists X. (=1 R_{i}) \sqsubseteq (=1 P)
$$
Clearly, if $\Dmf\models \text{cell}(d)$ then $\Omc,\Dmf\models (=1
P)(d)$ for the resulting ontology $\Omc$.  Conversely, the idea is
that if $\Dmf\not\models \text{cell}(d)$ and $X(d,d_{1})$,
$Y(d,d_{2}),Y(d_{1},d_{3})$, $X(d_{2},d_{4})\in \Dmf$ but
$d_{3}\not=d_{4}$, then one can extend $\Dmf$ by adding a single
$R_{1}$-successor and two $R_{2}$-successors to $d_{3}$, a single
$R_{2}$-successor and two $R_{1}$-successors to $d_{4}$, and two
$P$-successors to $d$ and thus obtain a model $\Bmf$ of $\Omc$ and $\Dmf$ in
which $d\not\in (=1P)^{\Bmf}$, see Figure~\ref{fig:no_cell_no_p}.
\begin{figure}[ht]
  \centering
  \includegraphics{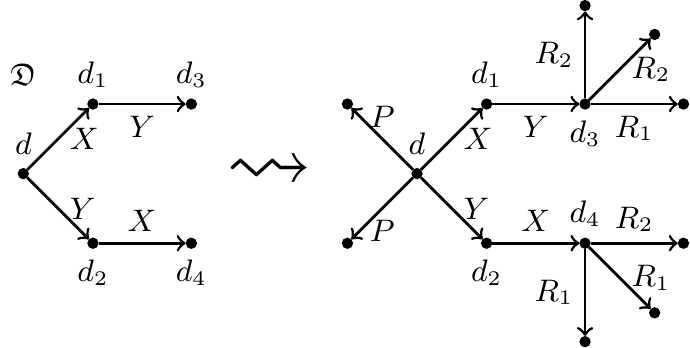}
  \caption{$\Dmf \not\models \text{cell}(d) \Rightarrow \Omc, \Dmf
    \not\models (= 1P)(d)$}
  \label{fig:no_cell_no_p}
\end{figure}
In general, however, this does not work and the latter sentence has
depth~3. The depth issue is easily resolved by introducing auxiliary
binary relation symbols $R_{i}^{X}$, $R_{i}^{Y}$, $R_{i}^{XY}$ and
$R_{i}^{YX}$, $i=1,2$, and replacing concepts such as $\exists
X.\exists Y.(=1 R_{i})$ by $(=1 R_{i}^{XY})$ and the sentences
$$
(=1 R_{i}^{XY}) \equiv \exists X.(=1 R_{i}^{Y}) \text{ and } (=1 R_{i}^{Y}) \equiv \exists Y.(=1 R_{i})
$$
Details are given below. Resolving the first issue is more involved.
There are two reasons why the converse does not hold. First, we might
have $X(d,d_{1})$, $Y(d_{1},d_{3})$, $X(d,d_{2}), Y(d_{2},d_{4})\in
\Dmf$ with $d_{3}\not=d_{4}$ but both $d_{3}$ and $d_{4}$ have already
two $R_{2}$-successors in $\Dmf$. Then the marker $(=1 P)$ is entailed
without the cell being closed at $d$ (i.e, without $\Dmf\models
\text{cell}(d)$).  Second, we might have an odd cycle of mutually
distinct $e_{0},e_{1},\ldots,e_{n}\in \Dmf$ such that each $e_{i}$
reaches $e_{(i+1)\bmod{n+1}}$ via a $Y^{-}X^{-}YX$-path in $\Dmf$, for $i=0,1,\ldots,n$.
Figure~\ref{fig:p_implies_cell_counterexample} illustrates this
for $n=2$.  Then, since in at least two neighbouring
$e_{i},e_{(i+1)\bmod{n+1}}$ the same concept $(=1 R_{i})$ is
enforced, the marker $(=1 P)$ is enforced at some node $d$ from which
$e_{i}$ and $e_{(i+1)\bmod{3}}$ are reachable along $XY$ and
$YX$-paths, respectively, without satisfying $\text{cell}(d)$.  We
resolve both problems by enforcing that $\Dmf$ is not consistent
w.r.t.~~our ontology if such constellations appear.

\begin{figure}[ht]
  \centering
  \includegraphics{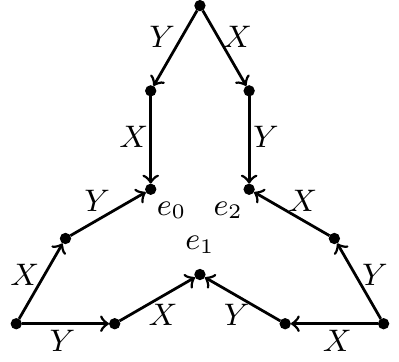}
  \caption{$\mathfrak D \models (= 1P)(d) \not\Rightarrow \mathfrak D
    \models \text{cell}(d)$}
  \label{fig:p_implies_cell_counterexample}
\end{figure}

In detail, we construct an ontology $\Omc_{\text{cell}}$ that uses in
addition to $X,Y,X^{-},Y^{-}$ the set $\text{AUX}_{\text{cell}}$ of
binary relations $P,R_{i},R_{i}^{W}$, where $i\in\{1,2\}$ and $W$
ranges over a set of words over the alphabet $\{X,Y,X^{-},Y^{-}\}$ we
define below. The $R_{i}^{W}$ serve as auxiliary symbols to avoid
sentences of depth larger than two. No unary relations are used.  To
ensure that CQ-evaluation is Datalog$^{\not=}$-rewritable
w.r.t.~$\Omc_{\text{cell}}$ we include in $\Omc_{\text{cell}}$ the
concept inclusions
$$
\top \sqsubseteq \exists Q.\top
$$
for all binary relations $Q\in \text{AUX}_{\text{cell}}$.  If an
instance $\Dmf$ is consistent w.r.t.~$\Omc_{\text{cell}}$, then
its materialization adds a certain number of $Q$-successors to any
$d\in \text{dom}(\Dmf)$ to satisfy $\top \sqsubseteq \exists Q.\top$
for $Q\in \text{AUX}_{\text{cell}}$.  The remaining sentences in
$\Omc_{\text{cell}}$ only influence the number of $Q$-successors that
have to be added and thus do not influence the certain answers to
CQs. In fact, we will have the following equivalence
$$
\Omc_{\text{cell}},\Dmf\models q(\vec{d}) \text{ iff } \{\top
\sqsubseteq \exists Q.\top \mid Q\in \text{AUX}_{\text{cell}}\},\Dmf\models
q(\vec{d})
$$
for any CQ $q$ and $\Dmf$ that is consistent w.r.t.~$\Omc_{\text{cell}}$.
Define for any non-empty word $W$
over $\{X,Y,X^{-},Y^{-}\}$ the set $\exists^{W}(=1 R_{i})$ of
sentences inductively by setting for $Z\in \{X,Y,X^{-},Y^{-}\}$:
\begin{eqnarray*}
  \exists^{Z}(=1 R_{i}) &  =  & \{(=1 R_{i}^{Z}) \equiv \exists Z.(=1 R_{i})\}\\
  \exists^{ZW}(=1 R_{i}) &  =  &\{(=1 R_{i}^{ZW}) \equiv \exists Z.(=1
  R_{i}^{W})\}\\
  &&{}\cup \exists^{W}(=1 R_{i})
\end{eqnarray*}
Thus, $\exists^{W}(=1 R_{i})$ states that the unique $d'$ reachable
from $d$ along a $W$-path has exactly one $R_{i}$-successor iff $d$
has exactly one $R_{i}^{W}$-successor.  Now $\Omc_{\text{cell}}$ is
defined as follows.
\begin{enumerate}
\item Functionality of $X,Y,X^{-}$ and $Y^{-}$ is stated using
$$
\top \sqsubseteq (\leq 1 Z)
$$
for $X,Y,X^{-},Y^{-}$.
\item All nodes have exactly one $R_{1}$ successor or exactly one
  $R_{2}$-successor:
$$
\top \sqsubseteq (=1 R_{1}) \sqcup (=1 R_{2})
$$
\item If all nodes reachable along an $XY$-path and a $YX$-path have
  exactly one $R_{1}$ and exactly one $R_{2}$-successor, then the
  marker $(=1 P)$ is set:
$$
\bigsqcap_{i=1,2}(=1 R_{i}^{XY}) \sqcap (=1 R_{i}^{YX}) \sqsubseteq
(=1 P)
$$
\item For $i=1,2$, the concept $(=1 R_{i})$ is true at least at every
  third node on the cycles in $\Dmf$ introduced above:
  $$
  (=1 R_{j}^{CC}) \sqsubseteq (=1 R_{i}) \sqcup (=1
  R_{i}^{C}) \sqcup (=1 R_{i}^{CC})
  $$
for $\{i,j\}=\{1,2\}$ and~$C=X^{-}Y^{-}XY$
\item If $(=1 R_{1})$ and $(=1 R_{2})$ are both true in a node in
  $\Dmf$ then they are both true in all `neighbouring' nodes in
  $\Dmf$:
$$
(=1 R_{1}^{X^{-}Y^{-}XY}) \sqcap (=1 R_{2}^{X^{-}Y^{-}XY}) \sqsubseteq
R^{12}
$$
$$
(=1 R_{1}^{Y^{-}X^{-}YX}) \sqcap (=1 R_{2}^{Y^{-}X^{-}YX}) \sqsubseteq R^{12}
$$
for $R^{12} \isdef{ (=1 R_{1}) \sqcap (=1 R_{2}) }$
\item The auxiliary sentences $\exists^{W}(=1 R_{i})$ for all
  relations $R_{i}^{W}$ used above.
\end{enumerate}
\begin{lemma}\label{lem:cell}
  The ontology $\Omc_{\text{cell}}$ has the following properties for
  all instances $\Dmf$:
\begin{enumerate}
\item for all $d\in \text{dom}(\Dmf)$: $\Omc_{\text{cell}},\Dmf\models
  (= 1 P)(d)$ iff $\Dmf$ is not consistent w.r.t.~$\Omc_{\text{cell}}$ or
  $\Dmf\models \text{cell}(d)$; moreover, if
  $\Dmf$ is consistent w.r.t.~$\Omc_{\text{cell}}$, then there
  exists a materialization $\Bmf$ of $\Dmf$ and $\Omc_{\text{cell}}$
  such that $d\in (=1 P)^\Bmf$ iff $d\in \text{dom}(\Bmf)$ and $\Dmf
  \models \text{cell}(d)$;
\item If all binary relations are functional in $\Dmf$, then $\Dmf$ is
  consistent w.r.t.~$\Omc_{\text{cell}}$;
\item CQ-evaluation w.r.t~$\Omc_{\text{cell}}$ is
  Datalog$^{\not=}$-rewritable.
\end{enumerate}
\end{lemma}
\begin{proof}
  We first derive a necessary and sufficient condition for consistency
  of instances $\Dmf$ w.r.t.~$\Omc_{\text{cell}}$. Lemma~\ref{lem:cell} then follows in a
  straightforward way. It is easy to see that if any of the following
  conditions is not satisfied, then $\Dmf$ is not consistent w.r.t.~$\Omc_{\text{cell}}$:
\begin{itemize}
\item all $X,Y,X^{-},Y^{-}$ are functional in $\Dmf$;
\item $\Dmf$ is consistent w.r.t.~the sentences $\exists^{W}(=1
  R_{i})$ in $\Omc_{\text{cell}}$;
\item if $\Dmf\models \text{cell}(d)$, then $d$ has at most one
  $P$-successor in $\Dmf$.
\end{itemize}
We thus assume in what follows that all three conditions are
satisfied.  Clearly, they can be encoded in Datalog$^{\not=}$.
Moreover, by Point~2 we can assume that $\Dmf$ is saturated for the
sentences $\exists^{W}(=1 R)$ in the sense that if $(=1 R^{ZW})\equiv
\exists Z.(=1 R^{W})\in \Omc_{\text{cell}}$ then for any $Z(d,d')\in
\Dmf$ the following holds: $d$ has at least two $R^{ZW}$-successors in
$\Dmf$ iff $d'$ has at least two $R^{W}$-successors in $\Dmf$.
Now let $e_{1}\leq e_{2}$ iff there are
$X(d,d_{1}),Y(d_{1},e_{1}),Y(d,d_{2}), X(d_{2},e_{2})\in \Dmf$. Let
$e_{1}\sim e_{2}$ iff $e_{1}\leq e_{2}$ or $e_{2}\leq e_{1}$ and let
$\sim^{\ast}$ be the smallest equivalence relation containing
$\sim$. For any equivalence class $E$ w.r.t.~$\sim^{\ast}$ either
\begin{itemize}
\item $E$ is of the form $e_{0}\leq \cdots \leq e_{n}$ with
  $e_{i}\not=e_{j}$ for all $i\not=j$, or
\item $E$ is a cycle $e_{0}\leq \cdots \leq e_{n}$ with $e_{i}=e_{j}$
  iff $\{i,j\}=\{0,n\}$ for all $i\not=j$.
\end{itemize}
Thus, if $E$ is not a singleton $\{e\}$ with $e\leq e$, we can
partition $E$ into two sets $E_{1}$ and $E_{2}$ (with one of them
possibly empty) such that

\medskip
\noindent
($\dagger$) there are no three $e\leq e' \leq e''$ in the same $E_{i}$.


\medskip
\noindent
Now set for any equivalence class $E$ and $\{i,j\}=\{1,2\}$,
$$
E_{j}'= \{d\in E \mid \Dmf\models (\geq 2 R_{i})(d)\}
$$

\medskip
\noindent
{\bf Claim 1}.
$\Dmf$ is consistent w.r.t.~$\Omc_{\text{cell}}$ iff the following conditions hold for all
equivalence classes $E$:
\begin{itemize}
\item[(a)] if $E=\{e\}$ with $e\leq e$ then $e\not\in E_{1}'\cup
  E_{2}'$;
\item[(b)] otherwise, there exists a partition $E_{1},E_{2}$ of $E$
  with $E_{i}\supseteq E_{i}'$ satisfying ($\dagger$).
\end{itemize}
Moreover, if (a) and (b) hold, then a materialization $\Bmf$
satisfying the conditions of Lemma~\ref{lem:cell} (1) exists.

\medskip
\noindent
($\Rightarrow$) First assume that Point~(a) does not hold for some
$E=\{e\}$ with $e\leq e$.  Then $\Dmf$ is not consistent w.r.t.~$\Omc_{\text{cell}}$
by the axioms given under (2) and (4) since it is
not possible to satisfy $(=1 R_{i})$ in $e$ if $e\in E_{j}'$
($i\not=j$).  Now assume that (b) holds. So there exists $E$ that has
either at least two elements or $e\not\leq e$ if $E=\{e\}$ but there
exists no partition $E_{1},E_{2}$ of $E$ with $E_{i}\supseteq E_{i}'$
satisfying ($\dagger$).  Then the axioms under (4) cannot be
satisfied without having at least one node in $E$ that is in both $(=1
R_{1})$ and $(=1 R_{2})$. But then by the axioms under (5) all
nodes in $E$ are in $(=1 R_{1})$ and in $(=1 R_{2})$ which implies
that $E_{1}'=E_{2}'=\emptyset$. This contradicts our assumption that
there is no partition $E_{1},E_{2}$ of $E$ with $E_{i}\supseteq
E_{i}'$ satisfying ($\dagger$).

\medskip
\noindent
($\Leftarrow$) Assume (a) and (b) hold for every equivalence class
$E$. For $E=\{e\}$ with $e\leq e$ we can thus construct the relevant
part of a model $\Bmf$ of $\Dmf$ and $\Omc_{\text{cell}}$ such that
$e$ has exactly one $R_{i}$-successor for $i=1,2$ and also exactly one
$P$-successor.  Thus axiom (2) is satisfied.  All $e$ that are not
members of such an equivalence class are given at least two
$P$-successors.  For any equivalence class with at least two members
or $E=\{e\}$ with $e\not\leq e$ we can construct the relevant part of
$\Bmf$ such that each $d\in E_{i}$ has exactly one $R_{i}$-successor
and each $d\in E\setminus E_{i}$ has at least two
$R_{i}$-successors. As $E_{1}$ and $E_{2}$ are mutually disjoint, the
axioms under (5) are satisfied.  As ($\dagger$) is satisfied, the
axioms under (4) are satisfied.  As $E_{1}\cup E_{2}$ contains $E$,
the axioms under (2) are satisfied.

\medskip
The conditions (a) and (b) can be encoded in a Datalog$^{\not=}$
program in a straightforward way and thus there is a
Datalog$^{\not=}$-program checking consistency of an instance $\Dmf$
w.r.t.~$\Omc_{\text{cell}}$. Datalog$^{\not=}$-rewritability of
CQ-evaluation w.r.t.~$\Omc_{\text{cell}}$ now follows from the
observation that
$$
\Omc_{\text{cell}},\Dmf\models q(\vec{d}) \text{ iff }
\{\top \sqsubseteq \exists Q.\top \mid Q\in
\text{AUX}_{\text{cell}}\},\Dmf\models q(\vec{d})
$$
for any CQ $q$, $\Dmf$ that is consistent w.r.t.~$\Omc_{\text{cell}}$,
and any $\vec{d}$ in $\Dmf$.
\end{proof}
This finishes the construction and analysis of $\Omc_{\text{cell}}$.
\newcommand{\xinv}{\ensuremath{x^{-}}}
\newcommand{\yinv}{\ensuremath{y^{-}}}

\bigskip
\noindent
{\bf Marking the lower left corner of grids.}
We now encode the rectangle tiling problem.
Recall that an instance of the \emph{finite rectangle tiling problem} is given by a
triple $\mathfrak{P}=(\Tmf,H,V)$ with $\Tmf$ a non-empty, finite set
of \emph{tile types} including an \emph{initial tile} $T_\mn{init}$ to
be placed on the lower left corner and nowhere else and a \emph{final
  tile} $T_\mn{final}$ to be placed on the upper right corner and
nowhere else, $H \subseteq \Tmf \times \Tmf$ a \emph{horizontal
  matching relation}, and $V \subseteq \Tmf \times \Tmf$ a
\emph{vertical matching relation}. A \emph{tiling} for $(\Tmf,H,V)$ is
a map $f:\{0,\dots,n\} \times \{0,\dots,m\} \rightarrow \Tmf$ such
that $n,m \geq 0$, $f(0,0)=T_\mn{init}$, $f(n,m)=T_\mn{final}$,
$(f(i,j),f(i+1,j)) \in H$ for all $i < n$, and $(f(i,j),f(i,j+1)) \in
V$ for all $i < m$. We say that $\mathfrak{P}$ \emph{admits a tiling}
if there exists a map $f$ that is a tiling for $\mathfrak{P}$. It is
undecidable whether an instance of the finite rectangle tiling problem
admits a tiling.

Now let $\mathfrak{P}=(\Tmf,H,V)$ with $\Tmf=\{T_1,\dots,T_p\}$.  We
regard the tile types in $\mathfrak{T}$ as unary relations and take
the binary relations symbols $X,Y,X^{-},Y^{-}$ from above and an
additional set $\text{AUX}_{\text{grid}}$ of binary relations $F,
F^{X}, F^{Y}, U, R, L, D$, and $A$.  The ontology
$\Omc_{\mathfrak{P}}$ is defined by taking $\Omc_{\text{cell}}$ and adding the sentences
$$
\top \sqsubseteq \exists Q.\top
$$
for all $Q\in \text{AUX}_{\text{grid}}$ as well as all sentences in
Figure~\ref{fig:O_P} to it, where $(T_i,T_j,T_\ell)$ range over all triples
from \Tmf such that $(T_i,T_j) \in H$ and $(T_i,T_\ell) \in V$:

\begin{figure*}[ht]
  \centering
  \begin{eqnarray*}
    T_\mn{final} & \sqsubseteq & (=1 F) \sqcap (=1 U) \sqcap (=1 R) \\
    \exists X. ((=1 U) \sqcap (=1 F) \sqcap T_j) \sqcap T_i & \sqsubseteq &  (= 1 U) \sqcap (= 1 F)  \\
    \exists Y . ((=1 R) \sqcap (=1 F) \sqcap T_{\ell}) \sqcap T_{i} & \sqsubseteq & (=1 R) \sqcap (=1 F) \\
    \exists Y.(=1 F) & \sqsubseteq & (=1 F^{Y})\\
    \exists X.(=1 F) & \sqsubseteq & (=1 F^{X})\\
    \exists X . (T_{j} \sqcap (= 1 F) \sqcap (= 1 F^{Y}))
    \sqcap &   & \\
    \exists Y . (T_{\ell} \sqcap (=1 F) \sqcap (=1 F^{X})) \sqcap (=1 P) \sqcap T_i & \sqsubseteq & (=1 F)  \\
    (=1 F) \sqcap T_\mn{init} & \sqsubseteq & (=1 A) \sqcap (=1 D) \sqcap (=1 L) \\
    \midsqcup_{1 \leq s < t \leq p} T_s \sqcap T_t  & \sqsubseteq &  \bot
  \end{eqnarray*}
  $$
  (=1 U) \sqsubseteq \forall Y. \bot\quad (=1 R) \sqsubseteq \forall
  X. \bot\quad (=1 U) \sqsubseteq \forall X. (=1 U)\quad (=1 R)
  \sqsubseteq \forall Y. (=1 R)
  $$
  $$
  (=1 D) \sqsubseteq \forall Y^{-}. \bot\quad (=1 L) \sqsubseteq \forall
  X^{-}. \bot\quad (=1 D) \sqsubseteq \forall X. (=1 D) \quad (=1 L)
  \sqsubseteq \forall Y. (=1 L)
  $$
\caption{Additional Axioms of~$\Omc_\Pmf$}
\label{fig:O_P}
\end{figure*}

We discuss the intuition behind the sentences of
$\Omc_{\mathfrak{P}}$.  The relations $X$ and $Y$ are used to
represent horizontal and vertical adjacency of points in a rectangle.
The concepts $(=1 X)$ of $\Omc_{\mathfrak{P}}$ serve the following
puroposes:
\begin{itemize}
\item $(=1 U)$, $(= 1 R)$, $(=1 L)$, and $(= 1 D)$ mark the upper,
  right, left, and lower (`down') border of the rectangle.
\item The concept $(=1 F)$ is propagated through the grid from the
  upper right corner where $T_{\sf final}$ holds to the lower left one
  where $T_{\sf initiaL}$ holds, ensuring that every position of the
  grid is labeled with at least one tile type, that the horizontal and
  vertical matching conditions are satisfied, and that the grid cells
  are closed (indicated by $(=1 P)$ from the ontology
  $\Omc_{\text{cell}}$).
\item The relations $F^{X}$ and $F^{Y}$ are used to avoid depth~3
  sentences in the same way as the relations $R^{W}_{i}$ are used to
  avoid such sentences in the construction of $\Omc_{\text{cell}}$.
\item Finally, when the lower left corner of the grid is reached, the
  concept $(=1 A)$ is set as a marker.
\end{itemize}
We write $\Dmf\models \text{grid}(d)$ if there is a tiling $f$ for
$\mathfrak{P}$ with domain $\{0,\ldots,n\}\times \{0,\ldots,m\}$ and a
mapping $\gamma: \{0,\ldots,n\}\times \{0,\ldots,m\} \rightarrow
\text{dom}(\Dmf)$ with $\gamma(0,0)=d$ such that
\begin{itemize}
\item for all $j < n$, $k \leq m$: $T(\gamma(j,k))\in \Dmf$ iff $T =
  f(j,k)$;
\item for all $b_{1},b_{2} \in \text{dom}(\Dmf)$: $X(b_{1},b_{2})\in
  \Dmf$ iff there are $j < n$, $k \leq m$ such that
  $(b_{1},b_{2})=(\gamma(j,k),\gamma(j+1,k))$;
\item for all $b_{1},b_{2} \in \text{dom}(\Dmf)$: $Y(b_{1},b_{2})\in
  \Dmf$ iff there are $j \leq n$, $k < m$ such that
  $(b_{1},b_{2})=(\gamma(j,k),\gamma(j,k+1))$.
\item $g$ is closed: if $d\in \text{ran}(\gamma)$
  and $Z(d,d')\in \Dmf$ for some $Z\in \{X,Y,X^{-},Y^{-}\}$, then
  $d'\in \text{ran}(\gamma)$.
\end{itemize}
We then call \emph{$d$ the root of the $n\times m$-grid with witness
  function $\gamma$ for $\mathfrak{P}$}.  The following result can now
be proved using Lemma~\ref{lem:cell}.
\begin{lemma}\label{lem:red3}
  The ontology $\Omc_{\mathfrak{P}}$ has the following properties for
  all instances $\Dmf$:
\begin{enumerate}
\item for all $d\in \text{dom}(\Dmf)$:
  $\Omc_{\mathfrak{P}},\Dmf\models (= 1 A)(d)$ iff $\Dmf$ is not
  consistent w.r.t.~$\Omc_{\mathfrak{P}}$ or $\Dmf\models
  \text{grid}(d)$; moreover, if $\Dmf$ is consistent w.r.t.~$\Omc_{\mathfrak{P}}$,
  then there exists a materialization $\Bmf$ of
  $\Dmf$ and $\Omc_{\mathfrak{P}}$ such that $d\in (=1 A)^\Bmf$ iff
  $d\in \text{dom}(\Bmf)$ and $\Dmf \models \text{grid}(d)$;
\item If $\Dmf\models \text{grid}(d)$ with witness $\gamma$ such that
  $\text{dom}(\Dmf)= \text{ran}(\gamma)$, and all relations are
  functional in $\Dmf$ then $\Dmf$ is consistent w.r.t.~$\Omc_{\mathfrak{P}}$;
\item CQ-evaluation w.r.t~$\Omc_{\mathfrak{P}}$ is
  Datalog$^{\not=}$-rewritable.
\end{enumerate}
\end{lemma}
We now use Lemma~\ref{lem:red3} to prove the undecidability result.
Let $\Omc = \Omc_{\mathfrak{P}} \cup \{(=1 A) \sqsubseteq B_{1} \sqcup
B_{2}\}$, where $B_{1}$ and $B_{2}$ are unary relations.
\begin{lemma}\label{lem:undecidability}
(1) If $\mathfrak{P}$ admits a tiling, then $\Omc$ is not materializable
and CQ-evaluation w.r.t.~$\Omc$ is {\sc coNP}-hard.

(2) If $\mathfrak{P}$ does not admit a tiling, then CQ-evaluation
  w.r.t.~$\Omc$ is Datalog$^{\not=}$-rewritable.
\end{lemma}
\begin{proof}
(1) Consider a tiling $f$ for $\mathfrak{P}$ with domain
$\{0,\ldots,n\}\times\{0,\ldots,m\}$.  Regard the pairs in
$\{0,\ldots,n\}\times\{0,\ldots,m\}$ as constants. Let $\Dmf$
contain
$$
X((i,j),(i+1,j)),
$$
for $i<n$ and $j\leq m$,
$$
Y((i,j),Y(i,j+1))
$$
for $i\leq n$ and $j<m$, and
$$
T(i,j)
$$
for $f(i,j)=T$ for $i\leq n$ and $j\leq m$.  Then $\Dmf$ is consistent w.r.t.~$\Omc_{\mathfrak{P}}$ and
$\Omc,\Dmf\models (=1 A)(0,0)$, by Lemma~\ref{lem:red3}. Thus $\Omc_{\mathfrak{P}},\Dmf\models
B_{1}(0,0) \vee B_{2}(0,0)$ but $\Omc,\Dmf\not\models B_{1}(0,0)$ and
$\Omc,\Dmf\not\models B_{2}(0,0)$. Thus $\Omc$ is not materializable
and so CQ-evaluation is {\sc coNP}-hard.

\medskip
(2) Assume $\mathfrak{P}$ does not admit a tiling. Any
instance $\Dmf$ is consistent w.r.t.~$\Omc$ iff it is consistent
w.r.t.~$\Omc_{\mathfrak{P}}$ and, by Lemma~\ref{lem:red3}, if
$\Omc_{\mathfrak{P}},\Dmf\models (=1A)(d)$ for some $d\in
\text{dom}(\Dmf)$, then $\Dmf$ is not consistent w.r.t.~$\Omc_{\mathfrak{P}}$.
Thus, $\Omc,\Dmf\models q(\vec{d})$ iff
$\Omc_{\mathfrak{P}},\Dmf\models q(\vec{d})$ holds for every CQ $q$
and all $\vec{d}$.  Thus CQ-evaluation w.r.t.~$\Omc$ is
Datalog$^{\not=}$-rewritable by Lemma~\ref{lem:red3}.
\end{proof}
Lemma~\ref{lem:undecidability} implies Theorem~\ref{thm:undecidability}
for $\mathcal{ALCIF}_{\ell}$ ontologies of depth~2. For uGF$^{-}_{2}(2,f)$ we
modify the construction of $\Omc_{\text{cell}}$ and $\Omc_{\mathfrak{P}}$ as follows:
\begin{itemize}
\item The relations $X,Y,X^{-},Y^{-}$ are defined as functions
and it is stated that $X^{-}$ and $Y^{-}$ are the inverse of $X$ and $Y$, respectively.
\item For any relation symbol $R$ in $\Omc_{\mathfrak{P}}$ distinct from $X,Y,X^{-},Y^{-}$
we introduce a function $F$, state $\forall x F(x,x)$, and replace
$$
\top \sqsubseteq \exists R.\top
$$
by
$$
\forall x \exists y (R(x,y) \wedge F(x,y)).
$$
\item We replace all occurrences of $(=1 R)$ for $R\not\in \{X,Y,X^{-},Y^{-}\}$ in $\Omc_{\mathfrak{P}}$ by
$$
\neg \exists y (R(x,y) \wedge \neg F(x,y))
$$
\end{itemize}
Now Lemma~\ref{lem:cell} and Lemma~\ref{lem:red3} still hold for the resulting ontologies $\Omc_{\text{cell}}$ and
$\Omc_{\mathfrak{P}}$ if $(=1P)$ and $(=1A)$ are replaced by $\neg \exists y (P(x,y) \wedge \neg F(x,y))$ and
$\neg \exists y (A(x,y) \wedge \neg F(x,y))$, respectively.

\bigskip

We now come to the proof of Theorem~\ref{thm:nondichotomy}.
We first give a more detailed definition of the run fitting problem. 

We consider \emph{non-deterministic Turing machines}
(\emph{TMs}, for short).
A TM $M$ is represented by a tuple $(Q,\Sigma,\Delta,q_0,q_a)$,
where $Q$ is a finite set of states,
$\Sigma$ is a finite alphabet, $\Delta \subseteq Q \times \Sigma
\times Q \times \Sigma \times \{L,R\}$ is the transition relation, and
$q_0,q_a \in Q$ are the start state and accepting state, respectively.
The tape of $M$ is assumed to be one-sided infinite, that is, it has a
left-most cell and extends infinitely to the right.  The set of
strings accepted by $M$ is denoted by $L(M)$.
A \emph{configuration} of $M$ is represented by a string $vqw$,
where $q$ is the state,
$v$ is the inscription of the tape to the left of the head,
and $w$ is the inscription of the tape to the right of the head
in the configuration
(as usual, we omit all but possibly a finite number of trailing blanks).
The configuration is \emph{accepting} if $q = q_a$.
A \emph{run} of $M$ is represented by a finite sequence
$\gamma_0,\ldots,\gamma_{n}$ of configurations of $M$ with
$|\gamma_{0}|=\cdots = |\gamma_{n}|$.
We assume that the accepting state has no successor states. A run is \emph{accepting} if its
last configuration is accepting.

\begin{definition}\upshape
  Let $M = (Q,\Sigma,\Gamma,\Delta,q_0,q_{a})$ be a TM.
  \begin{itemize}
  \item A \emph{partial configuration} of $M$ is a string
    $\tilde{\gamma}$ over $Q \cup \Sigma \cup \{\star\}$ such that
    there is at most one $i \in \{1,\dotsc,n\}$ with
    $\tilde{\gamma}[i] \in Q$.  Here, $\gamma[i]$ denotes the symbol
    that occurs at the $i$-th position of $\gamma$.  A configuration
    $\gamma$ \emph{matches} $\tilde{\gamma}$ if $\lvert\gamma\rvert =
    \lvert\tilde{\gamma}\rvert$ and for each $i \in \{1,\dotsc,n\}$
    with $\tilde{\gamma}[i] \neq \star$ we have $\gamma[i] =
    \tilde{\gamma}[i]$.
  \item A \emph{partial run} of $M$ is a sequence \(
    \tilde{\gamma} =
    (\tilde{\gamma}_0,\tilde{\gamma}_1,\dotsc,\tilde{\gamma}_m) \) of
    partial configurations $\tilde{\gamma}_i$ of $M$ such that
    $\tilde{\gamma}_{0}$ starts with $q_{0}$, has no other occurrence
    of a state $q\in Q$, and $|\tilde{\gamma}_0|=\cdots
    =|\tilde{\gamma}_m|$.  A run $\gamma_0,\gamma_1,\dotsc,\gamma_n$
    of $M$ matches $\tilde{\gamma}$ if $m = n$ and $\gamma_i$ matches
    $\tilde{\gamma}_i$, for each $i \in \{0,1,\dotsc,m\}$.
  \end{itemize}
\end{definition}

\begin{definition}
  The \emph{run fitting problem} 
  for a TM $M$, denoted by
  $\RunFit(M)$, is defined as follows: Given a partial run $\tilde{\gamma}$
  of $M$, decide whether there is an accepting run of $M$ that matches
  $\tilde{\gamma}$.
\end{definition}

It is easy to see that $\RunFit(M)$ is in $\NPTime$
for every 
TM $M$.
The following theorem states that the complexity of $\RunFit(M)$
may be intermediate between $\PTime$ and $\NPTime$-complete,
provided $\PTime \neq \NPTime$.

\begin{trivlist}\item
  \textbf{Theorem~\ref{intermediate} (restated)}~\itshape
  If $\PTIME \neq \NPTime$,
  then there is a TM whose run fitting problem is 
  neither in $\PTIME$ nor $\NPTime$-complete.
\end{trivlist}

To prove Theorem~\ref{intermediate}, we now construct such a TM.
The construction is a modification of the construction
from Impagliazzo's version of the proof of Ladner's Theorem~\cite{DBLP:journals/jacm/Ladner75},
as presented in \cite[Theorem 3.3]{DBLP:books/daglib/0023084}.

Fix a polynomial-time TM $M_\SAT$ for $\SAT$.  For a monotone
polynomial-time computable function $H\colon \mathbb{N} \to
\mathbb{N}$ to be specified later, let $M_H$ be a polynomial-time TM
that works as follows on a given input $v$:
\begin{enumerate}
\item Check that there is an integer $n \geq 0$ such that $v$ is the
  unary representation of $n^{H(n)}$ (i.e., $v = 1^{n^{H(n)}}$).  If
  such an $n$ does not exist, then reject $v$.
\item Guess an input $w$ of length $n$ for $M_\SAT$.
\item Generate the initial configuration $\gamma$ of $M_\SAT$ on input
  $w$.
\item Run $M_\SAT$ from $\gamma$, and accept $v$ iff
  $M_\SAT$ accepts $w$.
\end{enumerate}
We refer to the first three steps as the \emph{initialization phase}.

We now define the function $H\colon \mathbb{N} \to \mathbb{N}$.
Fix a polynomial time computable enumeration $M_0,M_1,M_2,\dotsc$
of deterministic TMs such that all runs of $M_i$ on inputs of length $n$
terminate after at most $i \cdot n^i$ steps,
and for each problem $A$ in $\PTime$ there are infinitely many $i$ such
that $L(M_i)$ is in $A$.%
\footnote{It is easy to construct a deterministic polynomial-time
  Turing machine that, given an integer $i \geq 0$, outputs a
  deterministic Turing machine $M_i$ such that the sequence $(M_i)_{i
    \geq 0}$ has the desired properties.  For instance, let $M'_i$ be the
  $i$-th deterministic Turing machine in lexicographic order under
  some string encoding of Turing machines, and add a clock to $M'_i$
  that stops the computation of $M'_i$ after at most $i \cdot n^i$
  steps (and rejects if $M'_i$ did not accept yet).}
Then, $H(n)$ is defined as
\begin{align*}
  \min\, \{i < \log \log n \mid\
  & \text{for all strings $z$ of length $\leq \log n$,} \\
  & \text{$M_i$ accepts $z$ iff $z \in \RunFit(M_H)$}
  \},
\end{align*}
or as $\log \log n$ if the minimum does not exist.
It is not hard to see that $H$ is well-defined,
and that there is a deterministic polynomial-time Turing machine that,
given a positive integer $n$ in unary, outputs $H(n)$.
For details, we refer to \cite{DBLP:books/daglib/0023084}.

This finishes the construction of $M_H$.
Lemma~\ref{intermediate-main} below
shows that $\RunFit(M_H)$ has the desired properties,
namely that $\RunFit(M_H)$ is neither in $\PTime$ nor $\NPTime$-complete,
unless $\PTime = \NPTime$.
It uses the following auxiliary lemma.

\begin{lemma}\mbox{}
  \label{H}
  \begin{itemize}
  \item
    If $\RunFit(M_H)$ is in $\PTIME$, then $H(n) = O(1)$.
  \item
    If $\RunFit(M_H)$ is not in $\PTIME$, then $\lim_{n \to \infty} H(n) = \infty$.
  \end{itemize}
\end{lemma}

\begin{proof}
  The proof is as in \cite{DBLP:books/daglib/0023084}.
  We provide a proof for the sake of completeness.

  Assume first that $\RunFit(M_H)$ is in $\PTIME$.
  Then, there is an index $i$ such that $L(M_i) = \RunFit(M_H)$.
  Now, for all $n > 2^{2^i}$,
  we have $i < \log \log n$ and thus $H(n) \leq i$ by the definition of $H$.
  It follows that $H(n) \leq \max \set{H(m) \mid m \leq 2^{2^i}+1}$,
  and therefore $H(n) = O(1)$.

  Next assume that $\RunFit(M_H)$ is not in $\PTIME$.
  For a contradiction,
  suppose that $\lim_{n \to \infty} H(n)  \neq \infty$.
  Since $H$ is monotone,
  this means that there are integers $n_0,i \geq 0$
  such that $H(n) = i$ for all integers $n \geq n_0$.
  Let $n \geq n_0$.
  By the definition of $H$,
  we have that $M_i$ agrees with $\RunFit(M_H)$
  on all strings of length at most $\log n$.
  Since this holds for all $n \geq n_0$,
  we conclude that $M_i$ decides $\RunFit(M_H)$.
  But then, $\RunFit(M_H)$ is in $\PTIME$,
  which contradicts our initial assumption that $\RunFit(M_H)$ is not in $\PTIME$.
\end{proof}

We now prove the main lemma,
which concludes the proof of Theorem~\ref{intermediate}.

\begin{lemma}
  \label{intermediate-main}
  If $\PTIME \neq \NPTime$,
  then $\RunFit(M_H)$ is neither in $\PTIME$ nor $\NPTime$-complete.
\end{lemma}
\begin{proof}
  \emph{``$\RunFit(M_H)$ is not in $\PTIME$'':}
  For a contradiction, suppose that $\RunFit(M_H)$ is in $\PTIME$.
  By Lemma~\ref{H},
  there is a constant $c \geq 0$ such that for all integers $n \geq 0$
  we have $H(n) \leq c$.
  Suppose that on inputs of length $n$,
  $M_\SAT$ makes at most $p(n) \isdef n^k+k$ steps.
  Then, the following is a polynomial-time many-one reduction
  from $\SAT$ to $\RunFit(M_H)$,
  implying $\PTIME = \NPTime$,
  and therefore contradicting the lemma's assumption.

  Given an input $x$ of length $n$ for $M_\SAT$:
  \begin{enumerate}
  \item Compute $h \isdef H(n)$ and $w \isdef 1^{n^h}$.
  \item Output the partial run
    $\tilde\gamma_0,\dotsc,\tilde\gamma_{i+p(n)}$ of $M_H$ such that:
    \begin{itemize}
    \item $\tilde\gamma_0,\dotsc,\tilde\gamma_i$ corresponds to the
      initialization phase of $M_H$ on input $w$ that generates the
      start configuration of $M_\SAT$ on input $x$.  In particular,
      $\tilde\gamma_0,\dotsc,\tilde\gamma_i$ are complete configurations
      of $M_H$, and $\tilde\gamma_0 = q_0 w$ and $\tilde\gamma_i = q'_0 x$, where
      $q_0$ and $q'_0$ are the start states of $M_H$ and $M_\SAT$,
      respectively;
    \item $\tilde\gamma_{i+1},\dotsc,\tilde\gamma_{i+p(n)}$ consist
      entirely of stars.
    \end{itemize}
  \end{enumerate}
  Note that the partial run
  $\tilde\gamma_0,\dotsc,\tilde\gamma_{i+p(n)}$ can be easily computed
  by simulating the initialization phase $M_H$ on input $w$, where in
  step 2 of the initialization phase we ``guess'' the input string $x$
  given as input to the reduction.  Then, we pad the sequence of
  configurations corresponding to the initialization phase by $p(n)$
  partial configurations, each consisting of exactly $p(n)$ stars.

  \bigskip\noindent
  \emph{``$\RunFit(M_H)$ is not $\NPTime$-complete'':}
  Suppose, to the contrary, that $\RunFit(M_H)$ is $\NPTime$-complete.
  Then there is a polynomial-time many-one reduction $f$ from $\SAT$
  to $\RunFit(M_H)$.
  Using $f$, we construct a polynomial-time many-one reduction $g$
  from $\SAT$ to $\SAT$
  such that for all sufficiently large strings $x$
  we have $\len{g(x)} < \len{x}$.
  This implies that $\SAT$ can be solved in polynomial time,
  and contradicts $\PTIME \neq \NPTime$.


  Consider an input $x$ for $\SAT$.
  Since $f$ is a many-one reduction from $\SAT$ to $\RunFit(M_H)$,
  we have
  \begin{align}
    \label{eq:padding-with-template}
    f(x) \,=\, \tilde\gamma_0\#\tilde\gamma_1\#\dotsb\#\tilde\gamma_m
  \end{align}
  for some partial run
  \begin{align*}
    \tilde\gamma
    \,\isdef\,
    (\tilde\gamma_0,\tilde\gamma_1,\dotsc,\tilde\gamma_m)
  \end{align*}
  of $M_H$.
  Moreover, $x \in \SAT$ iff there is an accepting run
  of $M_H$ that matches $\tilde\gamma$.
  By the construction of $M_H$,
  an accepting run of $M_H$ on an input $y$ can only exist
  if there is an integer $n \geq 0$ such that $y = 1^{n^{H(n)}}$.
  Note also that the length of $y$ has to be bounded by $\len{\tilde\gamma_0}$.
  Define
  \[
    N \,\isdef\, \set{n \in \nat \mid n^{H(n)} \leq \len{\tilde\gamma_0}}.
  \]
  Then, as argued above, the following are equivalent:
  \begin{enumerate}
  \item $x \in \SAT$;
  \item $\tilde\gamma_0\#\tilde\gamma_1\#\dotsb\#\tilde\gamma_{m} \in
    \RunFit(M_H)$;
  \item there is an $n \in N$ such that there is an accepting run of
    $M_H$ on input $1^{n^{H(n)}}$ that matches $\tilde\gamma$.
  \end{enumerate}

  In what follows, we show how to compute in polynomial time, for each
  $n \in N$, a propositional formula $\phi_n$ such that:
  \begin{itemize}
  \item $\phi_n$ is satisfiable if and only if there is an accepting
    run of $M_H$ on input $1^{n^{H(n)}}$ that matches $\tilde\gamma$;
  \item $\len{\phi_n} \leq \frac{\len{x}}{\card{N}}-2$ for all $n \in
    N$ (if $x$ is large enough).
  \end{itemize}
  Then, the following function $g$ is a polynomial-time many-one
  reduction from $\SAT$ to $\SAT$:
  \begin{align*}
    g(x) \,\isdef\, \bigvee_{n \in N} \phi_n.
  \end{align*}
  Assuming a suitable encoding of propositional formulas, the size of
  $g(x)$ is then bounded by $\len{x}-1$ for large enough $x$.  Thus,
  $g$ is the desired length-reducing polynomial-time self-reduction of
  $\SAT$.  It remains to construct $\phi_n$, for all $n \in N$.

  \medskip\textsc{Construction of $\phi_n$.}
  Fix $n \in N$.
  By the construction of $M_H$,
  any accepting run of $M_H$ on input $1^{n^{H(n)}}$
  has to start with the initialization phase.
  The first step of the initialization phase is deterministic,
  and checks whether the input has the form $1^{n^{H(n)}}$.
  Thus, we can complete $\tilde\gamma$ in polynomial time
  to a partial run of $M_H$
  where the first step of the initialization phase is completely specified.
  If this is not possible due to constraints imposed by $\tilde\gamma$,
  then we know that the desired accepting run does not exist,
  and we can output a trivial unsatisfiable formula $\phi_n$.
  Otherwise, let
  \begin{align*}
    \tilde{\tilde\gamma}
    \,=\,
    (\tilde{\tilde\gamma}_0,\tilde{\tilde\gamma}_1,\dotsc,
     \tilde{\tilde\gamma}_m)
  \end{align*}
  be the resulting partial run of $M_H$.
  It remains to construct a formula $\phi_n$ that is satisfiable
  iff there is an accepting run of $M_H$
  that matches $\tilde{\tilde\gamma}$.

  Let us take a closer look at $\tilde{\tilde\gamma}$.
  Let $i \geq 0$ be such that
  $\tilde{\tilde\gamma}_0,\dotsc,\tilde{\tilde\gamma}_{i}$
  corresponds to the first step of the initialization phase of $M_H$
  on input $1^{n^{H(n)}}$.
  In particular, for each $j \in \set{0,1,\dotsc,i}$,
  $\tilde{\tilde\gamma}_j$ is a completely specified configuration.
  It is possible to specify $M_H$ in such a way
  that the second and third step of the initialization phase of $M_H$
  on input $1^{n^{H(n)}}$
  take exactly $n$ computation steps combined,
  and that any configuration after the initialization phase uses space
  at most $n$.
  Thus, without loss of generality we can assume:
  \begin{enumerate}
  \item
    $\len{\tilde{\tilde\gamma}_j} \leq n$ for all $j \in \set{i+1,\dotsc,m}$;
  \item
    $m-i-n$ is bounded by the running time of $M_\SAT$ on inputs of length $n$.
  \end{enumerate}

  Let $h$ be a polynomial-time computable function that,
  given $\tilde{\tilde\gamma}_{i+1} \# \dotsb \# \tilde{\tilde\gamma}_{m}$,
  outputs a propositional formula that is satisfiable iff
  there is an accepting run of $M_H$
  that starts in the second step of the initialization phase of $M_H$
  in a configuration matching $\tilde{\tilde\gamma}_{i+1}$,
  and that matches $\tilde{\tilde\gamma}_{i+1},\dotsc,\tilde{\tilde\gamma}_{m}$.
  Let
  \[
    \phi_n \,\isdef\,
    h(\tilde{\tilde\gamma}_{i+1} \# \dotsb \# \tilde{\tilde\gamma}_{m}).
  \]
  This finishes the construction of $\phi_n$.

  It is immediate from the construction of $\phi_n$
  that $\phi_n$ is satisfiable if and only if there is an accepting run
  of $M_H$ on input $1^{n^{H(n)}}$ that matches $\tilde\gamma$.
  It remains to prove that the length of $\phi_n$
  is bounded by $\len{x}/\card{N}-2$.

  \medskip\textsc{Bounding the size of $\phi_n$.}
  Let $p$ be a polynomial such that for all strings $z$,
  $M_{\SAT}$ makes at most $p(\len{z})$ steps on input $z$,
  and both $\len{f(z)}$ and $\len{h(z)}$ are bounded by $p(\len{z})$.
  Since, as mentioned above,
  $m-i-n$ is bounded by the running time of $M_\SAT$ on inputs of length $n$,
  we have
  \begin{align*}
    m - i \,\leq\, p(n)+n.
  \end{align*}
  Since moreover $\len{\tilde{\tilde\gamma}_j} \leq n$
  for all $j \in \set{i+1,\dotsc,m}$,
  we have
  \begin{align*}
    \len{\phi_n}
    & \,\leq\,
    \len{h(\tilde{\tilde\gamma}_{i+1} \# \dotsb \# \tilde{\tilde\gamma}_{m})} \\
    & \,\leq\,
    p(\len{\tilde{\tilde\gamma}_{i+1} \# \dotsb \# \tilde{\tilde\gamma}_{m}}) \\
    & \,\leq\,
    p(\,(m-i) \cdot (n+1)\,) \\
    & \,\leq\,
    p(\,(p(n)+n) \cdot (n+1)\,).
  \end{align*}
  Hence,
  \begin{align*}
    \len{\phi_n} \,\leq\, q(n)
  \end{align*}
  for some polynomial $q$ depending only on $M_{\SAT}$, $f$, and $h$.
  It remains to show that for all $n \in N$
  we have $q(n) \leq \len{x}/\card{N}-2$ if $x$ is sufficiently large.

  \begin{trivlist}
  \item\textit{Claim.~}\itshape
    There is a polynomial $r(\ell) > 0$ depending only on $M_H$
    such that for sufficiently large $x$
    we have $\frac{\len{x}}{\card{N}} \geq r(\card{x})$.
  \end{trivlist}

  \begin{trivlist}
  \item\textit{Proof.}
    Recall that $N$ consists of all integers $n \geq 0$
    such that $n^{H(n)} \leq \len{\tilde\gamma_0}$.
    Since $\tilde\gamma_0$ is part of $f(x)$,
    whose overall length is bounded by $p(\len{x})$,
    we have $n^{H(n)} \leq p(\len{x})$.

    Now, for all integers $\ell \geq 0$,
    define
    \[
      N(\ell) \,\isdef\, \set{n \in \nat \mid n^{H(n)} \leq p(\ell)}.
    \]
    Then, $N \subseteq N(\len{x})$.
    We show that for all constants $c \in (0,1)$
    there is an integer $\lambda_c \geq 0$
    such that for all integers $\ell \geq \lambda_c$
    we have $\card{N(\ell)} \leq \ell^c$.
    This implies the claim.%
    \footnote{Set $r(\ell) \isdef \ell^{1-c}$ for some $c \in (0,1)$
      (e.g., $r(\ell) = \sqrt{\ell}$).
      Then,
      \(
        \frac{\len{x}}{\card{N}}
        \geq \frac{\len{x}}{\card{N(\len{x})}}
        \geq r(\len{x})
      \)
      if $\len{x} \geq \lambda_c$.}

    Fix a constant $c \in (0,1)$
    and $L \isdef \set{\ell \in \nat \mid \card{N(\ell)} > \ell^c}$.
    For each $\ell \in L$,
    there is an $n_\ell \in N(\ell)$ with $n_\ell \geq \ell^c$.
    Thus, by the definition of $N(\ell)$ and the monotonicity of $H$,
    for each $\ell \in L$ we have
    \begin{align}
      \label{eq:bound}
      \ell^{c \cdot H(\ell^c)} \,\leq\, n_\ell^{H(n_\ell)} \,\leq\, p(\ell).
    \end{align}
    Now, since $\lim_{\ell \to \infty} H(\ell) = \infty$ (by Lemma~\ref{H}),
    we have $\lim_{\ell \to \infty} c H(\ell^c) = \infty$.
    Hence, there is an integer $\lambda_c \geq 0$
    such that for all $\ell \geq \lambda_c$
    we have $\ell^{c \cdot H(\ell^c)} > p(\ell)$.
    This implies that for all $\ell \geq \lambda_c$
    we have $\card{N(\ell)} \leq \ell^c$
    (otherwise, we would violate~\eqref{eq:bound}).
    \hfill$\lrcorner$
  \end{trivlist}

  Assume that $q(n) = n^k + k$.
  Let $r$ be a polynomial as guaranteed by the claim.
  In what follows,
  we will assume that $x$ is large enough so that:
  \begin{enumerate}
  \item
    $\frac{\len{x}}{\card{N}} \geq r(\card{x})$;
    this can be satisfied by the previous claim.
  \item
    $\left(r(\len{x})-2-k\right)^{H\bigl(\left(r(\len{x})-2-k\right)^{\frac{1}{k}}\bigr)/k}
    > p(\len{x})$;
    this is possible since $\lim_{\ell \to \infty} H(\ell) = \infty$
    by Lemma~\ref{H}.
  \end{enumerate}
  Suppose that there is an $n \in N$ such that $q(n) > \len{x}/\card{N}-2$.
  Then,
  \begin{align*}
    n > \left(\frac{\len{x}}{\card{N}}-2-k\right)^{\frac{1}{k}}.
  \end{align*}
  This implies
  \begin{align*}
    n^{H(n)}
    & \,\geq\,
    \left(\frac{\len{x}}{\card{N}}-2-k\right)^{\frac{H(n)}{k}} \\
    & \,\geq\,
    \left(\frac{\len{x}}{\card{N}}-2-k\right)^{\frac{H\left(\left(\frac{\len{x}}{\card{N}}-2-k\right)^{\frac{1}{k}}\right)}{k}} \\
    & \,\geq\,
    \left(r(\len{x})-2-k\right)^{\frac{H\left(\left(r(\len{x})-2-k\right)^{\frac{1}{k}}\right)}{k}}\\
    & \,>\,
      p(\card{x}),
  \end{align*}
  where the last two inequalities follow from the monotonicity of $H$.
  Consequently,
  \begin{align*}
    \len{\tilde{\tilde\gamma}_{i+1}\#\dotsb\#\tilde{\tilde\gamma}_m}
    & \,\leq\,
      \len{\tilde{\tilde\gamma}_0\#\dotsb\#\tilde{\tilde\gamma}_m} -
      \len{\tilde{\tilde\gamma}_0\#\dotsb\#\tilde{\tilde\gamma}_i} \\
    & \,\leq\,
      p(\len{x}) - n^{H(n)} \\
    & \,<\,
      p(\len{x}) - p(\len{x}),
  \end{align*}
  which is the desired contradiction.
\end{proof}

\begin{trivlist}\item
  \textbf{Lemma~\ref{lem:nondicho} (restated)}~\itshape
For every Turing machine $M$, there is a
uGF$^{-}_{2}(2,f)$ ontology $\Omc$ and an $\mathcal{ALCIF}_{\ell}$
ontology \Omc of depth 2 such that the following hold, where $N$
is a distinguished unary relation:
\begin{enumerate}
\item there is a polynomial reduction of the run fitting 
problem for $M$ to the complement of evaluating the OMQ $(\Omc,q\leftarrow N(x))$;
\item for every UCQ $q$, evaluating the OMQ $(\Omc,q)$
is polynomially reducible to the complement of the run fitting 
problem for $M$.
\end{enumerate}
\end{trivlist}
\begin{proof}
We give the proof for $\mathcal{ALCIF}_{\ell}$ ontologies of depth~2. The proof for
uGF$^{-}_{2}(2,f)$ is obtained by modifying the $\mathcal{ALCIF}_{\ell}$ ontology in the same way
as in the proof of Theorem~\ref{thm:undecidability} by replacing, for example, $(\geq 2 R)$ by $\exists y (R(x,y) \wedge \neg F(x,y))$.

Assume $M = (Q,\Gamma,\Delta,q_{0},q_{a})$ is given.
The instances $\Dmf$ we use to represent partial runs and that
provide the space for simulating matching runs are $n \times m$
$X,Y$-grids with $T_{\sf init}$ written in the lower left corner, $T_{\sf final}$ written in the upper
right corner, and $B$ (for blank) written everywhere else. To re-use the notation and results from
the proof of Theorem~\ref{thm:undecidability} we regard such a structure as a tiling with tile types
$\mathfrak{T}= \{B, T_{\sf final},T_{\sf init}\}$. Then the ontology
$\Omc_{\mathfrak{G}}$ for $\mathfrak{G}= (\mathfrak{T},H,V)$ and
\begin{eqnarray*}
  H  & =  &\{(B,B),(B,T_{\sf final}),(T_{\sf initial},B)\}\\
  V & =  &\{(B,B),(B,T_{\sf final}),(T_{\sf initial},B)\}
\end{eqnarray*}
checks whether an instance represents a grid structure.  We now
construct the set $\Omc_{M}$ of sentences that encode runs of $M$ that
match a partial run.
%
For any $\Dmf$, the simulation of a run is triggered at a constant $d$
exactly if $\Omc_{\mathfrak{G}},\Dmf\models (=1 A)(d)$.
%
%
%
$\Omc_{M}$ uses in addition to the binary relations in
$\Omc_{\mathfrak{G}}$ binary relations $q\in Q$ that occur in concepts
$(\geq 2 q)$ that indicate that $M$ is in state $q$. It also uses
binary relations $G$ for $G\in \Gamma$ that occur in concepts $(\geq 2
G)$ that indicate that $G$ is writte on the corresponding cell of the
tage. The sentences of $\Omc_{M}$ are now as follows:
\begin{itemize}
\item The grid in which the lower left corner is marked with $(=1 A)$
  is colored with $(=1 A)$, $q_{0}$ is the first symbol of the first
  configuration and no state occurs later in the first configuration:
$$
(=1 A) \sqsubseteq \forall X.(=1 A) \sqcap \forall Y.(=1 A)\ ,
$$
$$
(=1 A) \sqcap T_{\sf init} \sqsubseteq (\geq 2 q_{0})\ ,
$$
$$
(=1 A) \sqcap (=1 D) \sqcap (\geq 2 q) \sqsubseteq (= 1 L)
$$
The remaining sentences are all relativized to $(=1 A)$ and so apply
to constants in a grid only.
\item every grid point is colored with some $(\geq 2 G)$ for $G\in
  \Gamma$ or $(\geq 2 q)$ for $q\in Q$:
$$
(=1 A) \sqsubseteq \bigsqcup_{G\in \Gamma}(\geq 2 G) \sqcup
\bigsqcup_{q\in Q} (\geq 2 q)
$$
\item The machine $M$ cannot be in two different states at the same
  time and no two distinct symbols can be written on the same cell.
  For all mutually distinct $H_{1},H_{2}\in Q \cup \Gamma$:
$$
(=1 A) \sqcap (\geq 2 H_{1}) \sqcap (\geq 2 H_{2}) \sqsubseteq \bot,
$$
\item For any triple $G_{0}qG_{1}\in \Gamma\times Q \times \Gamma$ let
  $S(G_{0}qG_{1})$ denote the set of all possible successor triples
  $S_{1}S_{2}S_{3}\in (Q\times \Gamma \times \Gamma) \cup
  (\Gamma\times \Gamma\times Q)$ according to the transition relation
  $\Delta$ of $M$.  To avoid sentences of depth larger than two we
  introduce for the words $W\in \{X,XX\}$ and for $S\in Q\cup \Gamma$
  fresh binary relations $S^{W}$ and the sentences
$$
(=1 A) \sqcap (\geq 2 S^{X}) \equiv (=1A) \sqcap \exists X.(\geq 2 S)\
,
$$
$$
(=1A) \sqcap (\geq 2 S^{XX}) \equiv (=1 A)\sqcap \exists X.(\geq
2 S^{X})
$$
and then add
$$
(=1 A) \sqcap (\geq 2 G_{0}) \sqcap (\geq 2 q^{X}) \sqcap (\geq 2
G_{1}^{XX}) \sqsubseteq
$$
$$ \quad \bigsqcup_{S_{1}S_{2}S_{3}\in S(G_{0}qG_{1})}
\!\!\!\exists Y.(\geq 2 S_{1}) \sqcap (\geq 2 S_{2}^{X}) \sqcap (\geq 2
S_{3}^{XX})
$$
to $\Omc_{M}$.
\item the symbol written on a cell does not change if the head is more
  than one cell away.  For all $G,G_{1},G_{2}\in \Gamma$:
$$
(=1 A) \sqcap \forall X.G_{1}^{\geq 2} \sqcap \forall
X^{-}.G_{2}^{\geq 2} \sqcap G^{\geq 2} \sqsubseteq\forall Y.G^{\geq 2}
$$
for~$G_i^{\geq 2}\isdef (\geq 2 G_i)$ with~$i\in\{1,2\}$ or empty.
\item the final state cannot be a state distinct from the accepting
  state $q_{a}$.  For all $q\in Q\setminus \{q_{a}\}$:
$$
(=1A) \sqcap (\geq 2 q) \sqsubseteq \exists Y.\top
$$
\item Finally, let $\text{AUX}_{M}$ denote the set of fresh binary
  relations used above and add
$$
\top \sqsubseteq \exists Q.\top
$$
to $\Omc_{M}$ for all $Q\in \text{AUX}_{M}$.
\end{itemize}
This finishes the definition of $\Omc_{M}$. Let $\Omc =
\Omc_{\mathfrak{G}}\cup \Omc_{M}$. We show that $\Omc$ is as required.

\medskip

Let $N$ be a fresh unary relation symbol. Then an instance $\Dmf$ is consistent w.r.t.~$\Omc$
if $\Omc,\Dmf\not\models q$ for the Boolean query $q \leftarrow N(x)$. It therefore suffices to provide
a polynomial reduction of the run fitting 
problem for $M$ to the problem whether
an instance $\Dmf$ is consistent w.r.t.~$\Omc$.

Assume that a partial run \(
      \tilde{\gamma} =
      (\tilde{\gamma}_0,\tilde{\gamma}_1,\dotsc,\tilde{\gamma}_m)
    \)
    of partial configurations $\tilde{\gamma}_i$ of $M$ such that
    $\tilde{\gamma}_{0}$ starts with $q_{0}$ and $|\tilde{\gamma}_0|=\cdots =|\tilde{\gamma}_m|=n+1$ is given.
    We define an instance $\Dmf$ with $\Dmf\models \text{grid}(0,0)$ which encodes the partial run.
    Thus we regard $(i,j)$ with $0\leq i \leq n$ and $0\leq j \leq m$ as constants and $\Dmf$
    contains the assertions
    \begin{eqnarray*}
      X((i,j),(i+1,j)), & & Y((i,j),(i,j+1)),\\
      T_{\sf init}(0,0), & &T_{\sf final}(n,m)
    \end{eqnarray*}
and $B(i,j)$ for $(i,j)\not\in \{(0,0),(n,m)\}$.  In addition, we
include in $\Dmf$ the atoms
$$
S((i,j),d_{i,j}^{1}), \quad S((i,j),d_{i,j}^{2})
$$
for distinct fresh constants $d_{i,j}^{1}$ and $d_{i,j}^{2}$ for all
$i,j$ such that $\tilde{\gamma}_{j}[i] = S$ and $S\not=\star$. It is now straightforward
to show that $\Dmf$ is consistent w.r.t.~$\Omc$ iff there is an accepting run of $M$ that
matches $\tilde{\gamma}$.

\bigskip

For the converse direction, we have to provide for every UCQ $q$ a polynomial reduction of the query
evaluation problem for $(\Omc,q)$ to the complement of the run fitting 
problem
for $M$. To this end observe that the following two conditions are equivalent for any UCQ $q(\vec{x})$,
instance $\Dmf$, and tuple $\vec{a}$:
\begin{enumerate}
\item $\Omc,\Dmf\models q(\vec{a})$
\item $\Dmf$ is not consistent w.r.t.~$\Omc$ or
$$
\{ \top
\sqsubseteq \exists Q.\top \mid Q\in \text{AUX}\},\Dmf \models
q(\vec{a})
$$
where $\text{AUX}= \text{AUX}_{\text{cell}} \cup
\text{AUX}_{\text{grid}} \cup \text{AUX}_{M}$.
\end{enumerate}
As the latter problem is in {\sc PTime} it suffices to provide a polynomial reduction of the problem
whether an instance $\Dmf$ is consistent w.r.t.~$\Omc$ to the run fitting 
problem for $M$.
Assume $\Dmf$ is given. First
decide in polynomial time whether $\Dmf$ is consistent w.r.t.~$\Omc_{\mathfrak{G}}$ (Lemma~\ref{lem:red3}).
If not, we are done. If yes, let $G$ be the set of all $d\in \text{dom}(\Dmf)$ such that $\Dmf\models
\text{grid}(d)$.  Assume without loss of generality that
$G=\{0,\ldots,k\}$. Then we find natural numbers $n_{i},m_{i}$ such
that each $i\in G$ is the root of an $n_{i}\times m_{i}$-grid with
witness function $\beta_{i}$ for $\mathfrak{G}$. By
Lemma~\ref{lem:red3}, there is a materialization $\Bmf$ of $\Dmf$ and
$\Omc_{\mathfrak{G}}$ such that $d\in G$ iff $d\in (=1
A)^{\mathfrak{B}}$. Next we check in polynomial time that $\Dmf$ is consistent
w.r.t.~the union of $\Omc_{\mathfrak{G}}$ and the sentences of
the form
$$
(=1 A) \sqcap (\geq 2 S^{X}) \equiv (=1A) \sqcap \exists X.(\geq 2 S),
$$
$$
(=1A) \sqcap (\geq 2 S^{XX}) \equiv (=1 A)\sqcap \exists X.(\geq 2 S^{X})
$$
in $\Omc_{M}$. If this is not the case we are done. If this is the case we assume that $\Dmf$ is saturated
in the sense that if, for example, $(=1A) \sqcap (\geq 2 S^{X}) \equiv
(=1A) \sqcap \exists X.(\geq 2 S)$ is in $\Omc_{M}$ then for any $d\in
(=1 A)^{\Bmf}$ and $X(d,d')\in \Dmf$ the following holds: $d$ has at
least two $S^{X}$-successors in $\Dmf$ iff $d'$ has at least two
$S$-successors in $\Dmf$. Now for each $i\in G$ we define the sequences
of strings
$$
\tilde{\gamma}^{i}=(\tilde{\gamma}_{0}^{i},\ldots,\tilde{\gamma}_{m_{i}}^{i})
$$
by setting for $0\leq r \leq n_{i}$ and $S\in \Gamma\cup Q$,
$$
\tilde{\gamma}_{r}^{i}[j] = S \mbox{ iff } \text{ $\beta_{i}(j,r)$ has
  at least two $S$-successors in $\Dmf$}
$$
If any of these strings is not a partial configuration, then $\Dmf$ is
not consistent w.r.t.~$\Omc$ and we are done.  Otherwise each
$\tilde{\gamma}^{i}$ is a partial run of $M$. It is now
straightforward to show that $\Dmf$ is consistent w.r.t.~$\Omc$
iff for each $i\in G$ there exists an accepting run of $M$ that
matches $\tilde{\gamma}^{i}$ which provides us with a polynomial reduction of
the consistency problem for instances $\Dmf$ to the run fitting 
problem for
$M$.
\end{proof}

\section{Proofs for Section 8}

\bigskip

We first give the proof of Lemma~\ref{lem:bouquetchar}.
Some notation is needed. An interpretation $\Dmf$ is \emph{$\Omc$-saturated} if
for all atoms $R(\vec{a})$ with $\vec{a}\subseteq \text{dom}(\Dmf)$
such that $\Omc,\Dmf\models R(\vec{a})$ it follows that $R(\vec{a})\in \Dmf$.
For every $\Omc$ and instance $\Dmf$ there exists a unique minimal (w.r.t.~set-inclusion)
$\Omc$-saturated instance $\Dmf_{\Omc}\supseteq \Dmf$. We call $\Dmf_{\Omc}$
\emph{the $\Omc$-saturation of $\Dmf$}.
The following lemma states basic properties of $\Omc$-saturated instances.
\begin{lemma}\label{lem:saturate}
Let $\Dmf \subseteq \Dmf'$ be instances with $\Dmf'_{|\text{dom}(\Dmf)}=\Dmf$ and let
$\Omc$ be a uGC$_{2}(=)$ ontology.
\begin{enumerate}
\item There exists a materialization of $\Omc$ and $\Dmf$ iff there exists a materialization
of $\Omc$ and the $\Omc$-saturation of $\Dmf$.
\item If $\Bmf$ is a materialization of $\Omc$ and $\Dmf$ and $\Dmf$ is $\Omc$-saturated,
then $\Bmf_{|\text{dom}(\Dmf)}= \Dmf$.
\item If $\Dmf'$ is $\Omc$-saturated, then $\Dmf$ is $\Omc$-saturated.
\end{enumerate}
\end{lemma}
We also require a lemma about stronger forms of unravelling tolerance and
forest models that one can employ for $\mathcal{ALCHIQ}$ ontologies. Call a model
$\Bmf$ of an instance $\Dmf$ an \emph{irreflexive forest model} if it is obtained from
$\Dmf$ by adding atoms of the form $R(a,b)$ with $a\not=b$ in $\text{dom}(\Dmf)$ to $\Dmf$
and hooking irreflexive tree interpretations $\Bmf_{a}$ to every $a\in \text{dom}(\Dmf)$.
The following variations of Lemma~\ref{lem:forestmodel} and Theorem~\ref{thm:dichotomy}
can be proved by modifying the proofs in a straightforward way taking into account the limited
expressive power of $\mathcal{ALCHIQ}$.
\begin{lemma}\label{lem:ALCQIH}
Let $\Omc$ be an $\mathcal{ALCHIQ}$ ontology of depth~1. Then the following hold:
\begin{enumerate}
\item Let $\Amf$ be a model of $\Omc$ and $\Dmf$.
Then there exists an irreflexive forest model $\Bmf$ of $\Omc$ and $\Dmf$ such that there
exists a homomorphism $h$ from $\Bmf$ to $\Amf$ that preserves $\text{dom}(\Dmf)$.
\item $\Omc$ is materializable
iff $\Omc$ is materializable for the class of all irreflexive tree instance
$\Dmf$ with $\text{dom}(\Dmf) \subseteq \text{sig}(\Omc)$.
\end{enumerate}
%

\end{lemma}
\begin{trivlist}\item
\textbf{Lemma~\ref{lem:bouquetchar} (restated)}~\itshape
Let $\Omc$ be a uGC$_{2}^{-}(1,=)$ ontology or a $\mathcal{ALCHIQ}$-ontology of depth~1. Then $\Omc$ is materializable iff $\Omc$ is
materializable for the class of all (irreflexive) bouquets $\Dmf$ of outdegree $\leq |\Omc|$ with
$\text{sig}(\Dmf)\subseteq \text{sig}(\Omc)$.
%
\end{trivlist}
\begin{proof}
We first prove the lemma for uGC$_{2}^{-}(1,=)$ ontologies and without the condition
on the outdegree. Let $\Sigma_{0}=\text{sig}(\Omc)$ and assume that $\Omc$ is
materializable for the class of all $\Sigma_{0}$-bouquets $\Dmf$.
By Theorem~\ref{thm:dichotomy} it is sufficient to prove that $\Omc$
is materializable for the class of $\Sigma_{0}$-tree instances. Fix a $\Sigma_{0}$-tree instance $\Dmf_{0}$.
To construct a materialization of $\Dmf_{0}$ we first compute the $\Omc$-saturation of $\Dmf_{0}$
by apply the following rules recursively and exhaustively to $\Dmf_{0}$:
\begin{itemize}
\item[$(\Fmf)$] For every injective homomorphism $h$ from some bouquet $\Fmf$ into the instance $\Dmf_{0}$
add to $\Dmf_{0}$ the set of all atoms $R(h(\vec{a}))$ such that $R(\vec{a})$ is in the $\Omc$-saturation of $\Fmf$.
\end{itemize}
Let $\Dmf$ be the resulting instance.
Recall that for a forest model materialization $\Bmf$ of an ontology $\Omc$ and instance $\Fmf$ that is consistent w.r.t.~$\Omc$
we have tree subinterpretations $\Bmf_{a}$, $a\in \text{dom}(\Fmf)$, of $\Bmf$ that are hooked to $\Fmf$ at $a$.
Take for any $a\in \text{dom}(\Dmf)$
the unique $\Omc$-saturated bouquet $\Fmf$ with root $a$ such that there is
an isomorphic embedding $h$ from $\Fmf$ into $\Dmf$ mapping $a$ to $a$ with range
$\text{dom}(\Dmf^{\leq 1}_{a})$ and hook to $\Dmf$ at $a$ the interpretation $\Bmf_{a}$ that
is hooked to $\Fmf$ at $a$ in a forest model materialization $\Bmf$ of $\Fmf$.
Denote by $\Amf$ the resulting interpretation.
Using the condition that $\Omc$ is a uGC$^{-}_{2}(1,=)$ ontology it is not difficult to prove
that $\Amf$ is a materialization of $\Omc$ and $\Dmf_{0}$.

We now prove the restriction on the outdegree. Assume $\Omc$ is given. Let $\Fmf$ be a bouquet with
root $a$ of minimal outdegree such that there is no materialization of
$\Omc$ and $\Fmf$. We show that the outdegree of $\Fmf$ does not exceed $|\Omc|$. Assume
the outdegree of $\Fmf$ is at least three (otherwise we are done). By Lemma~\ref{lem:saturate}
we may assume that $\Fmf$ is $\Omc$-saturated.
Take for any formula $\chi=\exists^{\geq n}z_{1} \alpha(z_{1},z_{2}) \wedge \varphi(z_{1},z_{2})$ that
occurs as a subformula in $\Omc$ the set $Z_{\chi}$ of all $b\not=a$ such that
$\Fmf\models \alpha(b,a)\wedge \varphi(b,a)$. Let $Z_{\chi}'=Z_{\chi}$ if $|Z_{\chi}|\leq n+1$;
otherwise let $Z'_{\chi}$ be a subset of $Z_{\chi}$ of cardinality $n+1$.
Let $\Fmf'$ be the restriction $\Fmf_{|Z}$ of $\Fmf$ to the union $Z$ of all $Z_{\chi}'$ and $\{a\}$.
We show that there exists no materialization of $\Fmf'$ and $\Omc$. Assume for a proof
by contradiction that there is a materialization $\Bmf$ of $\Fmf'$. Let $\Bmf'$ be
the union of $\Fmf\cup \Bmf$ and the interpretations $\Bmf_{b}$, $b\in \text{dom}(\Fmf)\setminus (Z\cup\{a\})$,
that are hooked to $\Fmf_{|\{a,b\}}$ at $b$ in a forest model materialization of $\Fmf_{|\{a,b\}}$.
We show that $\Bmf'$ is a materialization of $\Fmf$. Observe that by Lemma~\ref{lem:saturate}
the restriction $\Bmf'_{|\text{dom}(\Fmf)}$ of $\Bmf'$ to $\text{dom}(\Fmf)$ coincides with $\Fmf$.
Using the condition that $\Omc$ has depth~1 it is now easy to show that $\Bmf'$ is a model
of $\Omc$. It is a materialization since it is composed of materializations of subinstances.

The proof that irreflexive bouquets are sufficient for ontologies of depth~1 in $\mathcal{ALCHIQ}$ is similar
to the proof above but uses Lemma~\ref{lem:ALCQIH}.
\end{proof}

To characterize materializability in uGC$_{2}^{-}(1,=)$ we admit 1-materializability witnesses with loops and
require that they satisfy additional closure conditions which ensure that one can
construct a materialization in a step-by-step fashion, as in the proof of Lemma~\ref{lem:1materialization}.
To formulate the conditions we employ a `mosaic technique' and introduce mosaic pieces consisting of
a 1-materializability witness and a homomorphism into a bouquet.
Given a set of such mosaic pieces satisfying certain closure conditions we can then construct the materialization
of a bouquet step-by-step using the 1-materializability witnesses and, simultaneously, construct a homomorphism into
any other model of the bouquet (and thereby prove that we have indeed constructed a materialization).
The resulting procedure for checking materializability will be in {\sc NExpTime}
as we can first guess an exponential size set of mosaic pieces and then confirm in exponential time that
it satisfies our conditions. 
We start by defining the appropriate notion of 1-materializability witness: a \emph{1-materializability witness} is a
tuple $(\Fmf,a,\Bmf)$ such that
\begin{itemize}
\item $\Fmf$ is a bouquet with root
$a$ that is consistent w.r.t.~$\Omc$, of outdegree $\leq |\Omc|$ and such that $\text{sig}(\Dmf) \subseteq \text{sig}(\Omc)$;
\item $\Bmf$ is a 1-materialization of $\Fmf$ w.r.t.~$\Omc$ and a tree interpretation of outdegree $\leq 2|\Omc|$.
\end{itemize}
We require two different types of mosaic pieces. First, an \emph{injective hom-pair} takes the form
$(\Fmf,a,\Bmf)\rightarrow_{h}(\Fmf',a',\Bmf)$ such that
\begin{itemize}
\item $(\Fmf,a,\Bmf)$ and $(\Fmf',a',\Bmf')$ are 1-materializability witnesses;
\item $h$ is a homomorphism from $\Bmf$ to $\Bmf'$ and its restriction to $\text{dom}(\Fmf)$ is a bijective
homomorphism onto $\Fmf'$ mapping $a$ to $a'$.
\end{itemize}
Thus, in an injective hom-pair the first 1-materializability witness is a piece of the materialization we wish to construct and the second
1-materializability witness is a piece of the model into which we wish to homomorphically embed the materialization.
Injective hom-pairs assume the homomorphic embedding one wants to extend (i.e., the restriction of $h$ to $\text{dom}(\Fmf)$)
is injective. To deal with non-injective embeddings (as indicated by Example~\ref{loop}) we also consider \emph{contracting hom-pairs}
which take the form $(\Fmf,a,\Bmf)\rightarrow_{h}(\Bmf',a')$ where
\begin{itemize}
\item $(\Fmf,a,\Bmf)$ is a 1-materializability witness with $\text{dom}(\Fmf)=\{a,b\}$;
\item $(\Bmf',a')$ is a bouquet with root $a'$ of outdegree at most $2|\Omc|$ such that there exists a model
$\Amf$ of $\Omc$ with $\Amf^{\leq 1}_{a'}=\Bmf'$;
\item $h$ is a homomorphism from $\Bmf$ to $\Bmf'$ with $h(a)=h(b)=a'$.
\end{itemize}
Similarly to injective hom-pairs, in a contracting hom-pair the first 1-materializability witness is a piece of the
materialization we wish to construct and the second 1-materializability witness is a piece of the model into which
we wish to homomorphically embed the materialization.
\begin{table}[t] \small
  \centering
  \begin{boxedminipage}{\columnwidth}
\begin{enumerate}
\item If $(\Fmf,a,\Bmf)\in M$ and $\Fmf'$ is a bouquet with root $a'$ such that there
is a bijective homomorphism $h_{0}$ from $\Fmf$ to $\Fmf'$ mapping $a$ to $a'$, and $\Bmf'$
is an interpretation such that $(\Fmf',a',\Bmf')$ is a 1-materializability witness, then
there exists an extension $h$ of $h_{0}$ such that $(\Fmf,a,\Bmf)\rightarrow_{h}(\Fmf',a',\Bmf')\in H$.
\item If $(\Fmf,a,\Bmf)\rightarrow_{h}(\Fmf',a',\Bmf')\in H$ or $(\Fmf,a,\Bmf)\rightarrow_{h}(\Bmf',a') \in E$
with $b\in \text{dom}(\Bmf)\setminus \text{dom}(\Fmf)$ and $h(a)=h(b)=a'$, then there exist $h'$ and $\Bmf''$ such that
$(\Bmf_{|\{a,b\}},b,\Bmf'')\rightarrow_{h'}(\Bmf',a') \in E$.
\end{enumerate}
\end{boxedminipage}
  \caption{Conditions for $M$, $H$, and $E$ in materializability check for uGC$_{2}^{-}(1,=)$}
  \label{table:el}
\end{table}
The following lemma now provides a {\sc NexpTime} decision procedure for materializability
of uGC$_{2}^{-}(1,=)$ ontologies.
\begin{lemma}\label{thm:conugf}
Let $\Omc$ be a uGC$_{2}^{-}(1,=)$ ontology. Then $\Omc$ is materializable iff there exist
\begin{enumerate}
\item a set $M$ of 1-materializability witnesses containing exactly one 1-materializability witness
for every bouquet $\Fmf$ of outdegree $\leq |\Omc|$ with $\text{sig}(\Fmf)\subseteq \text{sig}(\Omc)$
that is consistent relative to $\Omc$;
\item sets $H$ of injective hom-pairs and $E$ of contracting hom-pairs whose first components are all in $M$
\end{enumerate}
such that the conditions of Table~\ref{table:el} hold.
\end{lemma}
\begin{proof}
Using Lemma~\ref{lem:canhom} one can show that sets $M$, $H$, and $E$ satisfying the conditions of
Table~\ref{table:el} exist if $\Omc$ is materializable. Conversely,
let the sets $M$, $H$, and $E$ satisfy the conditions given in Lemma~\ref{thm:conugf}.
Assume a bouquet $\Dmf$ of outdegree $\leq |\Omc|$ with $\text{sig}(\Dmf)\subseteq \text{sig}(\Omc)$ and root $a$ is given.
Take the 1-materializability witness $(\Dmf,a,\Bmf)\in M$.
As in the proof of Lemma~\ref{lem:1materialization} we construct a sequence of
interpretations $\Bmf^{0},\Bmf^{1},\ldots$ and sets $F_{i}\subseteq \text{dom}(\Bmf^{i})$
of frontier elements (but now using only 1-materializability witnesses in $M$):
set $\Bmf^{0}:=\Bmf$ and $F_{0}=\text{dom}(\Bmf)\setminus
\text{dom}(\Dmf)$. If $\Bmf^{i}$ and $F_{i}$ have been constructed, then take for any $b\in F_{i}$
its (unique) predecessor $a$ and a 1-materializability witness $(\Bmf^{i}_{|\{a,b\}},b,\Bmf_{b})\in M$ and set
$\Bmf_{i+1}:= \Bmf^{i}\cup \bigcup_{b\in F_{i}}\Bmf_{b}$.
Let $\Bmf^{\ast}$ be the union of all $\Bmf_{i}$. We show that $\Bmf$ is a materialization.
$\Bmf$ is a model of $\Omc$ by construction since $\Omc$ is a uGC$^{-}_{2}(1,=)$ ontology.
Consider a model $\Amf$ of $\Omc$ and $\Dmf$.
We construct a homomorphism $h$ from $\Bmf^{\ast}$ to $\Amf$ preserving $\text{dom}(\Dmf)$
as the limit of a sequence $h_{0},\ldots$ of homomorphisms from $\Bmf^{i}$ to $\Amf$. We may assume
that the outdegree of $\Amf$ is $\leq 2|\Omc|$.
By definition, there exists a homomorphism from $\Bmf^{0}$ to
$\Amf_{a}^{\leq 1}$ preserving $\Dmf$. Now, inductively, we ensure in each step that the homomorphisms $h_{i}$ satisfy the
following conditions for all $\Bmf^{i}$ and $F_{i}$, all $b\in F_{i}$ and the
predecessor $a$ of $b$ in $\Bmf^{i-1}$:
\begin{itemize}
\item[(a)] Assume $h_{i}(a)\not=h_{i}(b)$. Then, for the 1-materializability witness $(\Bmf^{\ast}_{|\{a,b\}},b,\Bmf_{b})\in M$
there exists a homomorphism $h_{b}$ such that
$$
(\Bmf^{\ast}_{|\{a,b\}},b,\Bmf_{b})\rightarrow_{h_{b}}(\Amf_{|\{h_{i}(a),h_{i}(b)\}},h_{i}(b),\Amf^{\leq 1}_{h_{i}(b)})\in H
$$
\item[(b)] Assume $h_{i}(a)=h_{i}(b)$. Then, for the 1-materializability witness $(\Bmf^{\ast}_{|\{a,b\}},b,\Bmf_{b})\in M$
there exists a homomorphism $h_{b}$ such that
$$
(\Bmf^{\ast}_{|\{a,b\}},b,\Bmf_{b})\rightarrow_{h_{b}}(\Amf^{\leq 1}_{h_{i}(b)},h_{i}(b))\in E
$$
\end{itemize}
Assume $h_{i}$ with the properties (a) and (b) has been constructed. Then we take for all $b\in F_{i}$ and the
predecessor $a$ of $b$ the homomorphism $h_{b}$ determined by (a) and, respectively, (b) and set
$$
h_{i+1}:= h_{i} \cup \bigcup_{b\in F_{i}}h_{b}
$$
Using the properties of $M$, $H$, and $E$ given in Table~\ref{table:el} it is not difficult to
show that $h_{i+1}$ again has the properties (a) and (b) above.
\end{proof}

We prove
Theorem~\ref{thm:complexitydepthtwo}, that is, for $\mathcal{ALC}$
ontologies of depth~2, deciding whether query-evaluation is in {\sc
  Ptime} is {\sc NExpTime}-hard (unless
$\text{\sc Ptime}=\text{\sc coNP}$).  The proof is by reduction of the
complement of a \NExpTime-complete tiling problem in which the aim is
to tile a $2^n\times 2^n$-grid. An instance of this problem is given
by a tuple $P=(\Tmc,t_0,H,V)$ where $\Tmc$ is a set of \emph{tile
  types}, $t_0 \in \Tmc$ is a distinguished tile to be placed on
position $(0,0)$ of the grid, and $H$ and $V$ are horizontal and
vertical matching conditions.  A \emph{solution} to $P$ is a function
$\tau: 2^n \times 2^n \to \Tmc$ such that
\begin{itemize}

\item if $\tau(i,j) = t$ and $\tau(i+1,j)=t'$ then $(t,t')\in H$, for all
$i<2^n-1$, $j<2^n$,

\item if $\tau(i,j) =t$ and $\tau(i,j+1)=t'$ then $(t,t')\in V$, for all
$i<2^n$, $j<2^n-1$,

\item $\tau(0,0)=t_{0}$.

\end{itemize}
Let $P=(\Tmc,t_0,H,V)$. We construct an \ALC ontology $\Omc$ of depth~2 such that
$P$ has a solution iff $\Omc$ is not materializable. 





%

\medskip

The idea is that \Omc verifies the existence of an $r$-chain in the
input instance whose elements represent the grid positions along with
a tiling, row by row from left to right, starting at the lower left
corner and ending at the upper right corner. The positions in the grid
are represented in binary by the concept names $X_1,\dots,X_{2n}$ in
the instance where $X_1,\dots,X_n$ indicate the horizontal position
and $X_{n+1},\dots,X_{2n}$ the vertical position. The tiling is
represented by unary relations $T_t$, $t \in \Tmc$. \Omc verifies the
chain by propagating a marker bottom up and while doing this, it
verifies the horizontal matching condition. When the top of the chain
is reached, \Omc generates another $r$-chain using existential
quantifiers. On that chain, a violation of the vertical matching
condition is guessed using disjunction, that is, the position where
the violation occurs and the tiles involved in it. If the tiling on
the first chain is defective, the guesses can be made such that the
second chain homomorphically maps into the first one, and then the
second chain is `invisible' to the query. Otherwise, the query can
`see' the second chain and because of the disjunctions involved in
building it, the instance has no materialization w.r.t.\
\Omc. Clearly, the second case can occur only if $P$ has a solution.

\smallskip

We now assemble the ontology $\Omc$. To avoid that \Omc is trivially
not materializable, we have to hide the marker propagated up the chain
in the input and all markers used on the existentially generated
chain. We thus use concepts of the form $\forall s . A$ instead of
concept names $A$, additionally stating in \Omc that
$\top \sqsubseteq \exists s . A$. For any concept name~$A$, we use
$H_A$ as an abbreviation of $\forall s. A$.
\begin{enumerate}

\item Set up hiding of concept names: for all $A \in \{ V,
  \mn{ok}_1,\dots,\mn{ok}_{2n}, D, L,
  Y_1,\dots,Y_{2n},
  \overline{Y}_1,\dots,\overline{Y}_{2n},$ $Z_1,\dots,Z_n,$
  $\overline{Z}_1,\dots,\overline{Z}_n\} \cup \{S_t \mid t \in \Tmc \}$,
$$\top \sqsubseteq \exists s . A$$

\item The initial position starts the propagation: for all $t \in \Tmc$,
$$
  \begin{array}{r@{\;}c@{\;}l}
\overline{X}_1 \sqcap \cdots \sqcap \overline{X}_{2n}
  \sqcap T_t 
  &\sqsubseteq& H_V
\end{array}
  $$
  %

\item Tiles are mutually exclusive: for all distinct $t,t' \in \Tmc$:
  $$
  T_t \sqcap T_{t'} \sqsubseteq \bot
  $$

\item The propagation proceeds upwards, checking the horizontal
  matching condition:
  $$
  \begin{array}{r@{\;}c@{\;}l}
      X_i \sqcap \exists r . X_i \sqcap \bigsqcup_{1 \leq j < i} \exists r
  . X_j &\sqsubseteq& H_{\mn{ok}_i}  \\[1mm]
  \overline{X}_i \sqcap \exists r . \overline{X}_i \sqcap \bigsqcup_{1 \leq j < i} \exists r
  . X_j &\sqsubseteq& H_{\mn{ok}_i}  \\[1mm]
  X_i \sqcap \exists r . \overline{X}_i \sqcap \bigsqcap_{1 \leq j < i} \exists r
  . \overline{X}_j &\sqsubseteq& H_{\mn{ok}_i}  \\[1mm]
  \overline{X}_i \sqcap \exists r . X_i \sqcap \bigsqcap_{1 \leq j < i} \exists r
  . \overline{X}_j &\sqsubseteq& H_{\mn{ok}_i}  \\[1mm]
  H_{\mn{ok}_1} \sqcap  \cdots \sqcap H_{\mn{ok}_{n}}
  \sqcap \overline{X}_i \, \sqcap \\[1mm]
  \exists r . (H_V
  \sqcap T_t) \sqcap T_{t'} 
&\sqsubseteq& H_V \\[1mm]
  \exists r . X_i \sqcap \exists r . \overline{X}_i &\sqsubseteq& \bot
  \end{array}
  $$
  where $i$ ranges over 1..$2n$ and $(t,t')$ over $H$; the use of
  $\overline{X}_i$ in the third last line prevents the counter from
  wrapping around at maximum value. The last inclusion is necessary to
  avoid that multiple successors that carry different concept name
  labels interact in undesired ways.

\item When the maximum value is reached, we make an extra step in the
  instance and then generate the first object on an existential chain that
  represents a tiling defect:
  $$
  \begin{array}{r@{\;}c@{\;}l}
    \exists r . (X_1 \sqcap \cdots \sqcap X_{2n}
    \sqcap H_V)
    \sqsubseteq \exists r . C
  \end{array}
  $$
  where
  $$
  \begin{array}{r@{\;}c@{\;}l}
    C &=& H_M \sqcap H_{Y_1} \sqcap \cdots \sqcap H_{Y_{2n}}
  \end{array}
  $$

\item We continue building the chain; in every step, we decide
  whether we have a defect here ($H_D$) or defer it to later ($H_L$);
  we can defer at most to the left-most position of the second row:
  $$
  \begin{array}{r@{\;}c@{\;}l}
    H_M &\sqsubseteq& (H_D \sqcap C_D) \sqcup H_L \\[1mm]
    H_M \sqcap H_{\overline{Y}_1} \sqcap \cdots \sqcap H_{\overline{Y}_{n}} \, \sqcap
    \\[1mm]
    H_{Y_{n+1}} \sqcap
    H_{\overline{Y}_{n+2}} \sqcap \cdots \sqcap
    H_{\overline{Y}_{2n}} &\sqsubseteq& H_D
  \end{array}
  $$
  where $C_D$ is a concept to be defined later.

\item When defering the defect to later, we keep going and maintain
  the $Y$-counter, decrementing it:
  $$
  \begin{array}{r@{\;}c@{\;}l}
    H_L &\sqsubseteq& \exists r . H_M \\[1mm]
    H_L \sqcap \bigsqcap_{1 \leq j \leq i} H_{\overline{Y}_i} & \sqsubseteq&
    \forall r . H_{Y_i} \\[1mm]
    H_L \sqcap H_{Y_i} \sqcap \bigsqcap_{1 \leq j < i} H_{\overline{Y}_i} & \sqsubseteq&
    \forall r . H_{\overline{Y}_i} \\[1mm]
    H_L \sqcap H_{Y_i} \sqcap \bigsqcup_{1 \leq j < i} H_{Y_i} & \sqsubseteq&
    \forall r . H_{Y_i} \\[1mm]
    H_L \sqcap H_{\overline{Y}_i} \sqcap \bigsqcup_{1 \leq j < i} H_{Y_i} & \sqsubseteq&
    \forall r . H_{\overline{Y}_i}
  \end{array}
  $$
  where $i$ ranges over $1..2n$.

\item When implementing the defect, we guess two $V$-incompatible
  tiles, start a new counter, travel exactly $2^n$ steps, and verify
  the tiles. The above concept $C_D$ is
  $$
  C_D = H_{Z_1} \sqcap \cdots \sqcap H_{Z_n}
  \sqcap \bigsqcup_{(t,t') \notin V} (T_t \sqcap H_{S_{t'}})
  $$
  and we add the following concept inclusions:
  $$
  \begin{array}{r@{\;}c@{\;}l}
    H_D \sqcap H_{Z_i} & \sqsubseteq & \exists r . (H_M
    \sqcap H_D) \\[1mm]
    H_D \sqcap \bigsqcap_{1 \leq j \leq i} H_{\overline{Z}_i} & \sqsubseteq&
    \forall r . H_{Z_i} \\[1mm]
    H_D \sqcap H_{Z_i} \sqcap \bigsqcap_{1 \leq j < i} H_{\overline{Z}_i} & \sqsubseteq&
    \forall r . H_{\overline{Z}_i} \\[1mm]
    H_D \sqcap H_{Z_i} \sqcap \bigsqcup_{1 \leq j < i} H_{Z_i} & \sqsubseteq&
    \forall r . H_{Z_i} \\[1mm]
    H_D \sqcap H_{\overline{Z}_i} \sqcap \bigsqcup_{1 \leq j < i} H_{Z_i} & \sqsubseteq&
    \forall r . H_{\overline{Z}_i} \\[1mm]
    H_D \sqcap H_{Z_i} \sqcap H_{S_t} & \sqsubseteq & \forall r
    . H_{S_t} \\[1mm]
    H_D \sqcap H_{\overline{Z}_1}  \sqcap \cdots \sqcap
    H_{\overline{Z}_n} \sqcap H_{S_t}
    & \sqsubseteq & T_t
  \end{array}
  $$
  where $i$ ranges over $1..n$.

\end{enumerate}
We now establish correctness of the reduction. A \emph{tiling chain}
is an instance $\Dmf$ that consists of the following assertions:
\begin{itemize}

\item $r(a_i,a_{i+1})$ for $0 \leq i < 2^{2n}-1$

\item $X_j(a_i)$ whenever the $j$-th bit of $i$ is one

\item $\overline{X}_j(a_i)$ whenever the $j$-th bit of $i$ is zero

\item exactly one $T_t(a_i)$ for each $i$, $t \in \Tmc$, such that
  when $T_t(a_i),T_{t'}(a_{i+1}) \in \Dmf$, then $(t,t') \in H$.

\end{itemize}
We say that the tiling chain $\Dmf$ is \emph{defective} if there is an
$i \leq 2^{2n}-(2^n+1)$ such that
$T_t(a_i),T_{t'}(a_{i+2^n}) \in \Dmf$ and $(t,t') \notin V$.

  \medskip

\begin{lemma}
  $P$ has a solution iff $\Omc$ is not materializable.
\end{lemma}
\begin{proof}
  (sketch) ``if''. Assume that $P$ has no solution. Take an instance $\Dmf$.
	We have to construct a CQ-materialization $\Amf$ of $\Dmf$ and $\Omc$. To do this, start with $\Dmf$ viewed as an
  interpretation~\Amf. To satisfy the concept inclusions in Points~1-4, which all
  fall within monadic Datalog, apply a standard chase procedure. To
  satisfy the concept inclusions in Point~5 to~8, consider every
  $a \in (\exists r . (X_1 \sqcap \cdots \sqcap X_{2n} \sqcap
  H_V))^\Imc$.
  The concept $H_V= \forall s . V$ cannot be made true by an atom in an instance and, consequently, it was made true at $a$ by the chase.
  Analyzing the concept inclusions in $\Omc$, one can show that there must thus be a
  tiling chain $\Cmc \subseteq \Dmf$ whose last element $a_{2^{2n}-1}$
  is $a$. Since there is no solution, \Cmc must be defective. When
  generating the existential chain, we can thus make the guesses in a
  way such that the chain homomorphically maps into \Cmc, thus into
  \Dmf.  It can be verified that this yields the desired
  CQ-materialization of $\Omc$ and $\Dmf$.

  \smallskip

  ``only if''.  Assume that $P$ has a solution. Let \Dmf be the tiling
  chain that represents it. We show that there is no materialization
  of $\Omc$ and $\Dmf$. In fact, \Omc propagates the $H_V$ marker all the
  way up to the last element $a=a_{2^{2n}-1}$ of $\Dmf$, where then an
  existential chain is generated. Because of the use of disjunctions
  to generate all different kinds of violations of the vertical tiling
  conditions, there are clearly at least two chains that can be
  generated and that are incomparable in terms of
  homomorphisms. Moreover, since $\Dmf$ represents a proper tiling, none
  of the chains homomorphically maps to $\Dmf$. Consequently, we can
  find CQs $q_1(x),\dots,q_m(x)$ such that
  $\Omc,\Dmf \models q_1(a) \vee \cdots \vee q_m(a)$, but
  $\Omc,\Dmf\not\models q_i(a)$ for any $i$. Thus, \Omc does not
  have the disjunction property and consequently is not materializable.
\end{proof}

\end{document}